\documentclass[11pt]{article}
\usepackage{amsfonts,amssymb,amsmath,amsthm,amscd,latexsym}
\usepackage{graphics,graphicx,color}
\usepackage[latin1]{inputenc}

\newcommand{\NN}{\mathbb{N}}

\newcommand{\RR}{\mathbb{R}}

\newcommand{\oo}[0]{\otimes}

\newcommand{\inv} {{-1}}

\newcommand{\g}{\mathfrak{g}}
\newcommand{\su}{\mathfrak{su}}
\newcommand{\sll}{\mathfrak{sl}}
\newcommand{\SU}{\mathrm{SU}}
\newcommand{\SO}{\mathrm{SO}}
\newcommand{\SL}{\mathrm{SL}}
\newcommand{\Spin}{\mathrm{Spin}}
\newcommand{\so}{\mathfrak{so}}
\newcommand{\spin}{\mathfrak{spin}}

\newcommand{\Hom}{\mathrm{Hom}}
\newcommand{\End}{\mathrm{End} \,}
\newcommand{\C}{\mathbb{C}}

\newcommand{\beq}{\begin{equation}}
\newcommand{\eeq}{\end{equation}}
\newcommand{\beqa}{\begin{eqnarray}}
\newcommand{\eeqa}{\end{eqnarray}}
\newcommand{\nn}{\nonumber}
\newcommand{\uqr}{U_q^{(res)}(\mathfrak{su}(2))}
\newcommand {\cres}{\mathcal C_{res}}

\newtheorem{theorem}{Theorem}[section]
\newtheorem{lemma}[theorem]{Lemma}

\newtheorem{proposition}[theorem]{Proposition}

\newtheorem{definition}[theorem]{Definition}

\newcommand{\ros}[0]{}

\begin{document}
\parskip 6pt
\parindent 0pt

\title{\Large\bf Quantum deformation of two four-dimensional spin foam models}

\author{Winston J. Fairbairn\footnote{winston.fairbairn@math.uni-erlangen.de}\ \,, \, Catherine Meusburger\footnote{catherine.meusburger@math.uni-erlangen.de}
\\ [4mm]
\itshape{\normalsize{Department Mathematik, Friedrich-Alexander-Universit\"at Erlangen-N\"urnberg,}} \\
\itshape{\normalsize{Cauerstra\ss e 11,  91058 Erlangen,}} \\
\itshape{\normalsize{Germany.}}
}

\date{February 5 2012}
\maketitle

\begin{abstract}
We construct the q-deformed version of two four-dimensional spin foam models, the Euclidean and Lorentzian versions of the EPRL model. The q-deformed  models 
are based on the representation theory of two copies of $U_q(\su(2))$ at a root of unity and on the quantum Lorentz group with a real deformation parameter. For both models we give a definition of the quantum EPRL intertwiners, study their convergence and braiding properties and construct an amplitude for the four-simplexes. We find that both of the resulting models are convergent.
\end{abstract}

\section{Introduction}

\subsubsection*{Background and motivation}
Manifold invariants of representation theoretic origin such as the Reshetikhin-Turaev invariant \cite{RTuraev} or the Turaev-Viro invariant \cite{Turaev}  play an important role in mathematical physics. In particular,  they  are of interest to  quantum gravity, where models of the Turaev-Viro type are known under the name of spin foam models. The basic principle in the construction of a $n$ dimensional spin foam model is to assign irreducible representations to the $(n-2)$-simplexes of a triangulated manifold and intertwining operators to the $(n-1)$-simplexes which intertwine the representations associated to their boundaries.
Higher-dimensional simplexes in the triangulation are then decorated with representation theoretical data constructed 
from these representations and intertwiners,  often referred to as ``amplitudes'' in the literature. After a summation over all such assignments of representations and intertwiners the product of the amplitudes assigned to each higher-dimensional simplex define a number, the partition function of the model. In some three- and higher-dimensional models, one can show that the partition function converges and is independent of the choice of the triangulation. It thus defines a manifold invariant.

Among the spin foam models in the literature,  one distinguishes models which are based on the representation theory of Lie groups (in the following referred to as ``classical models'') and models which are based on the representation theory of quantum groups. 
Compared to the classical models,  the latter  have several advantages. The most important one is that they  exhibit an improved convergence behaviour. When the relevant quantum group is a $q$-deformed universal enveloping algebra of a simple Lie algebra at a root of unity, its tilting modules form a non-degenerate finite semisimple spherical category \cite{tilt1,tilt2}. In the corresponding spin foam model, this has the effect of replacing infinite sums in the  classical models  by finite sums in the $q$-deformed models. 

 In the case where $q$ is not a root of unity, the representation theory of the $q$-deformed universal enveloping algebras is more complicated, but the associated spin foam models still often show an improved convergence behaviour compared to the classical ones. From the physics perspective, the $q$-deformed models can be viewed as regularisations of the corresponding classical models. They  remove divergences that occur when the representation labels grow large. Since these labels carry an interpretation in terms of lengths or areas of the simplicial complex, these divergences are often referred to as infrared.

Examples of this effect in three dimensions are the Turaev-Viro model \cite{Turaev}, which is based on the tilting modules for $U_q(\su(2))$ at a root of unity. It  can be viewed as a regularisation of the (divergent) Ponzano-Regge model \cite{ponzreg}, which is based on the representation theory of $\SU(2)$. A similar example in four dimensions is  the relation between the Crane-Yetter model \cite{yetter2,yetter3} and the Ooguri model \cite{ooguri}, which are based on the representation theory of,  respectively, $U_q(\su(2))$ at a root of unity and $\SU(2)$. The former defines a topological invariant of $4$-manifolds which can be expressed in terms of the signature and the Euler characteristic of the manifold \cite{Roberts}. 

Quantum groups also appear in constrained topological models of four-dimensional quantum gravity.  The idea underlying the construction of such models is the fact that gravity in dimensions  greater than three can be written as a constrained BF theory\footnote{The name BF theory stems from the letters $B$ and $F$ that were used for the field variables in the first formulation of the theory.} \cite{fd}.  Constrained topological models are an attempt of incorporating these constraints into a spin foam model by imposing appropriate restrictions on the 
representations and intertwiners that are assigned to the faces and tetrahedra of the triangulation.

The first such constrained topological model is the BC model due to Barrett and Crane \cite{BC,BC2}.  Similarly to the topological models, this model diverges for large values of the representation labels. Yetter \cite{yet4} and Noui and Roche \cite{PK} constructed  $q$-deformed analogues of the Euclidean and Lorentzian versions of the model based on, respectively,  the representation category of $U_q(\so(4))$ and $U_q(\so(3,1))$. Both $q$-deformed models exhibit enhanced convergence properties compared to their classical counterparts.

Recently, the implementation of the constraints reducing BF theory to gravity in 4d spin foam models was refined and improved. This led to the formulation of new constrained topological models by  Engle, Pereira, Rovelli and Livine (EPRL) \cite{eprlmod}  and Freidel and Krasnov (FK) \cite{Fkmodel}. These models incorporate a free parameter $\gamma$, called the Immirzi parameter. The limits $\gamma \rightarrow \infty$ and $\gamma \rightarrow 0$ yield respectively the BC and the flipped EPR model \cite{epr,epr2,epr3}.

These recent models exhibit interesting physical properties, such as an asymptotic behaviour related to the Regge action \cite{not1,not2, notts3, notts4, fc1,fc2}, but suffer from divergences for large area variables (see \cite{CS} for an analysis of these divergences). As originally suggested by Rovelli \cite{Carlo}, this indicates the need for $q$-deformed versions of these models, which still have the relevant physical properties but  have a better chance of convergence. Furthermore, there is evidence that such models could be related to gravity with a positive cosmological constant. The construction of these models is the aim of this article.

\subsubsection*{Main results}

In this article, we construct a $q$-deformed version of both the Euclidean  and Lorentzian EPRL model \cite{eprlmod}. 
 The formulation of the  Lorentzian model uses techniques from the representation theory of non-compact quantum groups  and applies them to the quantum Lorentz group with a real deformation parameter.  Using the results by Buffenoir and Roche \cite{PE,PE2}  on the harmonic analysis of the quantum Lorentz group, we construct a generalisation of the integral definition of the classical EPRL intertwiner to  the quantum Lorentz group by means of a Haar measure on the latter. This is followed by a detailed analysis  of the properties of the resulting intertwiner, including a proof of its convergence.  We also derive explicit expressions for its transformations under braiding, under which the quantum EPRL intertwiner turns out not to be invariant. We then define an amplitude for the $4$-simplexes of a triangulated four-manifold and construct the associated  quantum spin foam model. Remarkably, the model converges although it is based on representations of an infinite-dimensional Hopf algebra.

The Euclidean model  is formulated in terms of representation categories  rather than the language of  Hopf algebras. 
This structural difference with the Lorentzian model is explained by the fact that the Euclidean model is based on the representation theory of  (two copies of) $\uqr$ at a root of unity. While the representation theory simplifies considerably in this case -  the tilting modules \cite{tilt1,tilt2} of $\uqr$ at a root of unity define a non-degenerate, finite semisimple spherical category-  the corresponding Hopf algebra structures become complicated, and one is forced to enter the framework of 
 weak, quasi-Hopf algebras. The $q$-deformed generalisation of  the Euclidean EPRL intertwiner is therefore formulated in terms of the representation theory of $\uqr\otimes\uqr$ rather than in terms of Haar measures on Hopf algebras.
In this case, there are no convergence issues since there is only a finite number of representations with non-vanishing $q$-dimension. As in the Lorentzian case, the  quantum EPRL intertwiner does not appear to be invariant under braiding. After our discussion of the EPRL intertwiner, we then define the amplitude for a four simplex in the resulting model and show that it factorises into two quantum $15j$ symbols with fusion coefficients. This mirrors the behaviour of the  amplitude in the classical Euclidean EPRL model \cite{eprlmod}.
The resulting  $q$-deformed spin foam model is  again finite.

\subsection*{Outline of the paper}

The paper is organised as follows. Section $2$ is dedicated to the construction of the Lorentzian model. We start by reviewing essential notions on the quantum Lorentz group, its representation theory and  Harmonic analysis in Section \ref{qlorg}.  In Section  \ref{qeprl}  we show how the Lorentzian EPRL intertwiner can be generalised to the quantum Lorentz group  by means of a Haar measure on the quantum Lorentz group. We then prove an important convergence theorem that ensures that the quantum EPRL intertwiner is well-defined and investigate its behaviour of  under braiding. In Section \ref{qamp} we define an amplitude for the four-simplexes of a closed oriented triangulated four-manifold. This definition is given via a graphical calculus defined by means of an  invariant bilinear form on the representation spaces of the quantum Lorentz group. This naturally leads to the definition of the quantum spin foam model which is shown to converge for all triangulated manifolds. 

In Section $3$, we construct the $q$-deformed EPRL model for Euclidean signature. We start by summarising the relevant aspects of the quantum group $\uqr$ at a root of unity and of its representation theory in Section \ref{qrot}. In Section \ref{qso4} we give a brief description of the ``quantum rotation group'' $\uqr\otimes\uqr$ and its irreducible representations. 
This background is then used in the construction of the  Euclidean quantum  EPRL intertwiner in Section \ref{euceprl}. We analyse the properties of this intertwiner and then define the associated four-simplex amplitude in Section \ref{eucamp}. We conclude the discussion by proving  that this amplitude factorises into a quantum $15j$ symbol and that the associated  spin foam model  converges for all triangulated closed four-manifolds.
Section $4$ contains a discussion of the physical interpretation and properties of the two $q$-deformed EPRL models  as well as our  outlook and conclusions.

\section{The Lorentzian model}

\subsection{The quantum Lorentz group}
\label{qlorg}

In this subsection, we summarise the relevant definitions and results about the quantum Lorentz group following \cite{worpod, PE,PE2}. The quantum Lorentz group is an infinite-dimensional  ribbon Hopf algebra, which is obtained as the quantum double or Drinfel'd double of the $q$-deformed universal enveloping algebra $U_q(\su(2))$, where  
$q = e^{-\kappa} \in \, ] 0,1[$ is a {\em real} deformation parameter. 

\subsubsection{Hopf algebra structure}


\label{halgsec}

\paragraph{The Hopf algebra $U_q(\su(2))$.}

We start by introducing the Hopf algebra $U_q(\su(2))$, adopting the conventions from \cite{PE}.  The star Hopf algebra $U_q(\su(2))$ is the associative algebra generated multiplicatively by the elements
$q^{\pm J_z}$, $J_\pm$,  subject to the relations
\begin{align}\label{rels1}
q^{\pm J_z}q^{\mp J_z}=1, \qquad q^{J_z}J_\pm
q^{-J_z}=q^{\pm1} J_\pm, \qquad
[J_+,J_-]=\frac{q^{2J_z}-q^{-2J_z}}{q-q^\inv}.
\end{align}
The comultiplication, counit and antipode are given by
\begin{align}
\label{uqcomult2} &\Delta(q^{\pm J_z})=q^{\pm J_z}\oo q^{\pm
J_z}, &  &\Delta(J_\pm)=q^{-J_z}\oo J_\pm+J_\pm\oo q^{J_z},\\
\label{uqcounit2} &\epsilon(q^{\pm J_z})=1, &  &\epsilon(J_\pm)=0,\\
\label{uqantip2} &S(q^{\pm J_z})=q^{\mp J_z}, &   &S(J_\pm)=-q^{\pm
1} J_\pm,
\end{align}
and the star structure takes the form
\beq\label{lorstar}
\star q^{J_z} = q^{J_z}, \;\;\;\; \star J_{\pm} = q^{\mp 1} J_{\mp}.
\eeq
The star Hopf algebra $U_q(\su(2))$ is a ribbon algebra with universal $R$-matrix
\beq\label{lorr}
R = q^{2 J_z \otimes J_z} e_{q^{-1}}^{\left( (q - q^{-1}) q^{J_z} J_+ \otimes J_- q^{-J_z} \right)},
\eeq
where $e_{q^\inv}^{(z)}$ is the $q$-exponential:
\begin{align}\label{qexp}
e^{(z)}_{q^{-1}}=\sum_{k=0}^\infty  q^{\frac {k(k-1)} 2}\frac{z^k}{[k]!}, \qquad\text{where}\quad [k]=\frac{q^k-q^{-k}}{q-q^\inv}, \quad\text{and}\quad [k]!=[k][k-1]\cdots[2][1].
\end{align}
The ribbon element $v$ is defined by the identity
$v^2=uS(u)$, where $u$ is the invertible element $u = m((S \otimes 1)R')$ and 
 $m: U_q(\su(2)) \otimes U_q(\su(2))\rightarrow U_q(\su(2))$ denotes the multiplication map.
 The group-like element $\mu=u v^\inv$ is given by
\beq\label{mudef}
\mu = q^{2 J_z}.
\eeq

\paragraph{Representation theory of $U_q(\su(2))$} The representation theory of  $U_q(\su(2))$
with a  deformation parameter $q\in ]0,1[$ closely resembles the representation theory of the Lie group $\SU(2)$. 
Irreducible finite-dimensional $\star$-representations are labelled by ``spins'' $I\in\NN_0/2$ and by an additional  parameter  $\omega\in \{\pm 1, \pm i\}$. In the following we will restrict attention to  equivalence classes of unitarizable representations which are characterised by the condition $\omega=1$. 
As in the case of the Lie group $\SU(2)$,   the representation space $V_I$ of 
the  irreducible representation 
 $ \pi_{I}: U_q(\su(2)) \rightarrow\text{End}(V_I)$  is $(2I+1)$-dimensional. There exists an orthonormal basis $\{e^I_m\}_{m=-I,...,I}$  of the complex vector space $V_I$ in which the generators act according to
\begin{align}
\label{reps}
&\pi_{ I}(q^{J_z}) \, e^I_m=  q^{\pm m} \, e^I_m. \nonumber \\
&\pi_{I}(J_{\pm}) \, e^I_m= q^{\mp 1/2} \sqrt{[I \mp m] [I \pm m+1]} \, e^I_{m \pm 1}, 
\end{align}
where $[n]$ is defined as in \eqref{qexp}.  The fusion rules for the tensor products $V_I \otimes V_J$ resemble the ones for the representations of  $\SU(2)$. We have
\beq
\label{fusionunity}
V_{I} \otimes V_{J} \cong  \bigoplus_{K = | I - J |}^{I + J} V_K,
\eeq
where the isomorphism $\cong$ is given by the Clebsch-Gordan intertwining operators
\beq
C^{K}_{\; IJ}: V_I \otimes V_J \rightarrow V_K, \hspace{3mm} \mbox{and} \hspace{3mm} C^{IJ}_{\;\, K} : V_K \rightarrow V_I \otimes V_J.
\eeq
As all multiplicities in \eqref{fusionunity} are equal to one, 
these
 intertwiners are unique up to normalisation. 
They  are non-zero if and only if  $I + J - K$, $J+K-I$ and $K+I-J$ are non-negative integers.
Their coefficients with respect to the bases $\{e^I_m\}_{m=-I,...,I}$ are the  Clebsch-Gordan coefficients \ros
\beq
C^{IJ}_{\;\, K}(e^K_{m}) = \sum_{n,p} \left( \begin{array}{cc} n & p \\
                          I & J  \end{array} \right| \left. \begin{array}{c} K \\ m \end{array} \right)  e^I_{n} \otimes e^J_{p}, \hspace{3mm} \mbox{and} \hspace{3mm}
C^{K}_{\; IJ}( e^I_{n} \otimes e^J_{p}) = \sum_{m} \left( \begin{array}{c} m \\ K
 \end{array} \right| \left. \begin{array}{cc} I & J \\ n & p \end{array} \right)  e^K_{m}.
\eeq
To fix the phase of the Clebsch-Gordan coefficients, we impose reality conditions
\beq
\label{real}
\left( \begin{array}{cc} m & m \\
                          I & J  \end{array} \right| \left. \begin{array}{c} K \\ p \end{array} \right) = \left( \begin{array}{c} p \\ K
 \end{array} \right| \left. \begin{array}{cc} I & J \\ m & n \end{array} \right) \in \mathbb{R},
\eeq
and we use Wigner's convention to fix the remaining sign ambiguity.
The Clebsch-Gordan coefficients satisfy numerous relations. In the following we will frequently use  their first orthogonality property  and permutation symmetry \ros
\beq
\label{orthogonality}
\sum_{m,n} \left( \begin{array}{cc} m & n \\
                          I & J  \end{array} \right| \left. \begin{array}{c} K \\ p \end{array} \right) \,  \left( \begin{array}{c} q \\ L
 \end{array} \right| \left. \begin{array}{cc} I & J \\ m & n \end{array} \right) = \delta^K_L \, \delta^q_p, \qquad  \left( \begin{array}{c} p \\ K
 \end{array} \right| \left. \begin{array}{cc} I & J \\ m & n \end{array} \right) = \left( \begin{array}{c} -p \\ K
 \end{array} \right| \left. \begin{array}{cc} J & I \\ -n & -m \end{array} \right).
\eeq

In some parts of the paper, we will identify  the representation spaces $V_I$ with their duals  by means of  an  invariant bilinear form on $V_I$. 
We denote by   $V_I^*$ be the dual vector space of  the representation space $V_I$ in \eqref{reps} and by $\{e^{Im}\}_{m=-I,...,I}$ be the basis dual to the basis  $\{e^I_m\}_{m=-I,...,I}$:
$$e^{Im} (e^{I}_n) = \delta^m_n.$$ 
The antipode \eqref{uqantip2} associates to each representation $\pi_{I}$ on $V_I$ a representation
 $\pi_{I}^*$ on $V_I^*$ via 
$$
( \pi_{ I}^*(a) \, \alpha)( v )= \alpha \,(\pi_{I}(S(a)) \, v )\qquad \forall a \in U_q(\su(2)), \;\;  \forall \alpha \in V_I^*, v \in V_I.
$$
The representations $\pi_{ I}$ and $\pi_{ I}^*$  are equivalent. The equivalence is given by a bijective intertwiner $\epsilon_I : V_I \rightarrow V_I^*$.  Its matrix elements and those of the dual map $\epsilon^I : V_I^* \rightarrow V_I$ with respect to the bases $\{e^I_m\}_{m=-I,...,I}$, $\{e^{Im}\}_{m=-I,...,I}$ are defined by
\begin{align}
&\epsilon_I(e^I_m)=\epsilon_{I \, nm}e^{In}\qquad \epsilon^I(e^{Im})=\epsilon^{Inm}e^I_n.
\end{align}
Here and throughout the paper, we use Einstein summation convention: repeated upper and lower indices are summed over unless stated otherwise.
A short calculation using expression \eqref{uqantip2} and \eqref{reps} shows that  the matrix elements of $\epsilon_I$ and its dual are given  up to normalisation  by
\begin{align}
\label{epsexpl}
\epsilon_{I \, mn} = \epsilon^{I mn}  = c_I e^{i \pi(I-m)}q^m\delta_{m,-n},
\end{align}
where the constant is fixed to $c_I = e^{-i \pi I} v_I^{-1/2}$ and $v_I$ is given by the action of the ribbon element 
\eqref{ribbon} on $V_I$.
The  matrix elements satisfy the identities
\beq
\label{epsilon}
\epsilon_{I mn} \, \epsilon^{I n p} = v_I^{-1} \, \delta^p_n, \;\;\;\; \mbox{and} \;\;\;\;
\epsilon_{I mn} \, \epsilon^{I p n} = v_I^{-1} e^{2i \pi I} \, \pi_I(\mu)^p_{\;\; m}.
\eeq
The intertwiner $\epsilon_I: V_I\rightarrow V_I^*$ defines an invariant bilinear form on the vector space $V_I$  via
\begin{align}
\beta_I( e^I_m,e^I_n ) =\epsilon_{I mn}, 
\qquad \beta_I( v, \pi_{I}(a) w ) = \beta_I( \pi_{I}(S(a)) v,  w )  \quad\forall v,w\in V_I.
\end{align}

\paragraph{The Hopf algebra $F_q(\SU(2))$.}

The star  Hopf algebra $F_q(\SU(2))$ is the dual of the Hopf algebra $U_q(\su(2))$ and can be viewed as a quantum deformation of the algebra of polynomial functions on $\SU(2)$. 
A basis of $F_q(\SU(2))$ is given by the matrix elements $u^{\;\;m}_{I \;\; n}: U_q(\su(2))\rightarrow\C$ in the unitary irreducible representations  \eqref{reps} of $U_q(\su(2))$
\begin{align}
\label{matrixels}
u^{\;\;m}_{I \;\; n}(x)= e^{I m}\left(\pi_I(x)e^I_n\right), \qquad \forall x\in U_q(\su(2)).
\end{align} 
The pairing $\langle\,,\,\rangle: U_q(\su(2))\times F_q(\SU(2))\rightarrow \C$ between $U_q(\su(2))$ and $F_q(\su(2))$ takes the form
\begin{align}\label{pair}
\langle x , u^{\;\;m}_{I \;\; n} \rangle = \pi_I(x)^{m}_{\;\; n}.
\end{align}
The Hopf algebra structure of $F_q(\SU(2))$ induced by the one on $U_q(\su(2))$ via the pairing \eqref{pair}. In terms of the matrix elements $u^{\;\;m}_{I \;\; n}$, its algebra structure   is characterised by the relations \ros
\beq
\label{dualuqrels}
u^{\;\;m}_{I \;\; n} \,\cdot \, u^{\;\;p}_{J \;\; q} = \sum_{K,r,s} \left( \begin{array}{cc} m & p \\
                          I & J  \end{array} \right| \left. \begin{array}{c} K \\ r \end{array} \right)  \, u^{\;\;r}_{K \;\; s} \, 
                          \left( \begin{array}{c} s \\ K
 \end{array} \right| \left. \begin{array}{cc} I & J \\ n & q \end{array} \right), \qquad\qquad 1=u^{\;\;0}_{0\;\;0}.
 \eeq                                                  
Its comultiplication, counit  and antipode take the form
\begin{align}
\label{coFq}
&\Delta \left( u^{\;\;m}_{I \;\; n} \right) = \sum_p u^m_{I \;\; p} \otimes u^{\;\;p}_{I \;\; n},\\
&\epsilon(u^{\;\;m}_{I \;\; n}) = \delta^{\;\;m}_{I \;\; n}, \label{counFq}\\
&S(u^{\;\;m}_{I \;\; n}) = \epsilon_{I np} u^{\;\;p}_{I \;\; q} \epsilon_I^{-1qm},\label{antFq}
\end{align}
where $\delta^{\;\;m}_{I \;\; n}$ is the Kronecker symbol for the representation labelled by $I$ and the coefficients $\epsilon_{Inp}$ are given by \eqref{epsexpl}. The star structure is given by
\begin{align}\label{fqstar}\star u^{\;\;m}_{I \;\; n} = S(u^{\;\;\;m}_{I \; n}).\end{align}

As any irreducible representation of $U_q(\su(2))$  can be obtained by tensoring the fundamental  representation labelled by $I=1/2$, a set of multiplicative generators of $F_q(\SU(2))$ is given by the matrix elements $a,b,c,d$ in the fundamental representation
$$u_{\frac{1}{2}} = \left( \begin{array}{cc} a & b \\ c & d \end{array}\right).$$
 They generate $F_q(\SU(2))$ multiplicatively subject to the relations
\begin{align}
&q ab=ba,  &  &q ac = ca, & &q bd = db, & &q cd = dc, \\
&bc=cb, &  &ad - da = (q^{-1} - q) bc, & &ad - q^{-1} bc = 1.\nonumber
\end{align}
In terms of these generators, the comultiplication, counit and antipode are given by
\begin{align}
\label{fqcomult} &\Delta(a)=a\oo a+b\oo c,\; \Delta(b)=a\oo b+b\oo
d, \; \Delta(c)=c\oo a+ d\oo c, \; \Delta(d)=c\oo b + d\oo d, \\
\label{fqcounit}
&\epsilon(a)=\epsilon(d)=1, \qquad\epsilon(b)=\epsilon(c)=0,\\
\label{fqantip} &S(a)=d, \quad S(d)=a, \quad S(b)=-q b, \quad S(c)=-q^{-1}
c,
\end{align}
and the pairing takes the form
\begin{align}
\label{uqfqpair} \langle q^{\pm J_z}, a \rangle=q^{\pm 1/2}, \quad\langle
q^{\pm J_z}, d \rangle=q^{\mp 1/2}, \quad \langle J_+, b \rangle=1, \quad
\langle J_-, c \rangle=1.
\end{align}

\paragraph{The quantum Lorentz group $D(U_q(\su(2)))$.}

The quantum Lorentz group is the quantum double of $U_q(\su(2))$.
It is given as the star
Hopf algebra
$$
A := \mathcal{D} \hspace{0.3mm} (U_q(\su(2))) = U_q(\su(2)) \,  \hat{\otimes} \, F_q(\SU(2))^{op},
$$
where 
 $F_q(\SU(2))^{op}$ is the Hopf algebra $F_q(\SU(2))$ with opposite coproduct, and the symbol
  `$\hat{\otimes}$'  indicates that the Hopf subalgebras $U_q(\su(2)) \otimes 1$ and $1 \otimes F_q(\SU(2))^{op}$ do not commute inside  $\mathcal{D} \hspace{0.3mm} (U_q(\su(2)))$. 
The algebra structure is given by \eqref{rels1}, \eqref{dualuqrels} together with mixed relations,  that are most easily given in terms of the multiplicative generators $a,b,c,d$ of $F_q(\SU(2))$, see for instance the appendix of \cite{PE}. In terms of these variables and the standard generators of $U_q(\su(2))$ they take the form
 \begin{align}
&q^{J_z}c=q cq^{J_z} & &q^{J_z}b=q^{-1} bq^{J_z} & &[q^{J_z}, a]=0 & &[q^{J_z},d]=0\\
&[J_+,c]=0 & &[J_+, b]=q^\inv(q^{J_z} a-q^{-J_z} d) &  &[J_-, c]=q(q^{J_z} d- q^{-J_z} a)  & &[J_-,b]=0\nonumber\\
&a J_+\!=\!q J_+a\!+\!q^{-J_z} c & &J_+ d\!=\! q d J_+ \!+\! c q^{J_z} & &J_- a\!=\! q^\inv a J_-\!+\!bq^{J_z} & &dJ_-\!=\!q^\inv J_-d\!+\!q^{-J_z}b\nonumber.
 \end{align}
  The comultiplication, counit and antipode for $D(U_q(\su(2)))$ are given by equations  \eqref{uqcomult2}, \eqref{uqcounit2}, \eqref{uqantip2} for the Hopf subalgebra $U_q(\su(2))\subset D(U_q(\su(2)))$  and by the opposite of the comultiplication in \eqref{coFq}, the counit in \eqref{counFq} and the inverse\footnote{The antipode of a Hopf algebra $H^{op}$ with the opposite coproduct is the inverse of the antipode of the Hopf algebra $H$.}   of the antipode \eqref{antFq} for the Hopf subalgebra $F_q(\SU(2))^{op}\subset D(U_q(\su(2)))$.
The star structure is given by \eqref{lorstar} and by \eqref{fqstar},  where the antipode  in \eqref{fqstar} is replaced by its inverse.
As it is a quantum double, the quantum Lorentz group is a braided Hopf algebra. We will describe its universal $R$-matrix after introducing  its double dual  in Section \ref{irrepslor} below.

The quantum Lorentz group $D(U_q(\su(2)))$ is a quantum deformation of the universal enveloping algebra of the real Lie algebra $\mathfrak{sl}(2,\C)_{\mathbb{R}} \cong \mathfrak{so}(3,1)$. This is a direct consequence of the quantum duality principle, see for instance \cite{PE,kassel,CP, KS}.
Recall that the coalgebra structure on the deformed enveloping algebra $U_q(\mathfrak{g})$ of a Lie algebra $\mathfrak{g}$ induces a Lie algebra structure of the dual vector space $\mathfrak{g}^*$. The principle of quantum duality states that the $q$-deformed universal enveloping algebra $U_q(\mathfrak{g}^*)$ of $\mathfrak{g}^*$ is given by
$U_q(\mathfrak{g}^*) = (U_q(\mathfrak{g}))^*$. On the level of quantised function spaces, one obtains that the Hopf algebra of quantum deformations of polynomial functions on the group $G$ is given by
 $F_q(G) = (U_q(\mathfrak{g}))^*$, where $\mathfrak g=\text{Lie}(G)$. 

In the case at hand, this yields  $F_q(\SU(2)) = U_q(\mathfrak{su}(2))^* = U_q(\mathfrak{su}(2)^*) \cong U_q(\mathfrak{an}(2))$  \cite{PE}, where  $\mathfrak{an}(2)$ is the Lie algebra of the group  $\mathrm{AN}(2) = \mathrm{A}(2) \times \mathrm{N}(2)$. Here, $\mathrm{A}(2)$ denotes the group of  diagonal positive $2\times 2$-matrices of determinant one and $\mathrm{N}(2)$ is the nilpotent group of lower triangular two by two matrices with diagonal elements equal to one.
The quantum double construction is therefore the quantum analogue of the Iwasawa decomposition of the classical Lorentz algebra $\mathfrak{sl}(2,\C)_{\mathbb{R}} \cong \su(2) \oplus \mathfrak{an}(2)$, and we will use the notation $D(U_q(\su(2)))= U_q(\mathfrak{sl}(2,\C)_{\mathbb{R}})$.

\subsubsection{Irreducible representations, duals and $R$-matrix}
\label{irrepslor}

\paragraph{Irreducible representations of the quantum Lorentz group}
The irreducible unitary representations of $A = U_q(\mathfrak{sl}(2,\mathbb{C})_\mathbb{R})$  were first classified by 
Pusz \cite{Pusz}. In this paper, we will only consider  the representations of the principal series. These representations are labelled by a couple 
$\alpha=(n,p)$ with  $n\in \mathbb{Z}/2
$ and $p\in [0,\frac{4 \pi}{\kappa} [$ or with $n=0$ and $p \in [0,\frac{2\pi}{\kappa} ]$.
We denote by $(\pi_{\alpha},V_{\alpha})$ the representation of $U_q(\mathfrak{sl}(2,\mathbb{C})_\mathbb{R})$ labelled by $\alpha$. It is a Harish-Chandra representation which decomposes into representations of $U_q(\mathfrak{su}(2))$ as follows 
\beq
\label{qrepresentation}
V_{\alpha} = \bigoplus_{I= \mid n \mid}^{\infty} V_I,
\eeq
where $V_I$ is the left $U_q(\mathfrak{su}(2))$-module \eqref{reps}.  A basis of the infinite dimensional vector space $V_{\alpha}$ is given by  $\{e^I_m \mid I\in \mathbb{N}_0, I\geq \mid n \mid, m = -I, ..., I \}$ where, for fixed $I$, $\{e^I_m\}_{m=-I,...,I}$ is the basis of $V_I$ defined before equation \eqref{reps}. In terms of this basis, the action of $D(U_q(\su(2)))$ on the representation space $V_\alpha$ is given by equation \eqref{reps} for the action of $U_q(\su(2))$ and the following action of $F_q(\su(2))$ \ros
\begin{align}
\label{repQLGmod}
\pi_{\alpha}(u^{\;\;b}_{J \;\; b'}) \; e^L_c = \sum_{M,N= \mid n \mid}^\infty e^N_{c'} \; 
 \left( \begin{array}{cc} c' & b \\
                          N & J  \end{array} \right| \left. \begin{array}{c} M \\ d \end{array} \right)                            \left( \begin{array}{c} d \\ M
 \end{array} \right| \left. \begin{array}{cc} J & L \\ b' & c \end{array} \right) \; \Lambda^{JM}_{NL}(\alpha), 
\end{align}
where $\Lambda^{JM}_{NL}(\alpha) = \Lambda^{JM}_{LN}(\alpha)$ are complex numbers defined in terms of analytic 
continuations of $6j$ symbols for $U_q(\su(2))$. As their expressions are lengthy and complicated, we will not give 
them here but refer the reader to \cite{PE}, where they are derived explicitly, and to \cite{PE2} where their properties 
are studied in depth.

By introducing a hermitian form $( \, , \, )_{\alpha}$ for which the  basis $\{e^I_m\}_{I,m}$ is orthonormal
$$
(e^I_m, e^J_n)_\alpha=\delta^{IJ}\delta_{mn},
$$
the representation spaces $V_{\alpha}$ can be given a pre-Hilbert space structure. Its completion $\mathcal{H}_{\alpha}$ with respect to the associated norm is  separable Hilbert space with Hilbert basis $\{e^I_m\}_{I,m}$.
The representations of the principal series are unitary in the sense that for all $v,w$ in $V_{\alpha}$ and for all $a$ in $A$, $( \pi_{\alpha}(\star a) v, w )_{\alpha} = (  v, \pi_{\alpha}(a) w )_{\alpha}$, that is, $\pi_{\alpha}(\star a) = \pi_{\alpha}(a)^{\dagger}$.

Note that 
the finite-dimensional unitary representations of the quantum Lorentz group do not form a ribbon category. As the representation spaces are infinite-dimensional, there is no notion of (quantum) trace. As in the classical case, this poses a considerable obstacle for the definition of the amplitude for the four-simplexes of the EPRL model and the definition of a consistent graphical calculus for the quantum Lorentz group. We will come back to this point in the sequel.

\paragraph{Algebra of functions on the quantum Lorentz group.}
Matrix elements  of the principal representations of the quantum Lorentz group  are linear forms on $U_q(\mathfrak{sl}(2,\mathbb{C})_\mathbb{R})$ and hence elements of its dual $F_q(\SL(2,\C)_{\mathbb{R}})$. 
In the following, we will therefore  consider the Hopf algebra of functions  $F_q(\SL(2,\C)_{\mathbb{R}})$  on the quantum Lorentz group. 
The algebra $F_q(\SL(2,\C)_{\mathbb{R}})$ decomposes as
$F_q(\SL(2,\C)_{\mathbb{R}}) = F_q(\SU(2)) \hat \otimes F_q(\mathrm{AN}(2))$ \cite{PE}, where the Hopf algebra $F_q(AN(2))$ is interpreted as  the space of (compactly supported) functions on $\mathrm{AN}(2)$. It is the dual  of the Hopf algebra $F_q(\SU(2))^{op}$ and hence can be identified with  the Hopf algebra $(F_q(\SU(2))^*)_{op}$ with the opposite multiplication. 

Denoting by $u^{\;\;m}_{I \;\; n}$ the basis of $F_q(\SU(2))$ introduced in \eqref{matrixels} and by $E^{\;\;a}_{I \;\; b}$ the elements of its dual basis given by
$$
\langle E^{\;\;a}_{I \;\; b} \, , u^{\;\;c}_{J \;\; d} \rangle = \delta_{IJ} \delta^a_d \delta^c_b,
$$
one finds that a convenient basis 
of $F_q(\SL(2,\C)_{\mathbb{R}})=F_q(\SU(2)) \hat \otimes F_q(\mathrm{AN}(2))$ is provided by the elements 
 $\{x_A\}_A = \{u^{\;\;a'}_{I \;\; a} \otimes E^{\;\;b'}_{J \;\; b}\}_{I,J,a,a',b,b'}$.
The star Hopf algebra structure of $F_q(\SL(2,\C)_{\mathbb{R}})$ follows immediately from the one of $D(U_q(\su(2)))=U_q(\su(2))\hat{\otimes}F_q(SU(2))^{op}$ via the duality principle. 
In particular, one finds that the multiplication on $F_q(\mathrm{AN}(2))\subset F_q(SL(2,\C)_\RR)$
takes the form
\beq
\label{dualAN2rels}
E^{\;\;a}_{I \;\; b} \, E^{\;\;c}_{J \;\; d} = \delta_{IJ} E^{\;\;a}_{I \;\; d} \, \delta^c_b.
\eeq

\paragraph{Universal $R$-matrix and braiding.}
To obtain a simple expression for the universal-$R$ matrix of $D(U_q(\su(2)))$, it is convenient
to work with the Hopf algebra
 $F_q(\SL(2,\C)_{\mathbb{R}})^*$ introduced in \cite{PE}, which can be viewed as the double dual of $D(U_q(\su(2)))=U_q(\mathfrak{sl}(2,\mathbb{C})_\mathbb{R})$. As the Hopf algebra $U_q(\mathfrak{sl}(2,\mathbb{C})_\mathbb{R})$  is infinite-dimensional, this is not identical to $U_q(\mathfrak{sl}(2,\mathbb{C})_\mathbb{R})$ but 
 contains $U_q(\mathfrak{sl}(2,\mathbb{C})_\mathbb{R})$ as a Hopf subalgebra \cite{PE}. It factorises as $F_q(\SL(2,\C)_{\mathbb{R}})^* = F_q(\SU(2))^* \hat{\otimes} F_q(\SU(2))^{op}$. A convenient basis of $F_q(\SL(2,\C)_{\mathbb{R}})^*$ is given by the basis $\{x^A\}_A = \{X^{\;\;a}_{I \;\; a'} \otimes g^{\;\;b}_{J \;\; b'}\}_{I,J,a,a',b,b'}$ dual to $\{x_A\}_A$.
 In terms of this basis,  the universal  $R$-matrix of $D(U_q(\su(2))$ takes a particularly simple form, namely \ros
\beq
\label{Rmatrix}
R = \sum_{I,a,a'} X^{\;\;a}_{J \;\; a'} \otimes 1 \otimes 1 \otimes g^{\;\;a'}_{J \;\; a}.
\eeq
To derive an explicit expression for the braiding, we note that
if $\pi_{\alpha}$  is a principal representation of $U_q(\mathfrak{sl}(2,\mathbb{C})_\mathbb{R})$  on $V_{\alpha}$, there is a unique representation \cite{PE} of $F_q(\SL(2,\C)_{\mathbb{R}})^*$  on $V_{\alpha}$, which will also be denoted by $\pi_{\alpha}$. Its action on the  basis elements $e^I_m$ of the preHilbert space $V_{\alpha}$ is given by \ros
\beqa
\label{repQLG}
\pi_{\alpha}(X^{\;\;a}_{I \;\; a'}) e^L_c &=& \delta_{IL} \delta^a_c e^L_{a'} \\
\pi_{\alpha}(g^{\;\;b}_{J \;\; b'}) \; e^L_c &=& \sum_{M,N=\mid n \mid}^{\infty} e^N_{c'} \; 
 \left( \begin{array}{cc} c' & b \\
                          N & J  \end{array} \right| \left. \begin{array}{c} M \\ d \end{array} \right)                            \left( \begin{array}{c} d \\ M
 \end{array} \right| \left. \begin{array}{cc} J & L \\ b' & c \end{array} \right) \; \Lambda^{JM}_{NL}(\alpha), \nn
\eeqa
where $\Lambda^{JM}_{NL}(\alpha)$ are the coefficients from \eqref{repQLGmod}.
It  is then immediate to obtain explicit expressions for the action of the universal $R$-matrix in these representations: \ros
\beq
\label{repR}
(\pi_{\alpha} \otimes \pi_{\beta})(R) \, e^I_c \otimes e^J_d = \sum_{K=\mid n \mid}^{\infty} \sum_{L=\mid n' \mid}^{\infty} e^I_e \otimes e^L_f \; 
 \left( \begin{array}{cc} f & e \\
                          L & I  \end{array} \right| \left. \begin{array}{c} K \\ g \end{array} \right)                            \left( \begin{array}{c} g \\ K
 \end{array} \right| \left. \begin{array}{cc} I & J \\ c & d \end{array} \right) \; \Lambda^{IK}_{LJ}(\alpha),
 \eeq
 where $\alpha=(n,p)$ and $\beta=(n',p')$. Note that although these sums are infinite, there is only a finite number of non-zero terms \cite{PE2}. Consequently, there are no  issues with convergence.

\subsection{The quantum EPRL intertwiner}
\label{qeprl}

We are now ready to  construct the generalised EPRL model associated with the quantum Lorentz group.   The first step is to define the state space associated to the three-simplexes of a triangulated $4$-manifold $M$. This is the vector space of EPRL intertwiners 
between the four EPRL representations associated
 with its  boundary triangles  and the trivial EPRL representation on  $\C$.

 \subsubsection{Quantum EPRL representations}
 
A central ingredient in the construction  is the generalisation  of the notion of an EPRL representation to the quantum Lorentz group. 
 In the classical model, an EPRL representation assigns a representation of $\SL(2,\C)_{\mathbb{R}}$ to each finite-dimensional representation of $\SU(2)$. This assignment can be viewed as a lift from the representation category of $\SU(2)$ to the representation category of $\SL(2,\C)_{\mathbb{R}}$, and is given by the prescription
 $$
 K\mapsto(n(K),p(K))=(K, \gamma K),
 $$
 where $K\in\NN_0/2$ labels the irreducible representations of $\SU(2)$ and $\gamma\in\RR^+$ is the Immirzi parameter.  In the quantum Lorentz group, the $\SU(2)$-representations in the classical model 
 correspond to irreducible representations of the Hopf algebra $U_q(\su(2))$, which are labelled by a parameter $K\in\NN_0/2$. The inclusion
 $
 \SU(2)\rightarrow \SL(2,\C)_{\mathbb{R}} \cong \SU(2) \times \mathrm{AN}(2)$, $g\mapsto (g,a)$
 is replaced by the inclusion
 $
 U_q(\su(2))\rightarrow D(U_q(\su(2)))= U_q(\su(2))\hat{\otimes} F_q(\SU(2))$, $a\mapsto a\otimes 1$.
 It is thus natural to require that that quantum EPRL representations are given by a tensor functor between the 
 category  of representations of $U_q(\su(2))$ and the representation category of the quantum Lorentz group, which is compatible with this inclusion of $U_q(\su(2))$ into $D(U_q(\su(2)))$. The image of the functor is a subset of representations of $A$ called EPRL representations.  The analogy with the classical case suggests that this functor should act on the irreducible representations of $U_q(\su(2))$ according to
 $$
 K\in\NN_0/2 \mapsto \alpha(K)=(n(K), p(K))=(K,\gamma K),
 $$
where $\alpha(K)$ labels an irreducible unitary principal series representation of the quantum Lorentz group and $\gamma>0$ is the Immirzi parameter.  Note, however, that the parameter $p$
 in the representation labels $\alpha=(n,p)$ of the principal series is restricted to the interval $[0, 4\pi/\kappa[$.  To obtain a consistent definition, it is thus necessary to 
   restrict the preimage of this  functor, i.~e.~the representation label $K$. Remark also that there is no loss of generality in considering only the case where $n(K)$ is positive since the representations $(n,p)$ and $(-n,-p)$ are equivalent. 
  We are now led to the following definition. 
   
\begin{definition} {\bf (EPRL representations)}
Let $\gamma$ be a fixed real positive parameter (the Immirzi parameter) and consider the subset of irreducible representations of $U_q(\su(2))$ labelled by
$$
\mathcal{L} = \mathbb{N}/2 \, \cap \, [0, 4 \pi/ \gamma \kappa[.
$$
The Lorentzian EPRL representation  of  spin $K\in \mathcal{L}$ is the principal representation of $D(U_q(\su(2)))$ labelled by
$$
\alpha(K) = (n(K),p(K)) := (K,\gamma K) 
$$
\end{definition}

The restriction of the representation labels of $U_q(\su(2))$ to the label set  $\mathcal{L}$  ensures that the Lorentzian EPRL representation $\alpha(K)$ is a principal representation of the quantum Lorentz group, which would not be the case otherwise.
It decomposes into irreducible representations  of $U_q(\su(2))$  as follows
\beq
\label{decompL}
V_{\alpha(K)} = \bigoplus_{J=K}^{\infty} V_J.
\eeq

\subsubsection{Quantum EPRL intertwiners}

We are now ready to generalise the notion of quantum EPRL intertwiner to the quantum Lorentz group. This requires a generalisation of 
integral expressions of the type
\beq
\label{int}
T_{V[\rho]} = \int_G d X \left( \bigotimes_{i=1}^n \pi_{\rho_i} \right)(X),
\eeq
where  $G$ is a compact, unimodular  Lie group with Haar measure $dX$,  $\rho_i : G \rightarrow \End V_{\rho_i}$  are  unitary irreducible representations of $G$ and
and $V[\rho] := \bigotimes_{i=1}^n V_{\rho_i}$. 
The generalisation of such integrals to $q$-deformed universal enveloping algebras $U_q(\mathfrak g)$ and the associated $q$-deformed function spaces $F_q(G)=U_q(\mathfrak g)^*$  requires the notion of a Haar measure or biinvariant normalised integral on $F_q(G)$. 
This  is  a non-degenerate linear form  $h : F_q(G) \rightarrow \C$ that satisfies the identities
\begin{align}\label{invar}
(h \otimes id)\Delta(x) = h(x) 1 \;\;\;\; \mbox{and} \;\;\;\; (id \otimes h)\Delta(x) = h(x) 1\qquad \forall x \in F_q(G).
\end{align}
For a pedagogical introduction, see for instance \cite{majid}.
The identities \eqref{invar}  imply that a Haar measure on $F_q(G)$  is invariant under the left-  and right-action of $U_q(\g)$ on $F_q(G)$.  For given $a\in U_q(\mathfrak g)$, the left- and right- action
$L_a , R_a: F_q(G) \rightarrow F_q(G)$  are defined via the pairing $\langle\,,\,\rangle$  between $U_q(\g)$ and $F_q(G)$:
\begin{align}\label{actions}
 \langle L_a f, b \rangle =  \langle f, a b \rangle\qquad \langle R_a f, b \rangle =  \langle f, ba \rangle\qquad\forall b\in U_q(\g), f\in F_q(G).
\end{align}
 The invariance property \eqref{invar} then implies  
\begin{align}\label{haarinvar}
h(L_a f) = h (R_a f) = \epsilon(a) h(f)\qquad \forall f \in F_q(G), a\in U_q(\g).
\end{align}
The  Haar measure is thus invariant under the 
 left- and right-action of the Hopf algebra  $U_q(\g)$ on its dual $F_q(G)$, and this invariance can be viewed as a a generalisation  of  the left- and right-invariance of the Haar measure on a unimodular  Lie group.

Given a Haar measure on $F_q(G)$, it is straightforward to generalise expression \eqref{int} to the representations of $U_q(\g)$ \cite{PE2}, \cite{PK}. For this, one introduces a basis   $\{x_A\}_{A}$ of $F_q(G)$  and denotes by $\{x^A\}_{A}$ the associated dual basis of $F_q(G)^*$. 
Given irreducible representations $\rho_i: U_q(\g)\rightarrow \text{End}(V_{\rho_i})$, $i=1,...,n$, of $U_q(\g)$, one obtains a representation
 $\rho: U_q(\g)\rightarrow\text{End}(V_\rho)$ of $U_q(\g)$ on the tensor product $V[\rho]=\bigotimes_{i=1}^n V_{\rho_i}$ of the associated representation spaces
$$
\rho(a)=\left(\bigotimes_{i=1}^n \rho_i\right)\circ \Delta^{(n-1)}(a)\qquad \forall a\in U_q(\g),
$$
where $\Delta^{(n)}$ denotes the $n$-fold coproduct in $U_q(\g)$: $\Delta^{(n)}=(\Delta\otimes \text{id}^{\otimes (n-1)}) \circ \Delta^{(n-1)}$ for $n>1$, $\Delta^{(1)}=\Delta$.
The integral \eqref{int} then corresponds to  the expression
\begin{align}\label{classgenq}
T_{V[\rho]} = \sum_{A}  \left( \bigotimes_{i=1}^n {\rho_i} \right)\circ (\Delta^{(n-1)}(x^A)) \, h(x_A).
\end{align}
Note, however, that the existence of a Haar measure on $F_q(G)$ and the convergence of expression \eqref{classgenq} is not guaranteed a priori in the case where $U_q(\g)$ and $F_q(G)$ are infinite-dimensional.

In the case of the quantum Lorentz group, we consider the $q$-deformed universal enveloping algebra  $D(U_q(\su(2)))=U_q(\sll(2,\C)_\RR)$ and its dual $F_q(SL(2,\C)_\RR)$. It is shown in \cite{PE,PE2}  that a Haar measure on $F_q(SL(2,\C)_{\RR})$
is obtained from the decomposition $F_q(SL(2,\C)_\RR)=F_q(\SU(2))\hat \otimes F_q(AN(2))$ introduced in Section \ref{irrepslor}.  It is  unique up to normalisation and given as the tensor product of  the unique normalised biinvariant integral $h_1$ on $F_q(\SU(2))$ and a right-invariant integral $h_2$ on 
$F_q(\mathrm{AN}(2))$. On the basis elements $u^{\;\;a}_{I \;\; a'}$ of $F_q(SU(2))$ and the basis elements $E^{\;\;b}_{J \;\; b'}$ of $F_q(AN(2))$ the latter  take the form
\beq
h_1(u^{\;\;a'}_{I \;\; a}) = \delta_{I,0}, \;\;\;\; \mbox{and} \;\;\;\; h_2(E^{\;\;b'}_{J \;\; b}) = [2J+1] \pi_J(\mu^{-1})^{b'}_{\;\; b}.
\eeq
The  Haar measure on $F_q(SL(2,\C)_\RR)$ is therefore given by
\beq
\label{integral}
h(u^{\;\;a'}_{I \;\; a} \otimes E^{\;\;b'}_{J \;\; b}) = (h_{1} \otimes h_2) (u^{\;\;a'}_{I \;\; a} \otimes E^{\;\;b'}_{J \;\; b}) = \delta_{I0} \delta^a_0 \delta_{a'}^0 \; [2J+1] \pi_J(\mu^{-1})^{b'}_{\;\; b},
\eeq
where $\pi_J(\mu^\inv)^{b'}_{\;\;b}$ denotes  the matrix elements of the group-like element $\mu^\inv\in U_q(\su(2))$ in \eqref{mudef} in the $U_q(\su(2))$ representation labelled by  $J\in\NN_0/2$
$$
\pi_J(\mu^\inv)^{b'}_{\;\;b}=q^{-2b} \delta^{b'}_{\;\;b}
$$
This reduces the summation over the basis in \eqref{classgenq} to a summation over the parameters $J\in\NN_0/2$ and $b\in\{-J,...,J\}$. It is shown in \cite{PE2} that the associated expression for  $T_{V_\rho}$ in \eqref{classgenq} converges for $n=3$ and defines an intertwining map
$T_{V[\rho]}:  \bigotimes_{i=1}^n V_{\rho_i}\rightarrow \bigotimes_{i=1}^n V_{\rho_i}$.

This map allows us to generalise the notion of an EPRL intertwiner to the quantum Lorentz group. 
For this, we consider 
the projection $f_{\alpha}^K : V_{\alpha} \rightarrow V_K$, which maps the EPRL representation labelled by $
\alpha(K)$ to the lowest weight factor in  the decomposition \eqref{decompL} of  $V_{\alpha(K)}$. In terms of the basis $\{e^I_m(\alpha)\}$ of $V_{\alpha(K)}$ and the basis $\{e^K_m\}$ of $V_K$, this projector takes the form
\beq
\label{projection}
f_{\alpha}^K (e^I_m(\alpha)) = e^K_m \delta^{IK} 
\eeq
We will also make use of the inclusion maps
$f^{\alpha}_K : V_K \rightarrow V_{\alpha}$ associated to the  the decomposition \eqref{decompL} which satisfy the identity $f_{\alpha}^{K'}  \circ f^{\alpha}_K = id_K \delta_{K}^{K'}$.
The dual $(f^*)_K^{\alpha} : \mathcal{L}(V_K,\C) \rightarrow \mathcal{L}(V_{\alpha},\C)$ of the projector $f^K_\alpha$  defines an embedding of the vector space $\mathcal L(V_K,\C)$ of linear forms on $V_K$ into the vector space  $\mathcal L(V_\alpha,\C)$ of linear forms on $V_{\alpha}$. It defines one of the building blocks in the 
construction of an embedding
$$
\Hom_{U_q(\su(2))}(\otimes_{i=1}^n V_{K_i} , \C) \rightarrow \Hom_{D(U_q(\su(2)))}(\otimes_{i=1}^n V_{\alpha_i(K_i)} , \C), 
$$
of the vector space of $n$-valent $U_q(\su(2))$) intertwiners into the vector space of $n$-valent intertwiners for  the quantum Lorentz group. This embedding  map will be noted $f^*$. 
Given an element $\Lambda\in \Hom_{U_q(\su(2))}(\otimes_{i=1}^n V_{K_i} , \C)$, we define the quantum EPRL intertwiner associated with $\Lambda$ as  the image $\iota = f^*_{K}(\Lambda)$ of $\Lambda$ under  $f^*_K$.  This yields the following definition.

\begin{definition} {\bf (Quantum EPRL intertwiner)}
Let $K = (K_1, ..., K_n)$ be a $n$-tuple of elements of the label set $\mathcal{L}=\NN/2 \, \cap \, [0,4\pi/\gamma\kappa[$ and $V[K] = \bigotimes_{a=1}^n V_{K_a}$ be the corresponding representation space of $U_q(\su(2))$. Denote by  $\alpha = (\alpha_1(K_1),...,\alpha_n(K_n))$ the associated $n$-tuple $\alpha = (\alpha_1(K_1),...,\alpha_n(K_n))$ of EPRL representations and by $V[\alpha] = \bigotimes_{i=1}^n V_{\alpha_i}$ the tensor product of their  representation spaces. The quantum EPRL 
intertwiner $\iota_{\alpha} = f^*(\Lambda_K)$ associated to an intertwiner $\Lambda_{K}$  in  $\mathrm{Hom}_{U_q(\su(2))}(V[K], \C)$
 is  the linear  map $\iota_{\alpha} : V[\alpha] \rightarrow \C$   
\begin{align}
\label{quintertwiner}
\iota_{\alpha}= \sum_{A}  \Lambda_{K} \circ  f^{K}_{\alpha} \circ \left( \bigotimes_{i=1}^n \pi_{\alpha_i(K_i)} (\Delta^{(n-1)}(x^A)) \right) h(x_{A}),
\end{align}
where the map $f^{K}_{\alpha} : V[\alpha] \rightarrow V[K]$ is given as the tensor product of the map  \eqref{projection} 
\begin{align}
f_{\alpha}^K = \bigotimes_{i=1}^n f_{\alpha_i}^{K_i}.
\end{align}
\end{definition}

As the definition of the EPRL intertwiner involves an infinite sum over basis elements of  \linebreak$F_q(SL(2,\C)_{\RR})$ and their duals, it remains to show that the series in \eqref{quintertwiner} converges for all choices of $U_q(\su(2))$-intertwiners $\Lambda_K$.  For this, we introduce a basis of  the vector space  $\mathrm{Hom}_{U_q(\su(2))}(V[K], \C)$, which  
corresponds to the specification of a recoupling scheme.  
In the following we  will mainly be interested in the case $n=4$. In this case, an orthogonal basis of $\Hom_{H}(\otimes_{a=1}^4 V_{K_a} , \C)$ is given by the set of $U_q(\su(2))$-intertwiners $\{\lambda_{K,I}\}_{I\in\NN_0/2}$ with
\begin{align}
\lambda_{K,I} = d_I \circ \left(C_{K_1 K_2 I} \otimes C^I_{\; K_3 K_4}\right),
\end{align}
where $C_{K_1 K_2I} = \epsilon_I \circ C_{\; K_1 K_2}^I$ is an element of $\Hom_{H}(V_{K_1} \otimes V_{K_2} \otimes V_I, \C) \cong \Hom_{H}(V_{K_1} \otimes V_{K_2}, V_I^*)$, and $d_I : V_I^* \otimes V_I \rightarrow \C$, $e^{Ia} \otimes e^I_b \mapsto e^{Ia} (e^I_b) = \delta^a_b$, denotes the evaluation in the ribbon category of representations of $U_q(\su(2))$.
This implies that the evaluation of the intertwiner $\lambda_{K,I}$ is given in terms of Clebsch-Gordan coefficients as follows \ros
\begin{align}\label{lambdaikdef}
\lambda_{K,I} (e^{K_1}_{a_1} \otimes e^{K_2}_{a_2} \otimes e^{K_3}_{a_3} \otimes e^{K_4}_{a_4})\! =:\!\! \left( \begin{array}{cccc} K_1 & K_2 & K_3 & K_4 \\ 
 a_1 & a_2 & a_3 & a_4 \end{array} \right)_{I}\!\! \!=\! \left( \begin{array}{c} a \\ I
 \end{array} \right| \left. \begin{array}{cc} K_1 & K_2 \\ a_1 & a_2 \end{array} \right) \epsilon_{I ab} \left( \begin{array}{c} b \\ I
 \end{array} \right| \left. \begin{array}{cc} K_3 & K_4 \\ a_3 & a_4 \end{array} \right).
\end{align}
Using the orthogonality of the Clebsch-Gordan intertwiners \eqref{orthogonality} and the properties of the map $\epsilon_I : V_I \rightarrow V_I^*$ stated in \eqref{epsexpl}, one  obtains an  orthogonality relation for the coefficients of the intertwiners $\lambda_{I,K}$
\beq
\label{orthogonality4}
\left( \begin{array}{cccc} K_1 & K_2 & K_3 & K_4 \\ 
 a_1 & a_2 & a_3 & a_4 \end{array} \right)_{I} \left( \begin{array}{cccc} a_1 & a_2 & a_3 & a_4 \\ 
 K_1 & K_2 & K_3 & K_4 \end{array} \right)_{K} =  q^{2I(I+1)} \, [2I+1] \, \delta_{IK}.
\eeq 

We can now demonstrate the convergence of the four-valent
 quantum EPRL intertwiners  to elements of $\Hom_{D(U_q(\su(2)))}(V[\alpha], \C)$ and obtain the following theorem.

\begin{theorem}\label{qeprlth}
Consider the  $4$-tuple $K=(K_1,...,K_4)$ of irreducible representations of $U_q(\su(2))$ associated with parameters in the label set $K_i\in \mathcal L$ and denote by 
$\alpha = (\alpha_1(K_1), ..., \alpha_4(K_4))$ the corresponding  $4$-tuple of EPRL representations.
Let  $\lambda_{K,N}$ be an element of the basis $\{\lambda_{K,N}\}_{N\in\NN_0/2}$ of $\Hom_{U_q(\su(2))}(V[K], \C)$
and  $e^L_{c}[\alpha] = \bigotimes_{i=1}^4 e^{L_i}_{c_i}(\alpha_i)$ a basis of $V[\alpha]$. Then the evaluation of the quantum EPRL intertwiner $\iota_{\alpha,N} = f^*(\lambda_{K,N})$ is given by \ros

\begin{align}
\label{evalL}
\iota_{\alpha,N} \, (e^L_{c}[\alpha]) = \sum_{I} &\sum_{m_1,...,m_4} \!\!\![2I+1] \, \pi_I(\mu^{-1})^{b'_1}_{\;\; b_4} \Lambda^{I M_1}_{K_1L_1}(\alpha_1) \Lambda^{I M_2}_{K_2L_2}(\alpha_2)  \Lambda^{I M_3}_{K_3L_3}(\alpha_3)  \Lambda^{I M_4}_{K_4L_4}(\alpha_4) \nn \\
&
 \left( \begin{array}{cc} a_1 & b_1 \\
                          K_1 & I  \end{array} \right| \left. \begin{array}{c} M_1 \\ d_1 \end{array} \right)                            \left( \begin{array}{c} d_1 \\ M_1
 \end{array} \right| \left. \begin{array}{cc} I & L_1 \\ b'_1 & c_1 \end{array} \right) \left( \begin{array}{cc} a_2 & b_2 \\
                          K_2 & I  \end{array} \right| \left. \begin{array}{c} M_2 \\ d_2 \end{array} \right)\!                            \left( \begin{array}{c} d_2 \\ M_2
 \end{array} \right| \left. \begin{array}{cc} I & L_2 \\ b_1 & c_1 \end{array} \right)\nn \\
&
\left( \begin{array}{cc} a_3 & b_3 \\
                          K_3 & I  \end{array} \right| \left. \begin{array}{c} m_3 \\ d_3 \end{array} \right)                            \left( \begin{array}{c} d_3 \\ m_3
 \end{array} \right| \left. \begin{array}{cc} I & L_3 \\ b_2 & c_3 \end{array} \right) \left( \begin{array}{cc} a_4 & b_4 \\
                          K_4 & I  \end{array} \right| \left. \begin{array}{c} M_4 \\ d_4 \end{array} \right)                            \left( \begin{array}{c} d_4 \\ M_4
 \end{array} \right| \left. \begin{array}{cc} I & L_4 \\ b_3 & c_4 \end{array} \right) 
 \nn \\
&  
 \left( \begin{array}{c} a \\ n
 \end{array} \right| \left. \begin{array}{cc} K_1 & K_2 \\ a_1 & a_2 \end{array} \right) \epsilon_{N ab} \left( \begin{array}{c} b \\ N
 \end{array} \right| \left. \begin{array}{cc} K_3 & K_4 \\ a_3 & a_4 \end{array} \right).
\end{align}
This series converges absolutely and defines an element of $\Hom_{D(U_q(\su(2)))}(V[\alpha], \C)$. 
\end{theorem}

{\em Proof.} The proof  contains two steps. We first show the identity \eqref{evalL}  and then prove  the convergence. By definition, the evaluation of the EPRL intertwiner is given by the expression \ros
\beq
\iota_{\alpha,N} \, (e^L_{c}[\alpha]) = \sum_{A_1, ... A_4}  \lambda_{K,N} \circ  f^{K}_{\alpha} \circ \left( \pi_{\alpha_1}(x^{A_1}) \, e_{c_1}^{L_1} (\alpha_1) \otimes ... \otimes \pi_{\alpha_4} (x^{A_4}) \, e_{c_4}^{L_4} (\alpha_4) \right) h(x_{A_1} ... x_{A_4}), \nn
\eeq
where  $\alpha_i$ is a shorthand notation for $\alpha_i(K_i)$, $i=1,...,4$ and we use Sweedler's notation $\Delta(x^A) = x^{A_1} \otimes x^{A_2}$. We now choose as the basis $\{x_A\}_A$ of $F_q(\SL(2,\C)_{\mathbb{R}})$ the basis with elements $u^{\;\;a}_{I \;\; a'} \otimes E^{\;\;b}_{J \;\; b'}$ introduced in Section \ref{irrepslor}. Its dual is the basis
 of $F_q(\SL(2,\C)_{\mathbb{R}})^*$  with basis elements $X^{\;\;a'}_{I \;\; a} \otimes g^{\;\;b'}_{J \;\; b}$. Using equation \eqref{repQLG}, we obtain \ros
\begin{align} &\iota_{\alpha,N} \, (e^L_{c}[\alpha]) = \!\!\!\!\sum_{I_1, ... I_4} \sum_{J_1,...,J_4} \sum_{M_1,...,M_4}  \!\!\!\!\lambda_{K,N} \circ f_{\alpha}^K  ( e_{a'_1}^{I_1} (\alpha_1) \otimes ... \otimes \, e_{a'_4}^{I_4} (\alpha_4))\;\; h(u^{\;\;a'_1}_{I_1 \;\; a_1} \!\!\otimes\! E^{\;\;b'_1}_{J_1 \;\; b_1} ...  u^{\;\;a'_4}_{I_4 \;\; a_4}\!\! \otimes\! E^{\;\;b'_4}_{J_4 \;\; b_4}) \nn \\
&\Lambda^{J_1M_1}_{I_1L_1}(\alpha_1) \!
 \left( \begin{array}{cc} a_1 & b_1 \\
                          I_1 & J_1  \end{array} \right| \left. \begin{array}{c} M_1 \\ d_1 \end{array} \right) \!\!                           \left( \begin{array}{c} d_1 \\ M_1
 \end{array} \right| \left. \begin{array}{cc} J_1 & L_1 \\ b'_1 & c_1 \end{array} \right) \!\!... \;
\Lambda^{J_4M_4}_{I_4L_4}(\alpha_4) \!
 \left( \begin{array}{cc} a_4 & b_4 \\
                          I_4 & J_4  \end{array} \right| \left. \begin{array}{c} M_4 \\ d_4 \end{array} \right)  \!\!                          \left( \begin{array}{c} d_4 \\ M_4
 \end{array} \right| \left. \begin{array}{cc} J_4 & L_4 \\ b'_4 & c_4 \end{array} \right) \nn \\
& = \!\!\!\sum_{J_1,...,J_4} \sum_{M_1,...,M_4} \left( \begin{array}{cccc} K_1 & K_2 & K_3 & K_4 \\ 
 a'_1 & a'_2 & a'_3 & a'_4 \end{array} \right)_{N} \;\;h_1(u^{\;\;a'_1}_{K_1 \;\; a_1} \; ... \; u^{\;\;a'_4}_{K_4 \;\; a_4}) h_2(E^{\;\;b'_1}_{J_1 \;\; b_1} \; ... \;  E^{\;\;b'_4}_{J_4 \;\; b_4}) \nn \\
&\Lambda^{J_1M_1}_{K_1L_1}(\alpha_1) \!
 \left( \begin{array}{cc} a_1 & b_1 \\
                          K_1 & J_1  \end{array} \right| \left. \begin{array}{c} M_1 \\ d_1 \end{array} \right) \!\!                           \left( \begin{array}{c} d_1 \\ M_1
 \end{array} \right| \left. \begin{array}{cc} J_1 & L_1 \\ b'_1 & c_1 \end{array} \right) \!\!... \;
\Lambda^{J_4M_4}_{K_4L_4}(\alpha_4) \!
 \left( \begin{array}{cc} a_4 & b_4 \\
                          K_4 & J_4  \end{array} \right| \left. \begin{array}{c} M_4 \\ d_4 \end{array} \right)  \!\!                          \left( \begin{array}{c} d_4 \\ M_4
 \end{array} \right| \left. \begin{array}{cc} J_4 & L_4 \\ b'_4 & c_4 \end{array} \right),
 \nn
\end{align}
where we  inserted the definition \eqref{projection} of the $U_q(\su(2))$-intertwiner $f_{\alpha}^K$ to derive the second expression from the first.
The next step is to compute the integrals $h_1$ and $h_2$. This is done by use of the expressions \eqref{dualuqrels}, \eqref{dualAN2rels} for the multiplication on $F_q(\SU(2))$ and $F_q(\mathrm{AN}(2))$, respectively, together with the definition of the integrals \eqref{integraleuc}. This yields for the integral on  $F_q(\mathrm{AN}(2))$
\beq
\label{intAN}
h_2(E^{\;\;b'_1}_{J_1 \;\; b_1} \; ... \;  E^{\;\;b'_4}_{J_4 \;\; b_4}) = \delta_{J_1 J_2} \delta_{J_1 J_3}  \delta_{J_3 J_4} \delta_{b_1}^{b'_2} \delta_{b_3}^{b'_4}  \delta_{b_2}^{b'_3}  [2J_1+1] \pi_{J_1}(\mu^{-1})^{b'_1}_{\;\; b_4}.
\eeq
The Haar integral on $F_q(\SU(2))$ is computed analogously. The first step is to derive the identity \ros
\begin{align}
h_1(u^{\;\;a'_1}_{K_1 \;\; a_1} \; ... \; u^{\;\;a'_4}_{K_4 \;\; a_4}) =  \sum_{M,P} 
&\left( \begin{array}{cc} a'_1 & a'_2 \\
                     K_1 & K_2  \end{array} \right| \left. \begin{array}{c} M \\ a' \end{array} \right)                            \left( \begin{array}{c} a \\ M
 \end{array} \right| \left. \begin{array}{cc} K_1 & K_2 \\ a_1 & a_2  \end{array} \right)  \nn \\
 &\left( \begin{array}{cc} a'_3 & a'_4 \\
                          K_3 & K_4  \end{array} \right| \left. \begin{array}{c} P \\ b' \end{array} \right)                            \left( \begin{array}{c} b \\ P
 \end{array} \right| \left. \begin{array}{cc} K_3 & K_4 \\ a_3 & a_4 \end{array} \right) h_1(u^{\;\;a'}_{M\;\; a} u^{\;\;b'}_{P \;\; b}). \nn
\end{align}
It is then easy to compute the bivalent integral 
$$
h_1(u^{\;\;a'}_{M \;\; a} u^{\;\;b'}_{pP \;\; b}) = \frac{\delta_{MP}}{[2P +1]} \, e^{2i \pi P} \, v_P \, \epsilon_{P \, ab} \, \epsilon^{P \, a'b'},
$$
where we  used the identity
$$
\left( \begin{array}{c} 0 \\ 0
 \end{array} \right| \left. \begin{array}{cc} I & J \\ a & b \end{array} \right) = \frac{\delta_{IJ}}{\sqrt{[2J+1]}} (-1)^{J-a} q^a \delta_{a,-b} = \frac{\delta_{IJ}}{\sqrt{[2J +1]}} \, e^{i \pi J} \, v_J^{1/2} \, \epsilon_{J ab},
$$
to obtain \ros
\beq
\label{intSU}
h_1(u^{\;\;a'_1}_{K_1 \;\; a_1} \; ... \; u^{\;\;a'_4}_{K_4 \;\; a_4}) = \sum_P  \frac{v_P}{[2P +1]} \, e^{2i \pi P} \, \left( \begin{array}{cccc} K_1 & K_2 & K_3 & K_4 \\ 
 a_1 & a_2 & a_3 & a_4 \end{array} \right)_{P} \left( \begin{array}{cccc} a'_1 & a'_2 & a'_3 & a'_4 \\ 
 K_1 & K_2 & K_3 & K_4 \end{array} \right)_{P}. 
\eeq
The expression for the evaluation of the EPRL intertwiner then takes the form \ros
\begin{align}
\iota_{\alpha,N} \, (e^L_{c}[\alpha]) = \sum_{I,P} &\sum_{M_1,...,M_4} \frac{[2I+1]}{[2P+1]} v_P \, e^{2i\pi P} \, \pi_I(\mu^{-1})^{b'_1}_{\;\; b_4} \Lambda^{I M_1}_{K_1L_1}(\alpha_1) \Lambda^{I M_2}_{K_2L_2}(\alpha_2)  \Lambda^{I M_3}_{K_3L_3}(\alpha_3)  \Lambda^{I M_4}_{K_4L_4}(\alpha_4)\; 
 \nn \\
 &\left( \begin{array}{cccc} K_1 & K_2 & K_3 & K_4 \\ 
 a'_1 & a'_2 & a'_3 & a'_4 \end{array} \right)_{N} 
 \left( \begin{array}{cccc} a'_1 & a'_2 & a'_3 & a'_4 \\  K_1 & K_2 & K_3 & K_4 \end{array} \right)_{P}  \left( \begin{array}{cccc} K_1 & K_2 & K_3 & K_4 \\ 
 a_1 & a_2 & a_3 & a_4 \end{array} \right)_{P}\nn \\
&
 \left( \begin{array}{cc} a_1 & b_1 \\
                          K_1 & I  \end{array} \right| \left. \begin{array}{c} M_1 \\ d_1 \end{array} \right)                            \left( \begin{array}{c} d_1 \\ M_1
 \end{array} \right| \left. \begin{array}{cc} I & L_1 \\ b'_1 & c_1 \end{array} \right) \left( \begin{array}{cc} a_2 & b_2 \\
                          K_2 & I  \end{array} \right| \left. \begin{array}{c} M_2 \\ d_2 \end{array} \right)                            \left( \begin{array}{c} d_2 \\ M_2
 \end{array} \right| \left. \begin{array}{cc} I & L_2 \\ b_1 & c_1 \end{array} \right) 
 \nn \\
&\left( \begin{array}{cc} a_3 & b_3 \\
                          K_3 & I  \end{array} \right| \left. \begin{array}{c} M_3 \\ d_3 \end{array} \right)\left( \begin{array}{c} d_3 \\ M_3
 \end{array} \right| \left. \begin{array}{cc} I & L_3 \\ b_2 & c_3 \end{array} \right) 
\left( \begin{array}{cc} a_4 & b_4 \\
                          K_4 & I  \end{array} \right| \left. \begin{array}{c} M_4 \\ d_4 \end{array} \right)                            \left( \begin{array}{c} d_4 \\ M_4
 \end{array} \right| \left. \begin{array}{cc} I & L_4 \\ b_3 & c_4 \end{array} \right) . \nn
\end{align}
Using the orthogonality property \eqref{orthogonality4} of the Clebsch-Gordan coefficients and simplifying the resulting expressions, we then obtain  equation \eqref{evalL}.

To prove the absolute convergence of the series, we bound the absolute value of the summands by the summands of a convergent series. We start be expressing \eqref{evalL} as \ros
\beq
\label{series}
\iota_{\alpha,N} \, (e^L_{c}[\alpha]) = c(K) \sum_{I} \sum_{M_1,...,M_4} S(I,M_1, ... ,M_4),
\eeq
where $c(K)$ is a function of the  labels  $K$ given by
$$
c(K)= \left( \begin{array}{c} a \\ N
 \end{array} \right| \left. \begin{array}{cc} K_1 & K_2 \\ a_1 & a_2 \end{array} \right) \epsilon_{N ab} \left( \begin{array}{c} b \\ N
 \end{array} \right| \left. \begin{array}{cc} K_3 & K_4 \\ a_3 & a_4 \end{array} \right),
 $$ 
 and the summand takes the form \ros
\begin{align}
\label{summand}
S(I,M_1, ... ,M_4) = 
C(I,M_1,...,M_4) [2I+1] \Lambda^{I M_1}_{K_1L_1}(\alpha_1) \Lambda^{I M_2}_{K_2L_2}(\alpha_2)  \Lambda^{I M_3}_{K_3L_3}(\alpha_3)  \Lambda^{I M_4}_{K_4L_4}(\alpha_4),
\end{align}
with
\begin{align}
&C(I,M_1,...,M_4) = \sum_{b_1,..., b_4} \sum_{d_1,...,d_4} q^{-2 b_4}
 \left( \begin{array}{cc} a_1 & b_1 \\
                          K_1 & I  \end{array} \right| \left. \begin{array}{c} M_1 \\ d_1 \end{array} \right)                            \left( \begin{array}{c} d_1 \\ M_1
 \end{array} \right| \left. \begin{array}{cc} I & L_1 \\ b_4 & c_1 \end{array} \right) \left( \begin{array}{cc} a_2 & b_2 \\
                          K_2 & I \end{array} \right| \left. \begin{array}{c} M_2 \\ d_2 \end{array} \right)  \nn \\                          
&\left( \begin{array}{c} d_2 \\ M_2
 \end{array} \right| \left. \begin{array}{cc} I & L_2 \\ b_1 & c_1 \end{array} \right) 
\left( \begin{array}{cc} a_3 & b_3 \\
                          K_3 & I  \end{array} \right| \left. \begin{array}{c} M_3 \\ d_3 \end{array} \right)                            \left( \begin{array}{c} d_3 \\ M_3
 \end{array} \right| \left. \begin{array}{cc} I & L_3 \\ b_2 & c_3 \end{array} \right) \left( \begin{array}{cc} a_4 & b_4 \\
                          K_4 & I  \end{array} \right| \left. \begin{array}{c} M_4 \\ d_4 \end{array} \right)                            \left( \begin{array}{c} d_4 \\ M_4
 \end{array} \right| \left. \begin{array}{cc} I & L_4 \\ b_3 & c_4 \end{array} \right).
 \nn 
\end{align}
We can now bound the summand $S(I,M_1, ... ,M_4)$ of the series as follows. Using the fact that the Clebsch-Gordan coefficients are bounded by one and the identity
$\sum_{a=-I}^{I} q^{-2a} = [2I + 1]$, we  obtain 
$$
\mid C(I,M_1,...,M_4)  \mid \leq [2I +1] (2I+1)^3 (2M_1+1) (2M_2+1) (2M_3+1) (2M_4+1),
$$
We then use the asymptotic properties of the  coefficients  $\Lambda^{IM}_{KL}(\alpha)$ which are derived in \cite{PE2}. It is shown in \cite{PE2} that  there exist constants $\mathcal{J}, \mathcal{C} > 0$ such that
\begin{align}\label{lambdas}
 \mid \Lambda^{I I+N}_{K L} \mid \leq \mathcal{C} I q^{2I}\qquad \forall I > \mathcal{J}, K,L,N\in\NN_0/2.
\end{align}
This implies that  for $I$ sufficiently large, $I>\mathcal J$, we have the following bound on  the summands \eqref{summand} of the series \eqref{series}:
$$
\mid S(I,I+N_1,I+N_2,I+N_3,I+N_4) \mid \leq 
\mathcal{K} I^{11} q^{4I}\qquad \forall N_1,N_2,N_3,
$$
where $\mathcal{K}$ is a constant and we have used the asymptotic identities $\lim_{I\to \infty}[2I+1] /q^{-2I}=1$ and \eqref{lambdas}. As $q\in]0,1[$, the series $\sum_{I=0}^\infty   I^{\alpha} q^{\beta I}$  converges  for all $\beta>0$ 
and the  series \eqref{series} converges absolutely.  By using the left invariance of the Haar measure, one then finds
 that the quantum EPRL intertwiner is an element of $\Hom_{D(U_q(\su(2)))}(V[\alpha],\C)$. \hfill $\square$

To conclude our discussion of the quantum EPRL intertwiner, we investigate its transformation under  braiding. As the following proposition shows, unlike the quantum Barrett-Crane intertwiner \cite{PK}, the quantum EPRL intertwiner is not invariant under braiding. For simplicity, we will concentrate on the braiding of the two first arguments. The general case can be treated analogously.

\begin{proposition}

Let    $K=(K_1,...,K_4)$, $K_i\in\mathcal L$,  be a $4$-tuple of of irreducible representations of $U_q(\su(2))$ with associated EPRL representations
$\alpha = (\alpha_1(K_1), ..., \alpha_4(K_4))$. 
Consider the $EPRL$ intertwiner $\iota_{\alpha,J}=f^*(\lambda_{K,J})$  for  an element   $\lambda_{K,J}\in \Hom_{U_q(\su(2))}(V[K], \C)$. Then the transformation of the EPRL intertwiner $\iota_{\alpha,J}$ under  the braiding $c_{\alpha_2, \alpha_1}=\tau_{\alpha_2,\alpha_1} \circ (\pi_{\alpha_2} \otimes \pi_{\alpha_1} )(R): V_{\alpha_2(K_2)}\otimes V_{\alpha_1(K_1)}\rightarrow V_{\alpha_1(K_1)}\otimes V_{\alpha_2(K_2)}$ of the first two EPRL representations is given by
\footnote{Note the abuse of notation in the right hand side of equation \eqref{br}; $\alpha_1(K)$ is not necessarily an EPRL 
representation. The notation is used nevertheless for compacity purposes.} \ros
\beq \label{br}
\iota_{\alpha_1(K_1) ... \alpha_4(K_4), J} \circ c_{\alpha_2,\alpha_1} =  \sum_K \Lambda^{K_2 J}_{K_1 K} \; \iota_{\alpha_2(K_2) \alpha_1(K) \alpha_3(K_3) \alpha_4(K_4)}.
\eeq
\end{proposition}

{\em Proof.} As the braiding is an intertwiner between the representations of the quantum Lorentz group on $V_{\alpha_2(K_2)}\otimes V_{\alpha_1(K_1)}$ and $V_{\alpha_1(K_1)}\otimes V_{\alpha_2(K_2)}$, we have \ros
\beqa
\iota_{\alpha,J} \circ c_{\alpha_2,\alpha_1}
&=& \sum_{A}  \lambda_{K,J} \circ  f^{K}_{\alpha} \circ c_{\alpha_2,\alpha_1}  \circ \left( \bigotimes_{i=\sigma(1)}^{\sigma(4)} \pi_{\alpha_i(K_i)} (\Delta^{(3)}(x^A)) \right)  h(x_{A}), \nn 
\eeqa
where $\sigma: (1,2,3,4)\rightarrow (2,1,3,4)$ is the permutation that exchanges
the first and second argument. 
Inserting expression \eqref{repR} for the action of the $R$-matrix on the principal representations into the definition of the braiding and using the definition \eqref{projection} of the projector $f^K_\alpha$ we obtain \ros
\begin{align}\label{intid}
\iota_{\alpha,J} \circ c_{\alpha_2,\alpha_1} (\otimes_{i=1}^{4} &e^{L_i}_{c_i}(\alpha_i)) 
\!= \!\!\!\!\!\!\sum_{A_1,...,A_4} \sum_{L'_1,N}\!\! \lambda_{K,J} (e^{K_1}_{f} \!\otimes\! e^{K_2}_{e} \!\otimes\! e^{K_3}_{c'_3} \!\otimes\! e^{K_4}_{c'_4})  \left(\!\! \begin{array}{cc} f & e \\
                          K_1 & K_2  \end{array} \right| \left. \begin{array}{c} N \\ g \end{array} \!\!\right) \!\!                           \left( \!\!\begin{array}{c} g \\ N
 \end{array} \right| \left. \begin{array}{cc} K_2 & L'_1 \\ c'_2 & c'_1 \end{array} \!\!\right) \; \nn \\
& \Lambda^{K_2 N}_{K_1 L'_1} \pi_{\alpha_2}(x^{A_2})^{K_2 c'_2}_{\;\;\; L_2 c_2}   \pi_{\alpha_1}(x^{A_1})^{L'_1c'_1}_{\;\;\; L_1 c_1}  \pi_{\alpha_3}(x^{A_3})^{K_3 c'_3}_{\;\;\; L_3 c_3}    \pi_{\alpha_4}(x^{A_4})^{K_4 c'_4}_{\;\;\; L_4 c_4}  h(x_{A_2} x_{A_1} ... x_{A_4}). \nn
\end{align}
We now evaluate this expression by using identity \eqref{lambdaikdef} which expresses the $U_q(\su(2))$ intertwiner $\lambda_{K,J}$  in terms of the Clebsch-Gordan coefficients. Due to  the orthogonality of the Clebsch-Gordan maps \eqref{orthogonality}, we can then recombine two coefficients and rewrite expression \eqref{intid} as \ros
\begin{align}
&\iota_{\alpha,J} \circ c_{\alpha_2,\alpha_1} (\otimes_{i=1}^{4} e^{L_i}_{c_i}(\alpha_i)) \\
&=\sum_{A_1,...,A_4} \sum_{L'_1}    \left( \begin{array}{c} c \\ J
 \end{array} \right| \left. \begin{array}{cc} K_2 & L'_1 \\ c'_2 & c'_1 \end{array} \right) \; \epsilon_{Jcd} \left( \begin{array}{c} d \\ J
 \end{array} \right| \left. \begin{array}{cc} K_3 & K_4 \\ c'_3 & c'_4 \end{array} \right) \Lambda^{K_2 J}_{K_1 L'_1}  \nn \\
& \qquad\qquad\qquad\pi_{\alpha_2}(x^{A_2})^{K_2 c'_2}_{\;\;\; L_2 c_2}   \pi_{\alpha_1}(x^{A_1})^{L'_1c'_1}_{\;\;\; L_1 c_1}  \pi_{\alpha_3}(x^{A_3})^{K_3 c'_3}_{\;\;\; L_3 c_3}    \pi_{\alpha_4}(x^{A_4})^{K_4 c'_4}_{\;\;\; L_4 c_4}  h(x_{A_2} x_{A_1} ... x_{A_4})  \nn \\
& = \sum_{A_1,...,A_4} \sum_{K}    
\left( \begin{array}{cccc} K_2 & K & K_3 & K_4 \\ 
 c'_2 & c'_1 & c'_3 & c'_4 \end{array} \right)_J \Lambda^{K_2 J}_{K_1 K}  \nn \\
& \qquad\qquad\qquad\pi_{\alpha_2}(x^{A_2})^{K_2 c'_2}_{\;\;\; L_2 c_2}   \pi_{\alpha_1}(x^{A_1})^{K c'_1}_{\;\;\; L_1 c_1}  \pi_{\alpha_3}(x^{A_3})^{K_3 c'_3}_{\;\;\; L_3 c_3}    \pi_{\alpha_4}(x^{A_4})^{K_4 c'_4}_{\;\;\; L_4 c_4}  h(x_{A_2} x_{A_1} ... x_{A_4}). \nn
\end{align}
This proves the claim.
\hfill  $\square$

\subsection{The four-simplex amplitude}
\label{qamp}

We are now ready to construct the amplitude for $4$-simplexes labeled by EPRL representations and $q$-EPRL intertwiners. 
Such an amplitude is defined with the aid of the graphical calculus of spin networks.
There are two main difficulties in the definition of the amplitude that arise from the fact that the representation spaces 
of the EPRL representations are infinite-dimensional. The first is that there is no coevaluation map that intertwines 
the trivial representation of the quantum Lorentz group on $\C$ with a representation on the tensor product 
$V_\alpha\otimes V_\alpha^*$ and therefore no notion of a quantum trace. The second difficulty is that a naive definition 
of the amplitude for the four-simplexes gives an infinite answer.

These problems  arise in a similar fashion in the classical Lorentzian BC \cite{BC2} and EPRL \cite{eprlmod} models. A
solution to the first problem was provided in \cite{notts3} where a Lorentzian graphical calculus based on non-invariant
tensors and bilinear forms was invented. A regularisation prescription that circumvents the second difficulty has been 
given in
\cite{BB} and \cite{roberto} for the BC and EPRL models, respectively. 
Extending these procedures to the quantum Lorentz group, we will overcome these issues and consistently construct a finite 
amplitude for the $4$-simplexes. 

\subsubsection{EPRL tensors and bilinear form on the representation spaces}

We first need a notion of dual quantum EPRL intertwiners. These dual objects are required since we cannot pair 
$q$-EPRL intertwiners together because of the absence of a coevaluation map.
As in the Lie group case, the vector space  
$\bigotimes_{i=1}^n V_{\alpha_i}$ does not contain 
tensors that are invariant under the action of the quantum Lorentz group and such objects do not exist per se. Accordingly, 
dual $q$-EPRL intertwiners  are replaced by non-invariant quantities which can be viewed as the quantum group analogue of  
the boosted $\SU(2)$ intertwiners considered in \cite{notts3}. These quantities, which are referred to as vertex functions in 
\cite{PK},  will be called EPRL tensors in the following.

\begin{definition} {\bf (Quantum EPRL tensor)}
Let $K = (K_1, ..., K_n)$ be a $n$-tuple of representations of $U_q(\su(2))$ labelled by elements of $\mathcal{L}$. Denote by  $\alpha= (\alpha_1(K_1),...,\alpha_n(K_{n}))$ the associated $n$-tuple of EPRL representations, 
and consider an element  $\Lambda^{K} \in \mathrm{Hom}_{U_q(\su(2))}(\C, V[K])$.  The quantum EPRL tensor $\Psi^{\alpha}$ associated to $\Lambda^K$ is  defined by 
\beq
\label{tensor}
\Psi^{\alpha} = \left[ \sum_{A}  \left( \bigotimes_{i=1}^n \pi_{{\alpha}_{i(K_i)}} \left( \Delta^{(n-1)}(x^A) \right)\right) \circ f_{K}^{\alpha} \circ \Lambda^{K} \right] \otimes x_{A},
\eeq
where $\{x_A\}_{A}$ is a basis of $F_q(\SL(2,\C)_\RR)$ and  $\{x^A\}_{A}$ is the dual basis
of $F_q(\SL(2,\C)_\RR)^*$. The vector space of EPRL tensors associated with $\alpha$ is the vector space  $H[\alpha] := \mathcal{L}(\C, V[\alpha]) \otimes F_q(\SL(2,\C)_{\mathbb{R}})$,
where $\mathcal{L}(\C, V[\alpha])$ denotes the space of linear maps from $\C$ to $V[\alpha]$.
\end{definition}

The quantum EPRL tensors associated with the  elements 
of  the basis $\{\lambda_{K,J}\}_{J \in \mathbb{N}_0/2}$ of \linebreak $\Hom_{U_q(\su(2))}(\C, V[\alpha])$ will be denoted $\Psi^{\alpha , J}$ in the following.

To pair EPRL tensors together, we need a bilinear form on the representation spaces $V_\alpha$. This form is induced by an
intertwiner between the representation spaces $V_\alpha$ and their duals. 

\begin{lemma} 
\label{lemma}
Let $V_{\alpha}^* = \bigoplus_{I=m}^\infty V_I^*$ be the dual to the vector space $V_{\alpha}$.
There exist a bijective intertwiner $\phi^{\alpha} : V_{\alpha} \rightarrow V_{\alpha}^*$ whose expression with respect to the basis $\{e^I_a\}_{I,a}$ of $V_{\alpha}$ and the dual basis $\{e^{Ia}\}_{I,a}$ of $V_{\alpha}^{*}$ is given by
\beq 
\label{Lorint}
\phi^\alpha (e^I_a)=  \phi^{\alpha}_{IaJb} e^{Jb} \qquad \phi^{\alpha}_{IaJ_b} = c_{\alpha} v_I^{1/2} \delta_{IJ} \epsilon_{I ba},
\eeq
where  $c_{\alpha}$ is a constant and $\epsilon_{Iab}$ is given by \eqref{epsexpl}.
The representations of the quantum Lorentz group on  $V_\alpha$ and $V_\alpha^*$ are equivalent. The bilinear form $\beta_{\alpha} : V_{\alpha} \otimes V_{\alpha} \rightarrow \C$ 
\beq
\label{bilinear}
 \beta_{\alpha}(v,w) = \phi^{\alpha}(w)(v)\qquad \forall v,w \in V_{\alpha}, 
\eeq
satisfies the invariance property $\beta_{\alpha}( v , \pi_{\alpha}(a) w) = \beta_{\alpha}( \pi_{\alpha}(S(a)) v ,  w)$ for all $a\in D(U_q(\su(2)))$.
\end{lemma}

{\em Proof:}  We first show that the map $\phi_\alpha: V_\alpha\rightarrow V_\alpha^*$ is an intertwiner between the representation of the quantum Lorentz group on $V_\alpha$ and on its dual $V_\alpha^*$. 
By definition, the latter is given by the antipode of the quantum Lorentz group
$$\pi_{\alpha}^*(a) f = f \circ \pi_{\alpha}(S(a))\qquad \forall a \in D(U_q(\su(2))), f \in V_{\alpha}^*.$$ In terms of the basis $(X^{\;\;a}_{I \;\; a'} \otimes g^{\;\;b}_{J \;\; b'})_{I,J,a,a',b,b'}$ defined after \eqref{dualAN2rels} and the bases  $\{e^I_a\}_{I,a}$ and $\{e^{Ia}\}_{I,a}$ of $V_\alpha$ and $V_\alpha^*$,
the condition that $\phi_\alpha: V_\alpha\rightarrow V^*_\alpha$ defines an intertwiner 
takes the form \ros
\beq
\label{condition}
\sum_{J,b} \phi^{\alpha }_{IaJb} \, \pi_{\alpha}(S(X^{\;\;d}_{L \;\; e} \otimes g^{\;\;f}_{M \;\; g}))^{Jb}_{\;\;\; Kc} = \sum_{J,b}  \pi_{\alpha}(X^{\;\;d}_{L \;\; e} \otimes g^{\;\;f}_{M \;\; g})^{Jb}_{\;\;\; Ia} \, \phi^{\alpha}_{JbKc}.
\eeq
To prove this identity, we compute the matrix elements $\pi_\alpha(X^{\;\;d}_{L \;\; e} \otimes g^{\;\;f}_{M \;\; g})$ using 
\eqref{repQLG}. This yields \ros
\beq
\label{matrixelementsrepQLG}
\pi_{\alpha}(X^{\;\;d}_{L \;\; e} \otimes g^{\;\;f}_{M \;\; g})^{Jb}_{\;\;\; Ia} = \delta^b_e \delta_L^J \; \sum_N \left( \begin{array}{cc} d & f \\
                          L & M  \end{array} \right| \left. \begin{array}{c} N \\ h \end{array} \right)                            \left( \begin{array}{c} h \\ N
 \end{array} \right| \left. \begin{array}{cc} M & I \\ g & a \end{array} \right) \Lambda^{MN}_{JI}(\alpha).
\eeq
To compute the matrix elements of  $S(X^{\;\;d}_{L \;\; e} \otimes g^{\;\;f}_{M \;\; g})$ in $\End V_{\alpha}$, we use 
 that the antipode is an  anti-algebra morphism together with expression \eqref{coFq} for  the antipode on $F_q(\SU(2))$. From the identification $D(U_q(\su(2)))=F_q(\SU(2))\hat \otimes F_q(AN(2))$ and the duality between $F_q(AN(2))$ and $F_q(\SU(2))^{op}$, it follows that the  antipode on $D(U_q(\su(2)))$ takes the form
\begin{align}
&S(X^{\;\;d}_{L \;\; e} \otimes g^{\;\;f}_{M \;\; g}) =(1 \otimes S^{-1} (g^{\;\;f}_{M \;\; g}))\cdot (S(X^{\;\;d}_{L \;\; e}) \otimes 1)\\
&S^{-1}(g^{\;\;f}_{M \;\; g}) = \epsilon_{M ig} \, g^{\;\;i}_{M \;\; j} \, \epsilon_M^{-1fj}, \;\;\;\;\;\; S(X^{\;\;d}_{N \;\; e}) = \epsilon_{N ke} \, X^{\;\;k}_{L \;\; l} \, \epsilon_L^{-1dl}.\nonumber
\end{align}
From  \eqref{repQLG}  we then obtain that its matrix elements  are given by \ros
\begin{align}
&\pi_{\alpha}(S(X^{\;\;d}_{L \;\; e} \otimes g^{\;\;f}_{M \;\; g}))^{Jb}_{\;\;\; Kc} = \epsilon_{Lce} \, \epsilon_L^{-1dn} \, \epsilon_{M pg} \, \epsilon_M^{-1fq} \, 
\sum_N \left( \begin{array}{cc} b & p \\
                          J & M  \end{array} \right| \left. \begin{array}{c} N \\ h \end{array} \right)                            \left( \begin{array}{c} h \\ N
 \end{array} \right| \left. \begin{array}{cc} M & L \\ q & n \end{array} \right) \Lambda^{MN}_{JL}(\alpha)  
 \nn \\ 
 &=v_L^{1/2} e^{i \pi (g-d-f)} q^{-g+d+f} \, \epsilon_{Lce} \,
\sum_N \left( \begin{array}{cc} -f & -d \\
                          M & L  \end{array} \right| \left. \begin{array}{c} N \\ h \end{array} \right)                            \left( \begin{array}{c} h \\ N
 \end{array} \right| \left. \begin{array}{cc} J & M \\ b & -g \end{array} \right) \Lambda^{MN}_{JL}(\alpha),
\end{align}
where we  used the definition of the $U_q(\su(2))$-intertwiner $\epsilon$ given in \eqref{epsexpl} and the reality condition \eqref{real} on  the Clebsch-Gordan coefficients.
We can now rewrite the intertwining property \eqref{condition} as \ros
\beqa
\label{condition2}
&& \sum_{J,b } \phi^{\alpha }_{IaJb} \,  v_L^{1/2} e^{i \pi (g-d-f)} q^{-g+d+f}  \, \epsilon_{Lce} \,
\sum_N \left( \begin{array}{cc} -f & -d \\
                          M & L  \end{array} \right| \left. \begin{array}{c} N \\ h \end{array} \right)                            \left( \begin{array}{c} h \\ N
 \end{array} \right| \left. \begin{array}{cc} J & M \\ b & -g \end{array} \right) \Lambda^{MN}_{JK}(\alpha) \nn \\
 &=& \sum_N \left( \begin{array}{cc} d & f \\
                          L & M  \end{array} \right| \left. \begin{array}{c} N \\ h \end{array} \right)                            \left( \begin{array}{c} h \\ N
 \end{array} \right| \left. \begin{array}{cc} M & I \\ g & a \end{array} \right) \Lambda^{MN}_{LI}(\alpha) \, \phi^{\alpha}_{LeKc}.
\eeqa
This identity now follows by inserting the definition \eqref{Lorint}
 of $\phi^{\alpha}$. Using \eqref{epsexpl} and the permutation symmetry \eqref{orthogonality} of the Clebsch-Gordan coefficients, we obtain for the first expression in  \eqref{condition2} \ros
\beqa
&& \sum_{N,b} c_{\alpha} v_I^{1/2} v_K^{1/2} e^{i \pi (g-d-f)} q^{-g+d+f} \delta_{LK} \, \epsilon_{I ba} \, \epsilon_{Kce} \left( \begin{array}{cc} -f & -d \\
                          M & K  \end{array} \right| \left. \begin{array}{c} N \\ h \end{array} \right)                            \left( \begin{array}{c} h \\ N
 \end{array} \right| \left. \begin{array}{cc} I & M \\ b & -g \end{array} \right) \Lambda^{MN}_{IK}(\alpha) \nn \\
 && = \sum_N c_{\alpha} v_K^{1/2} e^{i \pi (g-d-f+a)} q^{-g+d+f-a} \delta_{LK} \, \epsilon_{Kce} \left( \begin{array}{cc} -f & -d \\
                          M & K  \end{array} \right| \left. \begin{array}{c} N \\ h \end{array} \right)                            \left( \begin{array}{c} h \\ N
 \end{array} \right| \left. \begin{array}{cc} I & M \\ -a & -g \end{array} \right) \Lambda^{MN}_{IK}(\alpha) \nn \\
 && = \sum_N c_{\alpha} v_K^{1/2} e^{i \pi (g-d-f+a)} q^{-g+d+f-a} \delta_{LK} \, \epsilon_{Kce} \left( \begin{array}{cc} d & f \\
                          K & M  \end{array} \right| \left. \begin{array}{c} N \\ -h \end{array} \right)                            \left( \begin{array}{c} -h \\ N
 \end{array} \right| \left. \begin{array}{cc} M & I \\ g & a \end{array} \right) \Lambda^{MN}_{IK}(\alpha), \nn
\eeqa
Note that this expression vanishes unless $d+f = g+a$ in which case it reduces to 
$$
\sum_N c_{\alpha} v_K^{1/2} \delta_{LK} \, \epsilon_{Kce} \left( \begin{array}{cc} d & f \\
                          K & M  \end{array} \right| \left. \begin{array}{c} N \\ h \end{array} \right)                            \left( \begin{array}{c} h \\ N
 \end{array} \right| \left. \begin{array}{cc} M & I \\ g & a \end{array} \right) \Lambda^{MN}_{IK}(\alpha).
$$
This coincides with the expression obtained by inserting \eqref{Lorint} into the second expression in\eqref{condition2} for
 $d+f = g+a$. Otherwise, the second expression in \eqref{condition2} vanishes. This proves that $\Phi^\alpha: V_\alpha\rightarrow V^*_\alpha$ defines an intertwiner between the representations of the quantum Lorentz group on $V_\alpha$ and on its dual $V_\alpha^*$. 
The second part of the lemma then follows  directly from the definitions and from the intertwining property of $\phi^{\alpha}$. \hfill $\square$

In the following, we will set the constant $c_{\alpha}$ to one for all $\alpha$.

\subsubsection{Graphical calculus}

We  will now construct the four-simplex amplitude from EPRL tensors.
To define the amplitude, it  is convenient to use diagrammatic methods. More precisely, we will make use of the graphical calculus of spin networks. The idea behind this calculus is to organise tensor calculations according to a diagram drawn in the plane. To cover calculations based on the representation theory of the quantum Lorentz group, some constructions used routinely  in the representation theory of finite-dimensional Hopf algebras do not work and the calculus is restricted to a certain class of diagrams. The elements needed for the construction of the four-simplex amplitude 
are as follows.

We denote by oriented dashed lines irreducible representations of $U_q(\su(2))$ and by oriented solid  lines EPRL representations. The tensor product of representations is depicted by drawing the lines next to each other from the left to the right. The Clebsch-Gordan maps for $U_q(\su(2))$ are depicted by three-valent vertices involving three dashed lines as shown in Figure \ref{beta} a).  The inclusion map $f^\alpha_K: V_K\rightarrow V_{\alpha(K)}$ is depicted by a black box as shown in Figure \ref{beta} b).
We choose the conventions that diagrams are evaluated from bottom to top.
\begin{figure}
  \includegraphics[scale=0.4]{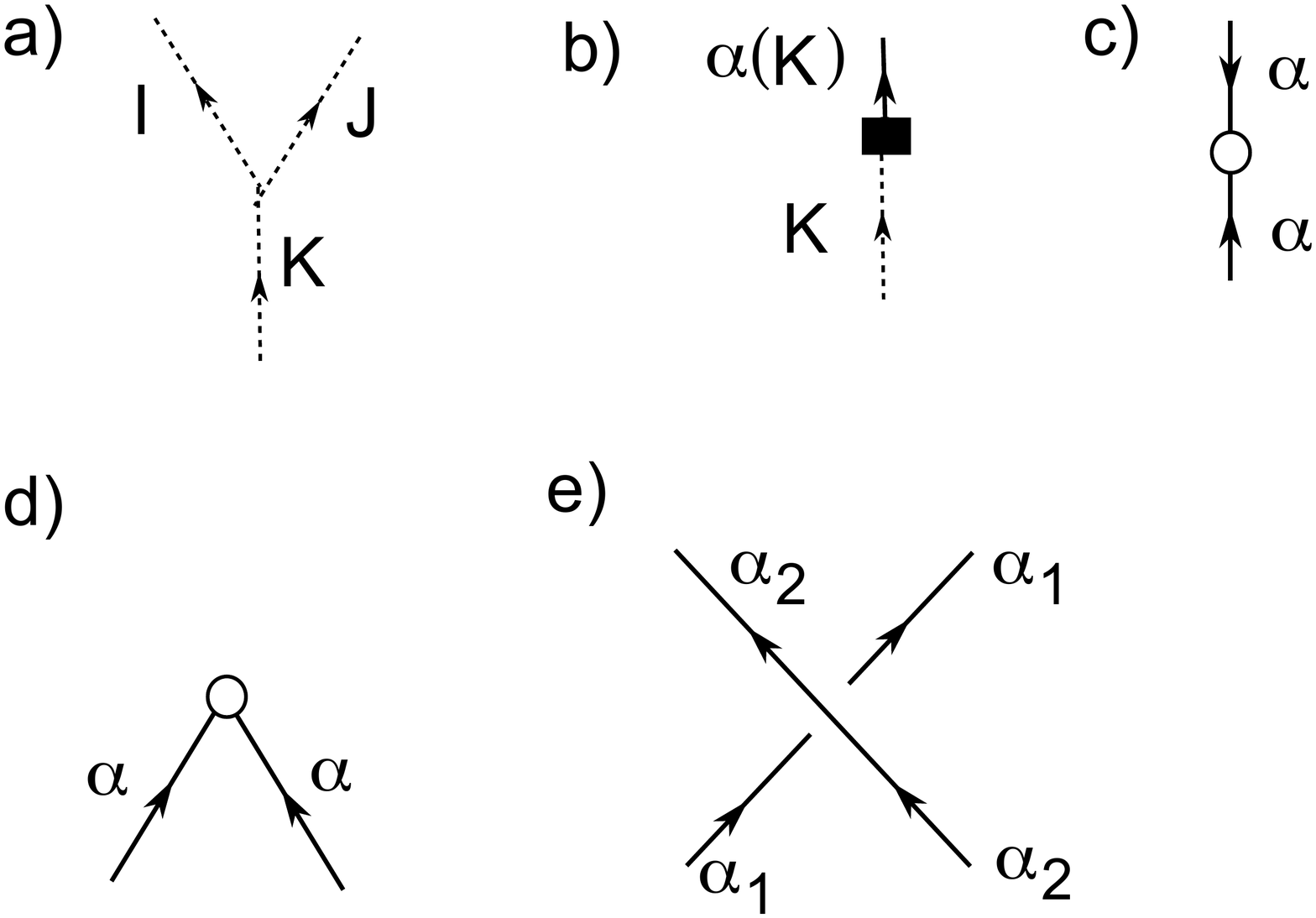}
\centering  \caption{ Diagram components:\newline
a) Clebsch-Gordan intertwiner $C_{K}^{IJ}: V_K\rightarrow V_I\otimes V_J$ for $U_q(\su(2))$.\newline
b) Inclusion map $f^\alpha_K: V_K\rightarrow V_{\alpha(K)}$.\newline
c) Intertwiner $\phi_\alpha: V_\alpha\rightarrow V_\alpha^*$.\newline
d) Bilinear form $\beta_\alpha: V_\alpha\otimes V_\alpha\rightarrow \C$.\newline
e) Braiding $c_{\alpha_2,\alpha_1}: V_{\alpha_2}\otimes V_{\alpha_1}\rightarrow V_{\alpha_1}\otimes V_{\alpha_2}$.
}
  \label{beta}
\end{figure}
With this notation, it follows that the diagram for a EPRL tensor $\Psi^{\alpha}$ is given by Figure \ref{eprltensor} a) and the one for an EPRL intertwiner $\iota_\alpha$ by Figure \ref{eprltensor} c). 
To keep the diagrams simple, we also introduce a shorthand notation, in which a EPRL tensor is depicted by a solid vertex
 $v$ with four  `legs' pointing upwards, away from the vertex as shown in Figure \ref{eprltensor} b). The vertex is  decorated  with a cilium (short dashed line in Figure \ref{eprltensor} b)) that indicates in which order the representations are coupled to the trivial representation. This cilium is labelled by the intermediate $U_q(\su(2))$-representation $J$ as shown in Figure \ref{eprltensor} b).  The picture for an EPRL intertwiner is obtained by rotating the one for an EPRL tensor by 180 degree and inverting the orientation of the edges as shown in Figure \ref{eprltensor} c),d). 
 
 To each vertex we associate  a  basis element $x_A$ of $F_q(\SL(2,\C))$, which is omitted in the diagrams for reasons of legibility.   The diagram for the tensor product of $p$ EPRL  tensors $\Psi^{\alpha_1} \otimes ... \otimes \Psi^{\alpha_p}$ is obtained by  placing the $p$ vertices on a horizontal line, in the order in which they appear in the tensor product read from left to right. 

\begin{figure}
  \includegraphics[scale=0.4]{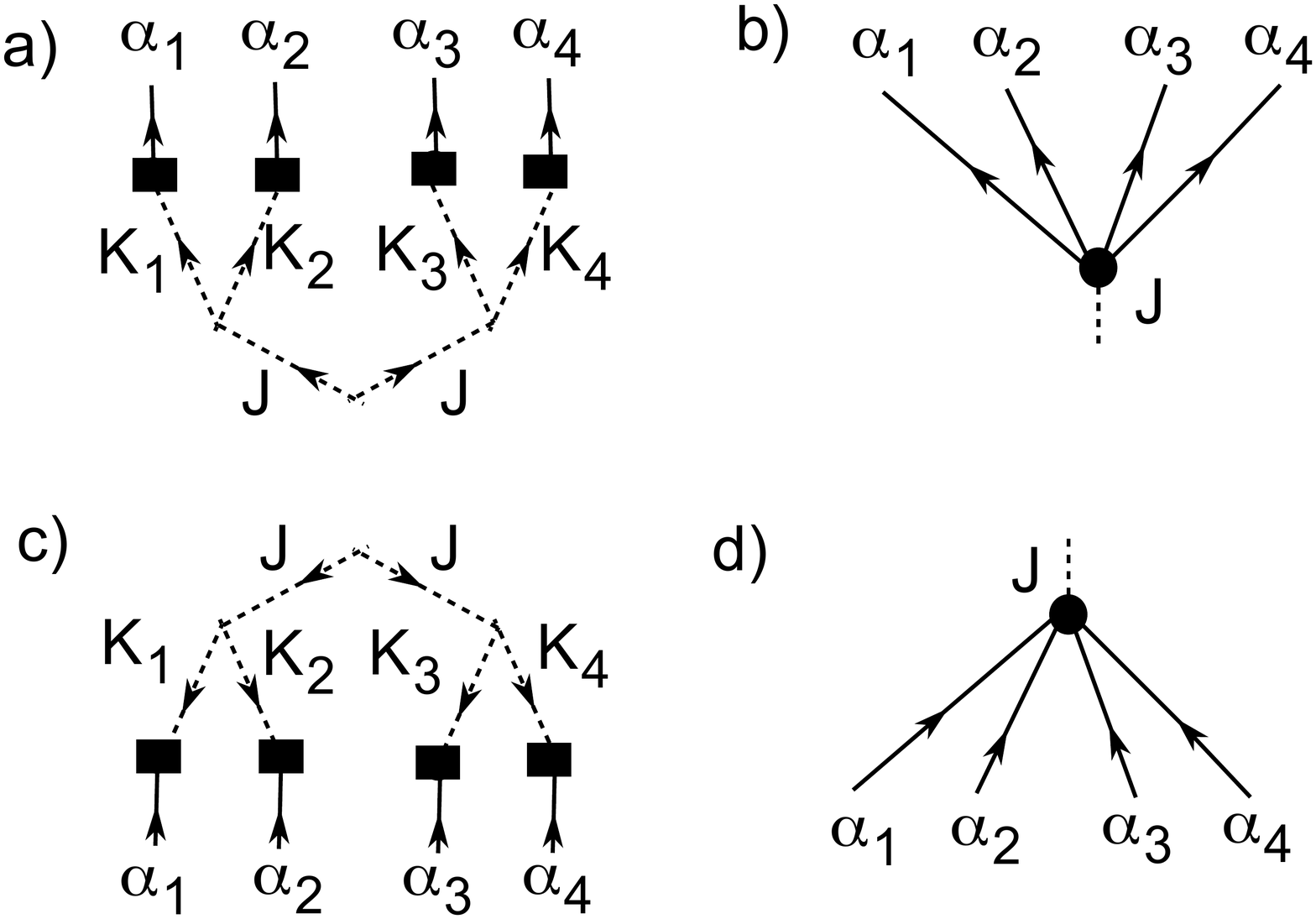}
\centering  \caption{$\quad$\newline
  a) EPRL tensor $\psi^\alpha\in\mathcal L( \C,\otimes_{i=1}^4 V_{\alpha_i})\otimes F_q(SL(2,\C)_\RR)$,  
  b) Short notation for the EPRL tensor $\psi^\alpha$.\newline
  c) EPRL intertwiner $\iota_\alpha:\otimes_{i=1}^4 V_{\alpha_i}\rightarrow\C$, 
  d) Short notation for the EPRL intertwiner $\iota_\alpha$.
  }
  \label{eprltensor}
\end{figure}

EPRL tensors can be paired using the invariant bilinear form $\beta_{\alpha} : V_{\alpha} \otimes V_{\alpha} \rightarrow \C$ defined in equation \eqref{bilinear}. In the diagrams, this bilinear form is depicted by a white circle with two legs, one to the left, one to the right and both pointing towards the vertex, as  shown in Figure \ref{beta} d). The precise details of the pairing are organised with the aid of the diagrammatic method.  
As the order of the two arguments of the bilinear form is important, the convention is that the first argument corresponds to the element on the left-hand side and the second argument to the right-hand side of the circle.

When more than one pair of edges are paired in a diagram,  crossings  can occur in the diagrams. To each such crossing we associate a braiding $c_{\alpha_2, \alpha_1} = \tau \circ( \pi_{\alpha_2} \otimes \pi_{\alpha_1})(R)$, where $R$ is the $R$-matrix of the quantum Lorentz group defined in \eqref{Rmatrix}. Again, the precise form of the crossing is important and the braiding $c_{\alpha_2, \alpha_1}$ will be associated to the crossing where the left-hand leg goes under the right-hand leg as shown in Figure \ref{beta} e).

The diagrams are composed horizontally by tensor product, and vertically by the composition of maps. In this context, a vertical line corresponds to the identity map on a representation space $V_{\alpha}$. Note, however,  that in contrast to the diagrams  for ribbon categories, the upward and downward arcs have no direct meaning. The rules are that lines go upwards the EPRL tensors  and  can only be paired by the bilinear form $\beta_\alpha$, but they are free to go up and down between the vertices. 

An important class of diagrams are closed diagrams. A closed diagram $\Gamma$ consists of a set of $p$ vertices arranged in a horizontal line, composed with crossings, vertical lines and pairings, in such a way that there are no  free ends. A closed diagram $\Gamma$  corresponds to an element $\phi(\Gamma)$ of $\End(\C) \otimes F_q(\SL(2,\C))^{\otimes p} \cong F_q(\SL(2,\C))^{\otimes p}$.
The evaluation $ev(\Gamma)$ of a closed diagram $\Gamma$ is then defined via the Haar integral in the spirit of Feynman diagram evaluations.

The naive evaluation of a closed diagram with $p$ vertices would correspond to setting  $ev(\Gamma)=h^{\otimes p} (\phi(\Gamma))$. However, such an evaluation is generically divergent for the Lorentzian model and needs to be regularised.  
This is done in analogy to the classical case \cite{BB,roberto} by removing the Haar measure or integration at one (randomly chosen) vertex as in \cite{PK}. The invariance of the Haar integral implies that the result is independent of the chosen vertex. Moreover, it  implies  that $(h^{\otimes p - 1} \otimes id) (\phi(\Gamma)) = ev(\Gamma) 1$, where $ev(\Gamma)$ is a complex number or infinity and 1 is the unit in $F_q(SL(2,\C))$.  
As we will show in the following, if $\Gamma$ is the complete graph with five vertices, $ev(\Gamma)$ is finite, i.~e.~the diagram  $\Gamma$ is integrable. The evaluation of $\Gamma$ is therefore obtained by applying $p-1$ copies of the Haar measure to $\Phi(\Gamma)$ and then applying the counit of $F_q(SL(2,\C))$  to the resulting expression
$$
ev(\Gamma) = \epsilon \left( (h^{\otimes p - 1} \otimes id) (\phi(\Gamma)) \right).
$$

\subsubsection{Amplitude for the $4$-simplexes}

\label{ampsec}

Let $M$ be an oriented, closed  triangulated $4$-manifold with sets of $n$-simplexes $\Delta^{(n)}$. Consider a $4$-simplex $\sigma$ 
of $M$. The set $\Delta_{\sigma}^{(3)}$ of tetrahedra of $\partial \sigma$ will be parametrised by $(a)$, $a=1,...,5$. 
Consequently, the triangles  $\Delta_{\sigma}^{(2)}$ of $\sigma$ will be labelled by ordered pairs $(ab)$, with $a<b$.
Let $\mathcal{C}$ denote the braided tensor category of representations of the quantum Lorentz group.

\begin{definition} {\bf (Colouring)}
A colouring $\alpha: \Delta_{\sigma}^{(2)} \rightarrow \mathrm{Ob}(\mathcal{C})$ associates an EPRL representation 
$\alpha_{ab} \in  \mathrm{Ob}(\mathcal{C})$ of $D(U_q(\su(2)))$ to each oriented triangle $(ab)$ in $\Delta_{\sigma}^{(2)}$. 
\end{definition} 

From the colourings, one can construct state spaces for each tetrahedron of $\partial \sigma$. The construction records 
the orientation of each tetrahedron in the boundary of $\sigma$
$$
\partial \sigma = (a) - (b) + (c) - (d) + (e).
$$
\begin{definition} {\bf (State space)}
Let $\alpha: \Delta_{\sigma}^{(2)} \rightarrow  \mathrm{Ob}(\mathcal{C})$ denote a colouring. The state space associated with a 
tetrahedron $(a)$ appearing with positive sign in $\partial \sigma$ is read out of the colouring of its boundary 
$\partial(a) = (ab) - (ac) + (ad) - (ae)$ and is defined by
$$
H_{a} = \mathcal{L} \left( \C, V_{\alpha_{ab}} \otimes V_{\alpha_{ad}} \otimes V_{\alpha_{ac}} \otimes V_{\alpha_{ae}} \right) \otimes F_q(\SL(2,\C)_{\mathbb{R}})
$$
Likewise, the state space for negative tetrahedra is given by 
$$
H_{a}^* =  \mathcal{L} \left( \C, V_{\alpha_{ae}} \otimes V_{\alpha_{ac}} \otimes V_{\alpha_{ad}} \otimes V_{\alpha_{ab}} \right) \otimes F_q(\SL(2,\C)_{\mathbb{R}})
$$
A state is an assignment of a EPRL tensor  $\Psi_{a}$ in either $H_{a}$ or $H_a^*$ to each tetrahedron $(a)$ in 
$\Delta_{\sigma}^{(3)}$.
\end{definition}

We are now ready to define the amplitude, or partition function, for the $4$-simplex $\sigma$, which is given by the evaluation 
of the diagram in   Figure \ref{amplitude}.
\begin{figure}
  \includegraphics[scale=0.4]{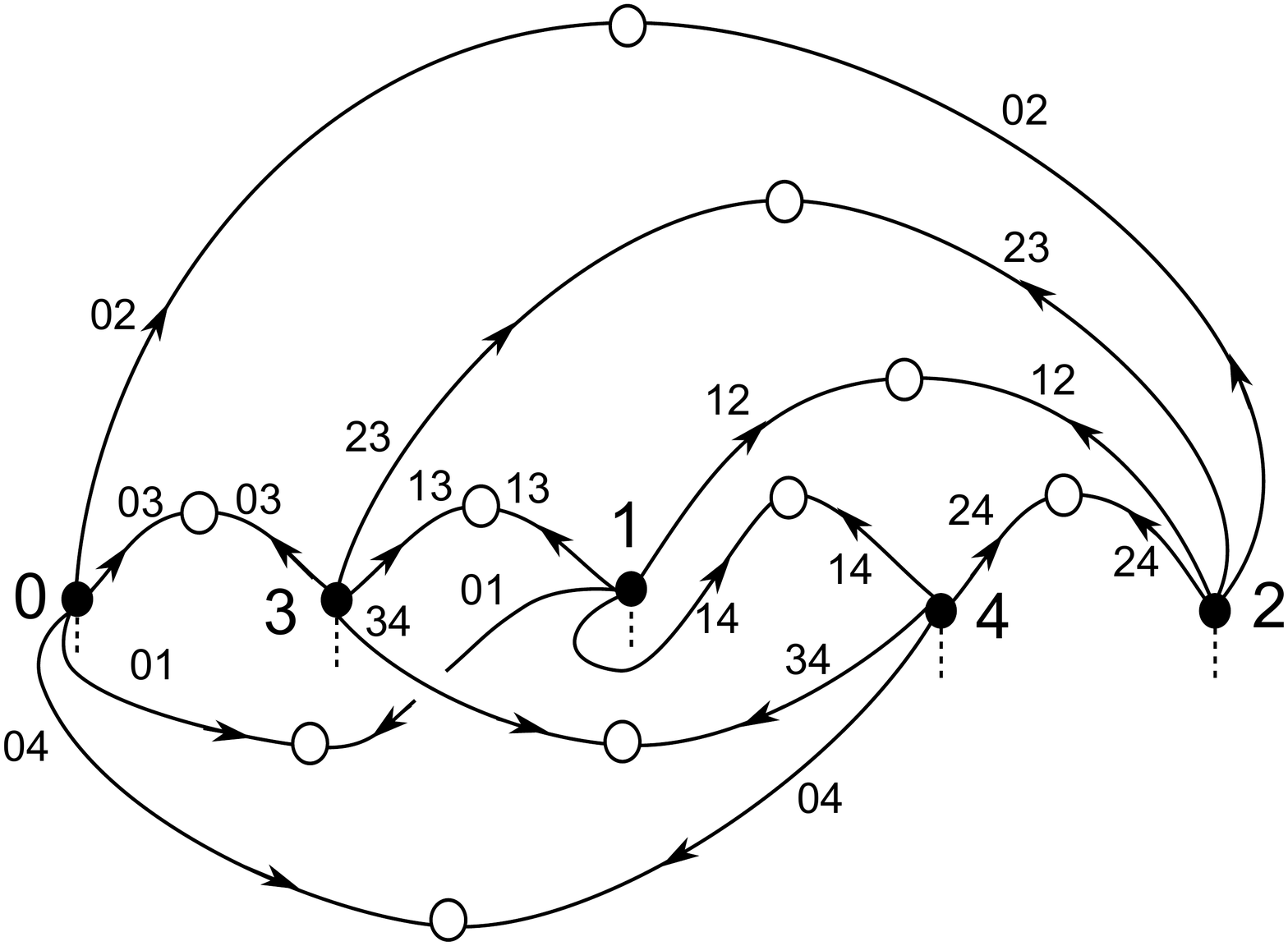}
\centering  \caption{Diagram for the four-simplex amplitude $ev(\Gamma)$. The oriented lines carry EPRL representations. The solid black circles represent EPRL tensors, and the white circles denote the bilinear forms $\beta_\alpha$ on the representation spaces $V_\alpha$.}
  \label{amplitude}
\end{figure}

\begin{definition} {\bf ($4$-simplex amplitude)}
\label{amplitudel}
Let $\alpha: \Delta_{\sigma}^{(2)} \rightarrow \mathrm{Ob}(\mathcal{C})$ denote a colouring and $\Psi$ a state. 
The amplitude  for the $4$-simplex $\sigma$  is a linear map
\begin{align}
&\mathcal{Z}_{\sigma} : H_{a}^* \otimes H_{b} \otimes H_{c}^* \otimes H_{d} \otimes H_{e}^* \rightarrow \C,\nonumber 
\end{align}
obtained by evaluating the spin network diagram in Figure \ref{amplitude}:  $\mathcal{Z}_{\sigma}(\Psi_a \otimes ... \otimes \Psi_e) = ev(\Gamma)$.
\end{definition}

We will now prove that the evaluation of this diagram is well-defined, i.~e.~that the four-simplex amplitude converges. 

\begin{theorem} The four-simplex amplitude $ev(\Gamma)$ defined in Definition \ref{amplitudel} converges absolutely.\label{ampconv}
\end{theorem}

\begin{proof}
To prove this theorem, we  consider the diagram $\Gamma$ depicted in Figure \ref{amplitude_proof}. The elements of the diagrammatic calculus are EPRL tensors, bilinear forms and braidings as explained in the previous subsection. 
At each vertex, we place an EPRL tensor $\Psi^{\alpha}$, $\alpha = (\alpha_1, \alpha_2, \alpha_3, \alpha_4)$, 
$$
\Psi^{\alpha} = \sum_{I}  \, \pi_{\alpha_{1}} (x^{I_{1}}) \, \otimes ... \otimes \pi_{\alpha_{4}} (x^{I_{4}}) e^{K_{1}}_{a_{1}} \otimes ... \otimes e^{K_{4}}_{a_{4}} \left( \begin{array}{cccc} a_{1} & a_{2} & a_{3} & a_{4} \\ 
 K_{1} & K_{2} & K_{3} & K_{4} \end{array} \right) \otimes x_{I_1} ... x_{I_4}.
$$
To each line connecting two vertices we associate the bilinear form $\beta$. This implies that each line coloured by $\alpha(K)$ carries a propagator, i.e., an element of $F_q(\SL(2,\C)_{\mathbb{R}})^{\otimes 2}$ given by
\beqa
\mathcal{K}_{ab}^\alpha\!=\!\mathcal{K}_{ab}(x^I,x^J)\, x_I \!\otimes\! x_J\!=\!\beta \left( \pi_{\alpha} (x^I) e^K_a, \pi_{\alpha} (x^J) e^K_b \right)\, x_I \!\otimes\! x_J
\!=\! \beta \left( \pi_{\alpha} (x^I S(x^J))  e^K_a, e^K_b \right)\, x_I \!\otimes\! x_J.\eeqa
Note that the propagator is given by 
$
\mathcal{K}_{ab}^\alpha = (id \otimes S) \Delta K_{ab}^{\alpha}
$,
where $K^{\alpha}_{ab}$ is the linear form on $U_q(\mathfrak{sl}(2,\C)_{\mathbb{R}})$ defined by
$
K_{ab}^{\alpha} = \beta \left( \pi_{\alpha} ( \cdot ) e^K_a,  e^K_b \right).
$

The diagram also contains a crossing to which we associate  the braiding  $c_{\alpha, \beta} : V_{\alpha} \otimes V_{\beta} \rightarrow V_{\beta} \otimes V_{\alpha}$, given by $c_{\alpha, \beta} =  \tau \circ \pi_{\alpha} \otimes \pi_{\beta}(R)$. Throughout this proof, we will use the notation $R = R_{(1)} \otimes R_{(2)}$ for the $R$-matrix.

\begin{figure}
  \includegraphics[scale=0.4]{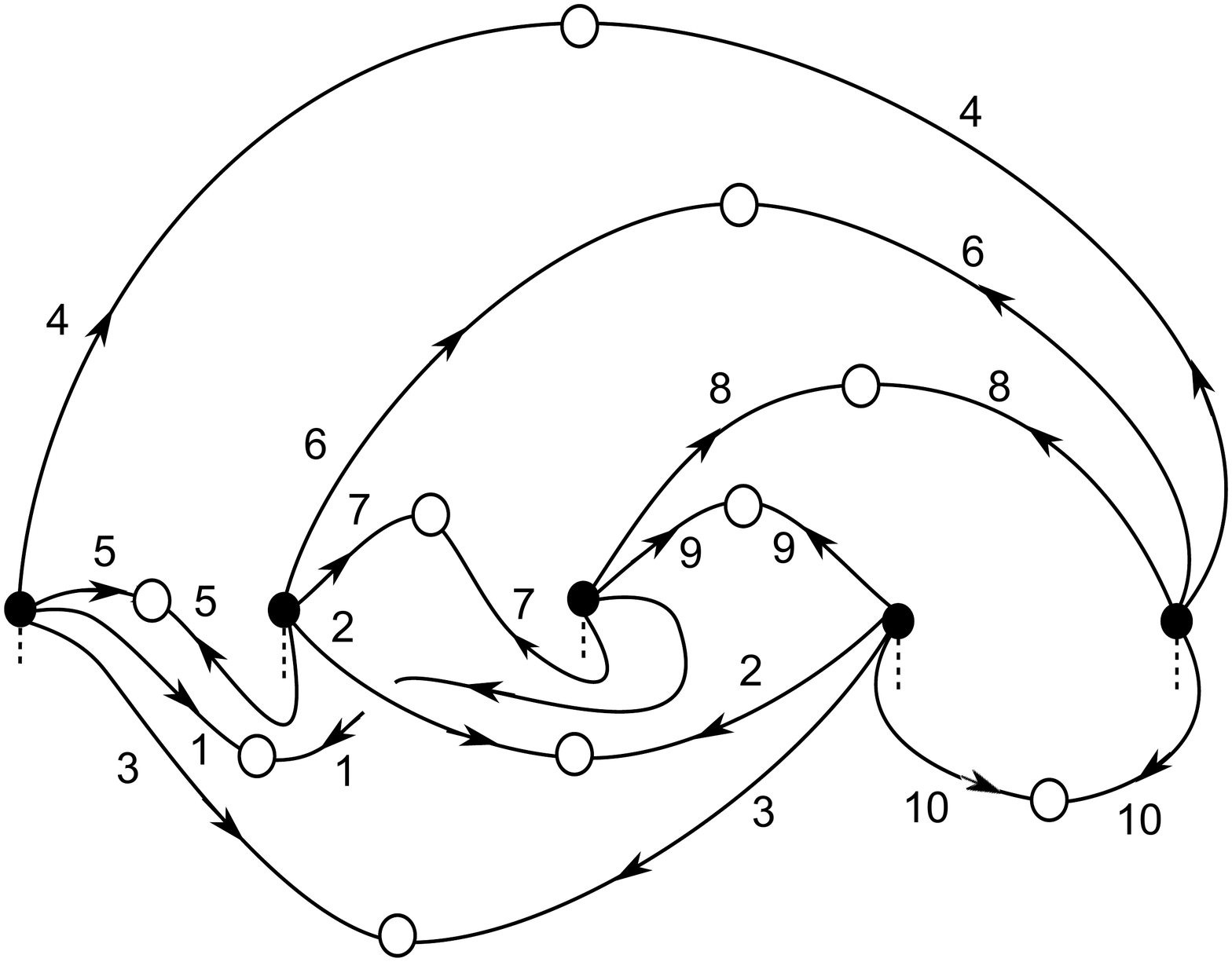}
\centering  \caption{Diagram for the four-simplex amplitude $ev(\Gamma)$ with the labeling used in the proof of Theorem \ref{ampconv}.}
  \label{amplitude_proof}
\end{figure}

To simplify the notation, we  label the lines  of the diagram with numbers  $i=1,...,10$ instead of  associating labels to the vertices. The endpoints of each line carry a basis element of $F_q(\SL(2,\C)_{\mathbb{R}})$ and our conventions are such that the left hand side of the line $i$ carries the basis element $x_{I_i}$ while the right hand side is coloured by the basis element $x_{J_i}$. With the choice of cilium depicted in Figure \ref{amplitude_proof}, the element $\phi(\Gamma)$ of $F_q(\SL(2,\C)_{\mathbb{R}})^{\otimes 5}$ associated to $\Gamma$ is thus given by \ros
\begin{align}
\label{phi}
\phi(\Gamma) = \sum_{I,J} 
&\left( \begin{array}{cccc} a_{4} & a_{5} & a_{1} & a_{3} \\ K_{4} & K_{5} & K_{1} & K_{3} \end{array} \right) 
\left( \begin{array}{cccc} a_{6} & a_{7} & a_{2} & b_{5} \\ 
 K_{6} & K_{7} & K_{2} & K_{5} \end{array} \right) 
 \left( \begin{array}{cccc} a_{8} & a_{9} & b_{1} & b_{7} \\ 
K_{8} & K_{9} & K_{1} & K_{7} \end{array} \right) \nn \\
&\left( \begin{array}{cccc} a_{10} & b_{3} & b_{2} & b_{9} \\ 
 K_{10} & K_{3} & K_{2} & K_{9} \end{array} \right) 
 \left( \begin{array}{cccc} b_{8} & b_{6} & b_{4} & b_{10} \\ 
 K_{8} & K_{6} & K_{4} & K_{10} \end{array} \right) \mathcal{K}_{a_1b_1}(x^{I_1}, x^{J_1} S^{-1}(R_{(2)})) \nn \\ 
&\mathcal{K}_{a_2b_2}(x^{I_2} S^{-1}(R_{(1)}), x^{J_2} )) \mathcal{K}_{a_3b_3}(x^{I_3} , x^{J_3} ))  \mathcal{K}_{a_{4}b_{4}}(x^{I_{4}} , x^{J_{4}} ))  \mathcal{K}_{a_{5}b_{5}}(x^{I_{5}} , x^{J_{5}} )) \mathcal{K}_{a_{6}b_{6}}(x^{I_{6}} , x^{J_{6}} )) \nn \\ 
&\mathcal{K}_{a_{7}b_{7}}(x^{I_{7}} , x^{J_{7}} )) \mathcal{K}_{a_{8}b_{8}}(x^{I_{8}} , x^{J_{8}} )) \mathcal{K}_{a_{9}b_{9}}(x^{I_{9}} , x^{J_{9}} ))  \mathcal{K}_{a_{10}b_{10}}(x^{I_{10}} , x^{J_{10}} ))\nn\\
&x_{I_4} x_{I_5} x_{I_1} x_{I_3} \otimes x_{I_6} x_{I_7} x_{I_2} x_{J_5} \otimes x_{I_8} x_{I_9} x_{J_1} x_{J_7} \otimes x_{I_{10}} x_{J_3} x_{J_2} x_{J_9} \otimes x_{J_8} x_{J_6} x_{J_4} x_{J_{10}}.
\end{align}
The next step is to calculate the propagators. This is achieved via a choice of basis, the elements of which are given by $x^I = X^{\;\;i}_{A \;\; j} \otimes g^{\;\;k}_{B \;\; l}$ and $x^J = X^{\;\;m}_{C \;\; n} \otimes g^{\;\;p}_{D \;\; q}$. In this basis, the propagator takes the form
\begin{align}
\mathcal{K}_{ab}(x^{I} , x^{J})=&\beta \left( \pi_{\alpha} \left( X^{\;\;i}_{A \;\; j} \otimes g^{\;\;k}_{B \;\; l} \, . \, S(X^{\;\;m}_{C \;\; n} \otimes g^{\;\;p}_{D \;\; q} \right)  e^K_a, e^K_b \right) \nn \\
=&\beta \left( \pi_{\alpha} \left(X^{\;\;i}_{A \;\; j} \otimes g^{\;\;k}_{B \;\; l} \, . \, 1 \otimes S^{-1}(g^{\;\;p}_{D \;\; q}) \, . \,  S(X^{\;\;m}_{C \;\; n}) \otimes 1 \right) e^K_a, e^K_b \right) \nn \\
=& \sum_N \epsilon_{D vq} \, \epsilon_{D}^{-1 pw} \epsilon_{C xn} \, \epsilon_{C}^{-1 my}  
\left( \begin{array}{cc} k & v \\
                          B & D  \end{array} \right| \left. \begin{array}{c} N \\ g \end{array} \right)                            \left( \begin{array}{c} h \\ N
\end{array} \right| \left. \begin{array}{cc} B & D \\ l & w \end{array} \right)\nn \\
&\quad\quad \beta \left( \pi_{\alpha} \left(X^{\;\;i}_{A \;\; j} \otimes 1 \, . \, 1 \otimes g^{\;\;g}_{N \;\; h} \, . \,  X^{\;\;x}_{C \;\; y} \otimes 1 \right) e^K_a, e^K_b \right),
\end{align}
where we  used the definition of the antipode and the multiplicative structure of $U_q(\mathfrak{an}(2))$. With expression  \eqref{repQLG}  for the representation and the explicit form of the bilinear form $\beta$, we obtain 
\begin{align}
\mathcal{K}_{ab}(x^{I} , x^{J}) = \sum_{M,N} &\;\delta_{AK} \delta_{CK} v_K^{1/2} \epsilon_{K bj} \epsilon_{K an} \epsilon_{K}^{-1 my} \epsilon_{D vq} \epsilon_{D}^{-1 pw} 
  \\ 
&\left( \begin{array}{cc} k & v \\
                          B & D  \end{array} \right| \left. \begin{array}{c} N \\ g \end{array} \right)                            \left( \begin{array}{c} h \\ N
\end{array} \right| \left. \begin{array}{cc} B & D \\ l & w \end{array} \right)\left( \begin{array}{cc} i & g \\
                          K & N  \end{array} \right| \left. \begin{array}{c} M \\ f \end{array} \right)                            \left( \begin{array}{c} f \\ M
\end{array} \right| \left. \begin{array}{cc} N & K \\ h & y \end{array} \right) 
 \Lambda^{NM}_{KK}.\nn
\end{align}
We can now compute the propagators associated to the lines  involved in the crossing. Denoting the $R$-matrix by $R= \sum_{E} X^{\;\;r}_{E \;\; s} \otimes g^{\;\;s}_{E \;\; r}$, we find that the first propagator is of the form
\beq
\mathcal{K}_{ab}(x^{I}, x^{J} S^{-1}(R_{(2)})) = \beta \left( \pi_{\alpha} \left( X^{\;\;i}_{A \;\; j} \otimes g^{\;\;k}_{B \;\; l} \, . \, 1 \otimes  g^{\;\;s}_{E \;\; r} \, . \, S(X^{\;\;m}_{C \;\; n} \otimes g^{\;\;p}_{D \;\; q} \right)  e^K_a, e^K_b \right).
\eeq
Using again the definition of the antipode and the product  for $U_q(\mathfrak{an}(2))$, we obtain
\begin{align}
&\mathcal{K}_{ab}(x^{I}, x^{J} S^{-1}(R_{(2)})) = \sum_{M,N,P} \delta_{AK} \delta_{CK} v_K^{1/2} \epsilon_{K bj} \epsilon_{K an} \epsilon_{K}^{-1 my} \epsilon_{D vq} \epsilon_{D}^{-1 pw} 
\left( \begin{array}{cc} s & v \\
                          E & D  \end{array} \right| \left. \begin{array}{c} P \\ d \end{array} \right) \\
& \left( \begin{array}{c} c \\ P
\end{array} \right| \left. \begin{array}{cc} E & D \\ r & w \end{array} \right)  \left( \begin{array}{cc} k & d \\
                          B & P  \end{array} \right| \left. \begin{array}{c} N \\ g \end{array} \right)                            \left( \begin{array}{c} h \\ N
\end{array} \right| \left. \begin{array}{cc} B & P \\ l & c \end{array} \right) 
\left( \begin{array}{cc} i & g \\
                          K & N  \end{array} \right| \left. \begin{array}{c} M \\ f \end{array} \right)                            \left( \begin{array}{c} f \\ M
\end{array} \right| \left. \begin{array}{cc} N & K \\ h & y \end{array} \right) \nn
 \Lambda^{NM}_{KK}.
\end{align}
The propagator  associated to the second line  in the crossing  is given by
\beq
\mathcal{K}_{ab}(x^{I} S^{-1}(R_{(1)}), x^{J} ) = \beta \left( \pi_{\alpha} \left( X^{\;\;i}_{A \;\; j} \otimes g^{\;\;k}_{B \;\; l} \, . \, S^{-1}(X^{\;\;r}_{E \;\; s}) \otimes  1 \, . \, S(X^{\;\;m}_{C \;\; n} \otimes g^{\;\;p}_{D \;\; q} \right)  e^K_a, e^K_b \right),
\eeq
and after some further computations,  we find
\begin{align}
\mathcal{K}_{ab}(x^{I} S^{-1}(R_{(1)}), x^{J}) = \sum_{M,N} &\delta_{AK} \delta_{CK} v_K^{1/2} \epsilon_{K bj} \epsilon_{K an} \epsilon_{K}^{-1 my} \epsilon_{D vq} \epsilon_{D}^{-1 pw} \epsilon_{E s's} \epsilon_E^{-1 r'r} 
\left( \begin{array}{cc} i & k \\
                          K & B  \end{array} \right| \left. \begin{array}{c} N \\ g \end{array} \right)  \nn \\ 
                          &\left( \begin{array}{c} g \\ N
\end{array} \right| \left. \begin{array}{cc} B & E \\ l & r' \end{array} \right)  \left( \begin{array}{cc} s' & v \\
                          E & D  \end{array} \right| \left. \begin{array}{c} M \\ f \end{array} \right)                            \left( \begin{array}{c} f \\ M
\end{array} \right| \left. \begin{array}{cc} D & K \\ w & y \end{array} \right) \Lambda^{BN}_{KE}
 \Lambda^{DM}_{EK}.
\end{align}
We have thus derived explicit expressions for all quantities arising in \eqref{phi}
in terms of the Clebsch-Gordan coefficients for $U_q(\mathfrak{\su}(2)))$, the quantities $\Lambda^{AB}_{CD}$ from \eqref{repQLGmod}, the intertwiners \eqref{epsilon} and the 
ribbon element for $U_q(\mathfrak{\su}(2)))$.
It remains to evaluate the Haar integrals  associated with the diagram. For this we need to select one of the five vertices  that does not carry a Haar integral. In the following, we suppose that this vertex is the second vertex from the left in Figure \ref{amplitude}. The amplitude associated to the diagram  is then given by
\beq\label{evexp}
ev(\Gamma) = h \otimes \epsilon \otimes h \otimes h \otimes h (\phi(\Gamma)).
\eeq
The Haar integrals act on the basis elements $x_I = u^{\;\;j}_{A \;\; i} \otimes E^{\;\;l}_{B \;\; k}$ and $x_J = u^{\;\;n}_{C \;\; m} \otimes E^{\;\;q}_{D \;\; p}$ according to formulae \eqref{intAN} and \eqref{intSU}. The counit is a morphism of algebras which acts on the basis elements according to
$$
\epsilon(u^a_{I \;\;b}) = \delta^a_{I \;\; b}, \;\;\;\; \mbox{and} \;\;\;\; \epsilon(E^a_{I \;\;b}) = \delta_{I0}.
$$

Combining these results, we find that the evaluation $ev(\Gamma)$ is given by  an expression of the form 
\begin{align}
\label{sum}
ev(\Gamma) = \sum_{P} \!\!\sum_{\substack{{M_1, M_2, }\\{ M_3,..., M_{10}}}} \sum_{\substack{{N_1, D_2,}\\{N_3, ... ,N_{10}}}} \quad &\Lambda^{N_1 M_1}_{K_1K_1} \Lambda^{D_2 M_2}_{K_2K_2} \Lambda^{N_3 M_3}_{K_3K_3} \cdots \Lambda^{N_{10} M_{10}}_{K_{10}K_{10}} \\[-1.7em]
&[2N_1+1]^{1/2}[2D_2 + 1]^{1/2} [2N_3 + 1]^{1/2} \cdots [2N_{10} + 1]^{1/2}\nn\\
& X(P, M_1, M_2, ..., M_{10}; N_1, D_2, N_3, ..., N_{10} ; K_1, K_2, ..., K_{10}), \nn
\end{align}
where $X$ is a function of the representation labels that factorises  as
\beqa
&& X(P, M_1, M_2, ..., M_{10}; N_1, D_2, N_3, ..., N_{10} ; K_1, K_2, ..., K_{10}) = C(K_1, K_2, ..., K_{10})^{j_1, ..., j_{10} ; n_1, ... , n_{10}}_{i_1, ..., i_{10} ; m_1, ... , m_{10}} \nn \\ 
&& Y(P, M_1, M_2, ..., M_{10}; N_1, D_2, N_3, ..., N_{10} ; K_1, K_2, ..., K_{10})^{i_1, ..., i_{10} ; m_1, ... , m_{10}}_{j_1, ..., j_{10} ; n_1, ... , n_{10}}.
\eeqa
The quantity $C$ depends only on the $U_q(\su(2))$ representation labels $K_1,...,K_{10}$, which are fixed. The
 sums over the associated labels $j_1,...,j_{10}$, $i_1,...,i_{10}$, $m_1,...,m_{10}$, $n_{1},...,n_{10}$  for the basis elements of $V_{K_1}$,..., $V_{K_{10}}$
  are thus finite. The quantities which depend on the summation indices  $P$, $M_1,...,M_{10}$, $N_1,D_2,N_3,...,N_{10}$
   in the sum \eqref{sum} are contained in the function $Y$, which is given explicitly by
\begin{align}
&Y =  \sum_{C,D,E} \left( \prod_{i=1}^{10} \epsilon_{D_i v_iq_i} \epsilon_{D_i}^{-1 p_iw_i} \epsilon_{K_i}^{-1 m_i y_i} \right) \epsilon_{E s's} \epsilon_{E}^{-1 r'r}
\left( \begin{array}{cc} s & v_1 \\
                          E & D_1  \end{array} \right| \left. \begin{array}{c} P \\ d \end{array} \right)   
\left( \begin{array}{c} c \\ P
\end{array} \right| \left. \begin{array}{cc} E & D_1 \\ r & w_1 \end{array} \right)  
\nn  \\
&
\left( \begin{array}{cc} k_1 & d \\
                          B_1 & P  \end{array} \right| \left. \begin{array}{c} N_1 \\ g_1 \end{array} \right)   \left( \begin{array}{c} h_1 \\ N_1 \end{array} \right| \left. \begin{array}{cc} B_1 & P \\  l_1 & c \end{array} \right) 
\left( \begin{array}{cc} i_1 & g_1 \\
                          K_1 & N_1  \end{array} \right| \left. \begin{array}{c} M_1 \\ f_1 \end{array} \right)                           
 \left( \begin{array}{c} f_1 \\ M_1
\end{array} \right| \left. \begin{array}{cc} N_1 & K_1 \\ h_1 & y_1 \end{array} \right)   \nn \\
&
\left( \begin{array}{cc} i_2 &  k_2 \\
                          K_2 & B_2  \end{array} \right| \left. \begin{array}{c} N_2 \\ g_2 \end{array} \right)                         
 \left( \begin{array}{c} g_2 \\ N_2
\end{array} \right| \left. \begin{array}{cc} B_2 & E \\ l_2 & r' \end{array} \right)
 \left( \begin{array}{cc} s' & v_2 \\
                          E & D_2  \end{array} \right| \left. \begin{array}{c} M_2 \\ f_2 \end{array} \right)   \left( \begin{array}{c} f_2 \\ M_2
\end{array} \right| \left. \begin{array}{cc} D_2 & K_2 \\ w_2 & y_2 \end{array} \right) 
 \nn \\ 
&
\left( \begin{array}{cc} k_3 & v_3 \\
                          B_3 & D_3  \end{array} \right| \left. \begin{array}{c} N_3 \\ g_3 \end{array} \right)                           
 \left( \begin{array}{c} h_3 \\ N_3
\end{array} \right| \left. \begin{array}{cc} B_3 & D_3 \\ l_3 & w_3 \end{array} \right)  \left( \begin{array}{cc} i_3 & g_3 \\
                          K_3 & N_3  \end{array} \right| \left. \begin{array}{c} M_3 \\ f_3 \end{array} \right)                            
\left( \begin{array}{c} f_3 \\ M_3
\end{array} \right| \left. \begin{array}{cc} N_3 & K_3 \\ h_3 & y_3 \end{array} \right) 
 \nn \\ 
&
\left( \begin{array}{cc} k_4 & v_4 \\
                          B_4 & D_4  \end{array} \right| \left. \begin{array}{c} N_4 \\ g_4 \end{array} \right)                           
 \left( \begin{array}{c} h_4 \\ N_4
\end{array} \right| \left. \begin{array}{cc} B_4 & D_4 \\ l_4 & w_4 \end{array} \right)  \left( \begin{array}{cc} i_4 & g_4 \\
                          K_4 & N_4  \end{array} \right| \left. \begin{array}{c} M_4 \\ f_4 \end{array} \right)                            
\left( \begin{array}{c} f_4 \\ M_4
\end{array} \right| \left. \begin{array}{cc} N_4 & K_4 \\ h_4 & y_4 \end{array} \right) 
 \nn \\ 
&
\left( \begin{array}{cc} k_5 & v_5 \\
                          B_5 & D_5  \end{array} \right| \left. \begin{array}{c} N_5 \\ g_5 \end{array} \right)                           
 \left( \begin{array}{c} h_5 \\ N_5
\end{array} \right| \left. \begin{array}{cc} B_5 & D_5 \\ l_5 & w_5 \end{array} \right)  \left( \begin{array}{cc} i_5 & g_5 \\
                          K_5 & N_5  \end{array} \right| \left. \begin{array}{c} M_5 \\ f_5 \end{array} \right)                            
\left( \begin{array}{c} f_5 \\ M_5
\end{array} \right| \left. \begin{array}{cc} N_5 & K_5 \\ h_5 & y_5 \end{array} \right) 
 \nn \\ 
&
\left( \begin{array}{cc} k_6 & v_6 \\
                          B_6 & D_6  \end{array} \right| \left. \begin{array}{c} N_6 \\ g_6 \end{array} \right)                           
 \left( \begin{array}{c} h_6 \\ N_6
\end{array} \right| \left. \begin{array}{cc} B_6 & D_6 \\ l_6 & w_6 \end{array} \right)  \left( \begin{array}{cc} i_6 & g_6 \\
                          K_6 & N_6  \end{array} \right| \left. \begin{array}{c} M_6 \\ f_6 \end{array} \right)                            
\left( \begin{array}{c} f_6 \\ M_6
\end{array} \right| \left. \begin{array}{cc} N_6 & K_6 \\ h_6 & y_6 \end{array} \right) 
 \nn \\ 
&
\left( \begin{array}{cc} k_7 & v_7 \\
                          B_7 & D_7  \end{array} \right| \left. \begin{array}{c} N_7 \\ g_7 \end{array} \right)                           
 \left( \begin{array}{c} h_7 \\ N_7
\end{array} \right| \left. \begin{array}{cc} B_7 & D_7 \\ l_7 & w_7 \end{array} \right)  \left( \begin{array}{cc} i_7 & g_7 \\
                          K_7 & N_7  \end{array} \right| \left. \begin{array}{c} M_7 \\ f_7 \end{array} \right)                            
\left( \begin{array}{c} f_7 \\ M_7
\end{array} \right| \left. \begin{array}{cc} N_7 & K_7 \\ h_7 & y_7 \end{array} \right) 
 \nn \\ 
&
\left( \begin{array}{cc} k_8 & v_8 \\
                          B_8 & D_8  \end{array} \right| \left. \begin{array}{c} N_8 \\ g_8 \end{array} \right)                           
 \left( \begin{array}{c} h_8 \\ N_8
\end{array} \right| \left. \begin{array}{cc} B_8 & D_8 \\ l_8 & w_8 \end{array} \right)  \left( \begin{array}{cc} i_8 & g_8 \\
                          K_8 & N_8  \end{array} \right| \left. \begin{array}{c} M_8 \\ f_8 \end{array} \right)                            
\left( \begin{array}{c} f_8 \\ M_8
\end{array} \right| \left. \begin{array}{cc} N_8 & K_8 \\ h_8 & y_8 \end{array} \right) 
 \nn \\ 
&
\left( \begin{array}{cc} k_9 & v_9 \\
                          B_9 & D_9  \end{array} \right| \left. \begin{array}{c} N_9 \\ g_9 \end{array} \right)                           
 \left( \begin{array}{c} h_9 \\ N_9
\end{array} \right| \left. \begin{array}{cc} B_9 & D_9 \\ l_9 & w_9 \end{array} \right)  \left( \begin{array}{cc} i_9 & g_9 \\
                          K_9 & N_9  \end{array} \right| \left. \begin{array}{c} M_9 \\ f_9 \end{array} \right)                            
\left( \begin{array}{c} f_9 \\ M_9
\end{array} \right| \left. \begin{array}{cc} N_9 & K_9 \\ h_9 & y_9 \end{array} \right) 
 \nn \\ 
  \intertext{}
&
\left( \begin{array}{cc} k_{10} & v_{10} \\
                          B_{10} & D_{10}  \end{array} \right| \left. \begin{array}{c} N_{10} \\ g_{10} \end{array} \right)                           
 \left( \begin{array}{c} h_{10} \\ N_{10}
\end{array} \right| \left. \begin{array}{cc} B_{10} & D_{10} \\ l_{10} & w_{10} \end{array} \right) \left( \begin{array}{cc} i_{10} & g_{10} \\
                          K_{10} & N_{10}  \end{array} \right| \left. \begin{array}{c} M_{10} \\ f_{10} \end{array} \right)                            
\left( \begin{array}{c} f_{10} \\ M_{10}
\end{array} \right| \left. \begin{array}{cc} N_{10} & K_{10} \\ h_{10} & y_{10} \end{array} \right) 
\nn \\ 
& \delta_{B_6 0} \delta_{B_7 0}  \delta_{B_2 0} \delta_{D_5 0} \; 
\delta_{B_4 B_5} \delta_{B_4 B_1}  \delta_{B_1 B_3}  \delta_{k_4}^{l_5} \delta_{k_5}^{l_1}  \delta_{k_1}^{l_3} \;
\delta_{B_8 B_9} \delta_{B_8 D_1}  \delta_{D_1 D_7}  \delta_{k_8}^{l_9} \delta_{k_9}^{q_1}  \delta_{p_1}^{q_7} \; \delta_{B_{10} D_3} \delta_{B_{10} D_2}  \delta_{D_2 D_9}  \nn \\
& \delta_{k_{10}}^{q_3} \delta_{p_3}^{q_2}  \delta_{p_2}^{q_9} \;
\delta_{D_8 D_6} \delta_{D_8 D_4}  \delta_{D_4 D_{10}}  \delta_{p_8}^{q_6} \delta_{p_6}^{q_4}  \delta_{p_4}^{q_{10}} \,
 \pi_{B_4}(\mu^{-1})^{l_4}_{k_3} \pi_{B_8}(\mu^{-1})^{l_8}_{p_7} \pi_{B_{10}}(\mu^{-1})^{l_{10}}_{p_9} \pi_{D_8}(\mu^{-1})^{q_8}_{p_{10}}\nn\\
 &[2D_2+1]^{1/2}[2N_5+1]^{1/2}[2N_6+1]^{1/2}[2N_7+1]^{1/2}[2N_1+1]^{-1/2}[2N_3+1]^{-1/2}[2N_4+1]^{-1/2}\nn\\
 &[2N_8+1]^{-1/2}[2N_9+1]^{-1/2}[2N_{10}+1]^{-1/2}.\label{yexp}
\end{align}
Implementing the delta functions in the above expression enforces the relations $B_6 = B_7 = B_2 = D_5 = 0$. Together with the Clebsch-Gordan conditions and the remaining delta functions, the conditions imply  $B_1 = B_3 = B_4 = B_5 = N_5$; $B_8 = B_9 = D_1 = D_7 = N_7$; $B_{10} = D_3 = D_9 = D_2$; $D_8 = D_6 = D_4 = D_{10} = N_6$, and $E = N_2 = K_2$. This removes the sum over the representation labels $C,D,E$ in \eqref{yexp}. 

The resulting expression for $Y$
 is of the type obtained in the $q$-deformation \cite{PK} of the Barrett-Crane model and can be represented by a $U_q(\su(2))$ spin network.
 The only difference is that our expression for $Y$ involves contributions from the $R$-matrix\footnote{Note however that the summation label arising from the expression of the $R$-matrix is fixed to $K_2$.} and that, in contrast to \cite{PK},  our  spin network diagram contains open edges labelled with the variables $K_1,...,K_{10}$. 
We can therefore follow the procedure  in \cite{PK} and recouple the diagram to the product  of an  open diagram with a product
 of $U_q(\su(2))$ $6j$ symbols. The open diagram depends only on the variables $K_1,...,K_{10}$. It therefore contains only finite sums and is finite. The $6j$ symbols for $U_q(\su(2))$ are bounded by one. We therefore find that the function $Y$ is bounded by a polynomial $Q(P,D_2,\{M\},\{N\},\{K\})$ in the variables $P, M_1, M_2, ..., M_{10}, N_1, D_2, N_3, ..., N_{10}, K_1,...,K_{10}$. 
This implies that  the evaluation \eqref{sum} of the diagram is given by
\begin{align}\label{evred}
ev(\Gamma)= &C(K_1, K_2, ..., K_{10})^{j_1, ..., j_{10} ; n_1, ... , n_{10}}_{i_1, ..., i_{10} ; m_1, ... , m_{10}}\\
&\sum_{P,D_2}\sum_{M_1,...,M_{10}}\sum_{N_1,N_3,...,N_{10}}  S(P,D_2, \{M\},\{N\}, \{K\})^{i_1, ..., i_{10} ; m_1, ... , m_{10}}_{j_1, ..., j_{10} ; n_1, ... , n_{10}},\nn
\end{align}
where the absolute values of the summands are bounded as follows
\begin{align}
\big | S(P,D_2,\{M\}, &\{N\}, \{K\})^{i_1, ..., i_{10} ; m_1, ... , m_{10}}_{j_1, ..., j_{10} ; n_1, ... , n_{10}} \big | \leq   \mid \Lambda^{N_1M_1}_{K_1K_1} \mid\mid \Lambda^{D_2M_2}_{K_2K_2}\mid  \mid \Lambda^{N_3M_3}_{K_3K_3} \mid ... \mid \Lambda^{N_{10}M_{10}}_{K_{10}K_{10}} \mid\nn \\
&   [2N_1 + 1]^{1/2} [2D_2 + 1]^{1/2} [2N_3 + 1]^{1/2} \cdots [2N_{10} + 1]^{1/2} Q(P, D_2,\{M\}, \{N\},\{K\})\nn.
\end{align}
Moreover, it follows from the expression \eqref{yexp} that the function $Y$ vanishes unless the representation labels in $Y$ satisfy the following inequalities 
\begin{align}\label{meq}
&|N_i-K_i|\leq M_i\leq N_i+K_i\quad \forall i\in\{1,3,...,10\} & &|D_2 -K_2|\leq M_2\leq D_2+K_2\\
&|K_2-N_7|\leq P\leq K_2+N_7\nn\\
&|N_5-D_2|\leq N_3\leq N_5+D_2 & &|N_5-N_6|\leq N_4\leq N_5+N_6\nn \\
&|N_6-N_7|\leq N_8\leq N_5+N_7 & &|N_7-D_2|\leq N_9\leq N_7+D_2\nn\\
&|N_6-D_2|\leq N_{10}\leq D_2+N_6 & &|P-N_5|\leq N_1\leq P+N_5\nn.
\end{align}
The first three inequalities in \eqref{meq} imply that  for fixed  $N_1,D_2,N_3,..., N_{10}$, the summation range of the variables $M_1,...,M_{10}$ and  $P$  in \eqref{evred} is restricted to finite intervals. 
We now take into account the asymptotic behaviour of the expressions $[2N+1]$ and $\Lambda^{N N+R}_{KK}$ derived in \cite{PE2}
\begin{align}
[2N+1]\sim q^{-2N}\qquad \mid \Lambda^{N N+R}_{KK} \mid \leq C(2N + 2R +1) q^{2N}\qquad\text{for}\; N\to\infty,
\end{align}
where $C>0$ is a constant. For $N_1,D_2,N_3,...,N_{10}$ sufficiently large and $M_1,...,M_{10}$, $P$ subject to the first three  conditions in \eqref{meq}, we can thus bound the value of the summand in \eqref{evred} by
\begin{align}
\sum_P\sum_{M_1,...,M_{10}}&\left|S(P,D_2,\{M\}, \{N\}, \{K\})^{i_1, ..., i_{10} ; m_1, ... , m_{10}}_{j_1, ..., j_{10} ; n_1, ... , n_{10}}\right|\\
 &\leq    q^{2(N_1+D_2+N_3+\ldots N_{10})}q^{-N_1-D_2-N_3-\ldots -N_{10}}P_K(N_1,D_2,N_3,...,N_{10})\nn \\
&\leq q^{N_1+D_2+N_3+\dots+N_{10}}
P_K(N_1,D_2,N_3,...,N_{10}) \nn,
\end{align}
where $P_K$ is a polynomial in the variables $N_1,D_2,N_3,\cdots N_{10}$ that depends on the external parameters $K_1,...,K_{10}$.
As $0<q<1$, this implies that there exists a constant $A\in\mathbb N$ such that 
$$
\sum_{P}\sum_{M_1,...,M_{10}} \left | S(P,D_2,\{M\}, \{N\}, \{K\})^{i_1, ..., i_{10} ; m_1, ... , m_{10}}_{j_1, ..., j_{10} ; n_1, ... , n_{10}}\right |\leq 
 2 q^{N_1+D_2+N_3+\dots+N_{10}}
$$
for all $N_1,D_2,N_3,...,N_{10}\geq A$. The series in \eqref{evred} then converge absolutely, because we have
\begin{align}
&\sum_{N_1,N_3,...,N_{10}}  \sum_{P}\sum_{M_1,...,M_{10}} \left|S(P,D_2, \{M\},\{N\}, \{K\})^{i_1, ..., i_{10} ; m_1, ... , m_{10}}_{j_1, ..., j_{10} ; n_1, ... , n_{10}}\right|\\ 
&\leq 
\sum_{N_1,N_3,...,N_{10}< A}  \sum_{P}\sum_{M_1,...,M_{10}} \left|S(P,D_2, \{M\},\{N\}, \{K\})^{i_1, ..., i_{10} ; m_1, ... , m_{10}}_{j_1, ..., j_{10} ; n_1, ... , n_{10}}\right|  +\!\!\!\!\!\!\!\!\!\!
\sum_{N_1,N_3,...,N_{10}\geq A} \!\!\!\!\!\!\!\!\!\!q^{N_1+D_2+N_3+\dots+N_{10}}\nn\\
&\leq 
\sum_{N_1,N_3,...,N_{10}< A}  \sum_{P}\sum_{M_1,...,M_{10}} \left|S(P,D_2, \{M\},\{N\}, \{K\})^{i_1, ..., i_{10} ; m_1, ... , m_{10}}_{j_1, ..., j_{10} ; n_1, ... , n_{10}}\right|    +\frac{2q^{10 A}}{(1-q)^{10}},\nn
\end{align}
 and the series 
in the last expression 
reduce to  finite sums due to the first three inequalities in \eqref{meq}.
This proves the claim. 
\end{proof}

\subsection{The quantum spin foam model}

Using the definitions and results from  the previous subsections, we can now define the quantum EPRL spin foam model. 

\begin{definition}\label{zem}
Let $M$ be a closed oriented triangulated $4$-manifold with sets of $n$-simplexes $\Delta^{(n)}$. Let $\sigma \in \Delta^{(4)}$ be a $4$ simplex of $M$ and let $\Delta^{(2)}_{\sigma}$ and $\Delta^{(3)}_{\sigma}$ denote its set of triangles and tetrahedra respectively. Let $\alpha: \Delta^{(2)}_{\sigma} \rightarrow \mathrm{Ob}(\mathcal{C})$ denote the corresponding EPRL-colouring, $\Psi : \Delta^{(3)}_{\sigma} \rightarrow \Hom(\mathcal{C})$ the associated EPRL-state, and let the 4-simplex amplitudes $\mathcal{Z}_{\sigma}$ be given by Definition \ref{amplitudedef}. The partition function for the quantum EPRL model associated to $M$ is given by 
\beq
\mathcal{Z}(M) = \sum_{K,J} \prod_{t} [2K_{t} + 1] \prod_{\sigma} \mathcal{Z}_{\sigma} \left(\alpha(K_{t(\sigma)}),\Psi_{\sigma}(J_{tet(\sigma)}) \right).
\eeq
Here, the sum ranges over all $K$ in $\mathcal{L}$ and over the elements of a basis of $U_q(\su(2))$-intertwiners $(\Lambda_{K,J})_J$ entering the definition of the EPRL tensors for each tetrahedron $tet$ of $M$. The state $\Psi_{\sigma}$ for the $4$-simplex $\sigma$ is given by $\Psi_{\sigma} = \bigotimes_{tet(\sigma) \in \partial \sigma} \Psi_{tet(\sigma)}$, where $\Psi_{tet(\sigma)}$ is the state associated to the tetrahedron $tet(\sigma)$ and the order in the tensor product is the one given in Definition \ref{amplitudedef}. The products run over all the triangles $t$ and $4$-simplexes $\sigma$ of $M$.
\end{definition}

The weight associated to the triangles is fixed from gluing arguments as in the classical case \cite{Carlo2}. As there is only a finite number of representations in the label set  $\mathcal{L}$, the sum in Def.~\ref{zem} involves only a finite number of terms and hence converges. Given the convergence of the EPRL intertwiners and the  $4$-simplex amplitude for fixed labels $K$, the convergence of the 
 Lorentzian $q$-EPRL model is therefore immediate.  However, the proofs of the convergence of the EPRL intertwiners and the  $4$-simplex amplitude in the previous subsection are intricate and require a careful analysis of the representation theory of the quantum Lorentz group.   

\section{The Euclidean model}

In this section we study the Euclidean formulation of the model which is based on the representation theory of the quantum group $U_q^{res}(\su(2))$ at a root of unity. 
In this case, the generalisation of the EPRL-model is simple and more direct because it can be achieved in the framework of modular categories. These categories have been studied extensively and all ingredients needed in the construction of the model are available in the literature.
The resulting model does not present the difficulties associated with the Lorentzian case. This is due to the fact that there is only a finite number of irreducible representations and, consequently, no risks of divergences.

\subsection{The quantum rotation group}

\label{qrot}

\subsubsection{$U_q(\su(2))$ at root of unity.}

We start with a brief summary of the representation theory of the Hopf algebra  $U_q(\su(2))$
at a root of unity. In the following, we suppose that $q$ is a primitive $r$th root of unity $q=e^{2\pi i/r}$ with $r>2$ odd. The Hopf algebra $U_q(\su(2))$ at a root of unity is most easily presented in terms of
four generators $E,F,K^{\pm 1}$   that are given in terms of the generators in Section \ref{halgsec} by
\begin{align}
E=J_+q^{J_z}\qquad F=q^{-J_z}J_-\qquad K^{\pm 1}=q^{\pm 2J_z}.
\end{align}
The algebra structure and coalgebra structure are given by this identification together with formulas \eqref{rels1} to \eqref{uqantip2}. However, the star structure differs from the one in Section \ref{halgsec} and takes the form
\beq
\label{star}
\star K^{\pm 1} = K^{\pm 1}\qquad\star E=q^\inv KF\qquad\star F=qEK^\inv.
\eeq
The  finite-dimensional star Hopf algebra  $\uqr$  is the quotient of the Hopf subalgebra generated by $E,F,K^{\pm 1}$  by the two sided ideal generated by the elements $E^{r}, F^{r}, K^{\pm r}-1$. It can be identified with the associative algebra generated by $E,F,K$ subject to the algebra relations derived from \eqref{rels1} and  the additional relations
\begin{align}
E^{r}=F^{r}=0\qquad K^{\pm r}=1
\end{align}
Note that the star Hopf algebra $\uqr$  is a  ribbon Hopf algebra with universal $R$-matrix
\begin{align}
&R=\frac 1 {r} \sum_{0\leq i,j,k\leq r-1}\frac{(q-q^\inv)^k}{[k]!}q^{\tfrac k 2(k-1)+2k(i-j)-2ij}\; E^k K^i\oo F^kK^j,\label{rmatunity}
\end{align} 
where  $[k]$ and $[k]!$ are defined as in \eqref{qexp}.
The ribbon element $v$ is defined by the relation  $v^2=uS(u)$  with $v^2=uS(u)$, and the group-like element   $\mu = u v^{-1}$ takes the form $\mu=K$.

\subsubsection{Representation theory of $\uqr$ at a root of unity}

The representation theory of $\uqr$ differs from the generic case and the classical case in two fundamental  ways. Firstly, up to isomorphisms, there is only a finite number of  irreducible 
representations. Secondly,  $\uqr$  has finite-dimensional {\em indecomposable}  representations\footnote{An indecomposable representation has dimension $2r$ and is characterised by a half-integer $I$ such that $1 \leq 2I + 1 \leq r$, see for instance \cite{tilt1, tilt2, Arnaudon, Alvarez-Gaume} or  \cite{CP}.}   \cite{Arnaudon,Alvarez-Gaume},  which are not  fully reducible and have vanishing quantum dimension. 
As the tensor product decomposition of two irreducible representations  contains indecomposable representations, the representations of $\uqr$ do not form a fusion category.

This obstacle  is  overcome with the tilting module construction  \cite{tilt1,tilt2}, for a pedagogical introduction see \cite{CP}. In this construction  the tensor product of the finite-dimensional irreducible representations is modified in a consistent way such that all factors arising in the tensor product decomposition of finite-dimensional irreducible representations are again finite-dimensional irreducible representations.
The result is a modular category $\cres$. As this category has been discussed extensively in the literature, see for instance \cite{Kirillov}, we limit our presentation to a brief  summary  that emphasises the link with the classical theory and the Lorentzian model.

\paragraph{Monoidal structure}
We start by outlining its structure as $\C$-linear tensor category.
The objects of $\cres$  are certain finite-dimensional irreducible representations
 $(\pi_V,V)$  of $\uqr$ over $\C$, the tilting modules. Its morphisms are intertwiners between tilting modules.  The category $\cres$ is a strict $\C$-linear monoidal category.  This means that there is a $\C$-bilinear  functor $\otimes: \cres\times\cres\rightarrow\cres$, which corresponds to the tensor product of representations,  and a unit object $1$, 
such that $1\otimes (\pi_V,V)\cong (\pi_V,V)\otimes 1\cong (\pi_V,V)$ for all objects $(\pi_V,V)$. The unit object  corresponds to the representation of $\uqr$ on $\C$ that is defined by the counit $\epsilon:\uqr\rightarrow \C$.

\paragraph{Semisimplicity}
The category $\cres$ is abelian, finite and semisimple. This implies  that there is a finite collection  of simple objects $(\pi_J,V_J)$ of $\cres$, which correspond to finite-dimensional irreducible tilting modules  and which are unique up to isomorphism. These simple objects are labelled 
 by a half-integer $J$ with $0\leq J \leq (r-2)/2$.  Again, the representation spaces $V_J$  are $(2J+1)$-dimensional, and we denote by $\{e^J_m\}_{m=-J,...,J}$ the orthonormal basis of $V_J$. Then the action of the generators $E,F,K$ on $V_J$ takes the form
\begin{align}
\label{repsuqr}
&\pi_{J}(E) \, e^J_m= q^{m+1/2} \sqrt{[J- m] [J+ m+1]} \, e^J_{m+ 1}, \\
&\pi_{J}(F)\, e^J_m=q^{-(m-1/2)}\sqrt{[J+m][J-m+1]}\, e^J_{m-1}, \nonumber\\
&\pi_{J}(K^{\pm 1}) \, e^J_m= q^{\pm 2m} \, e^J_m. \nonumber
\end{align}
Every object in $\cres$ is given as a direct sum of simple objects
\beq\label{decomuq}
V=\bigoplus_{J} V_J,
\eeq
and the tensor product of two simple objects decomposes into a direct sum of simple objects according to 
\beq
\label{fusionunityroot}
V_{I} \otimes V_{J} \cong  \bigoplus_{K = | I - J |}^{\mbox{{ \tiny min}} (I + J, r - 2 -(I + J))} V_K.
\eeq
Note that the decomposition is very similar to the one for  generic $q$. The only difference is that the tensor product 
decomposition is cut off at $r-2-(I+J)$ instead of continuing to $I+J$ if $I+J>(r-2)/2$. The isomorphism in \eqref{fusionunity} is given by the  Clebsch-Gordan morphisms
$C^{K}_{\; IJ}: V_I \otimes V_J \rightarrow V_K$, $C^{IJ}_{\;\, K} : V_K \rightarrow V_I \otimes V_J$.
As in the generic case, the vector space $\text{Hom}(V_I\otimes V_J, V_K)$ of intertwiners between $V_I\otimes V_J$ and $V_K$ is either one-dimensional or trivial and the Clebsch-Gordan intertwiners are therefore unique up to normalisation. 
Their matrix elements with respect to the bases $\{e^j_a\}_{a=-j,...,j}$ are the Clebsch-Gordan coefficients. They satisfy the same relations as in the generic case with the additional condition  that they vanish unless
\beq\label{qCG}
|I-J|\leq K\leq \text{min}(I+J, r-2-(I+J)).
\eeq
This condition implements the cut off in the fusion rules \eqref{fusionunityroot}.

\paragraph{Braiding}
The category $\cres$ is braided. This means that   for all objects $(\pi_V,V)$, $(\pi_W,W)$  there exists an intertwiner between $V\otimes W$ and $W\otimes V$, the  braiding, which is   given by the representation of the universal $R$-matrix  $c_{V,W}=\tau\circ (\pi_V\otimes \pi_W)(R)$.
Using expression \eqref{rmatunity} and \eqref{repsuqr} one finds that the braiding $c_{I,J} : V_I \otimes V_J \rightarrow V_J \otimes V_I$ between two simple objects takes the form
 \begin{align}
 \label{braidingres}
c_{IJ}(e^I_m\oo e^J_n)=\sum_{K=0}^{r-1} &(q-q^\inv)^k[k]!q^{-k(k+1)/2+k(n-m)+2nm} \nn \\
&\sqrt{\left[\begin{array}{c} \!\!\!I-m\!\!\!\\ k\end{array}\right]\!\!\left[\begin{array}{c} \!\!\!I+m+k\!\!\!\\ k\end{array}\right]\!\!\left[\begin{array}{c}\!\!\! J+n\!\!\!\\ k\end{array}\right]\!\!\left[\begin{array}{c} \!\!\!J-n+k\!\!\!\\ k\end{array}\right]} e^J_{n-k}\oo e^I_{m+k},\nonumber
\end{align}
\paragraph{Pivotal structure and duals}
The category $\cres$ is a braided monoidal category with duals. This means that for each object
$(\pi_V,V)$ of $\cres$ there is a dual object $(\pi_{V^*}, V^*)$, which corresponds to the representation of $\uqr$ on the dual vector space $V^*$. It is given by
 the antipode
$$
\pi_{V^*}(a)\alpha=\alpha\circ \pi_V(S(a))\qquad \forall \alpha\in V^*, a\in \uqr.
$$
Moreover, for each object $(\pi_V,V)$ there are morphisms
$d_V: V^*\otimes V\rightarrow 1$, the evaluation, and  $b_V:  1\rightarrow V\otimes V^*$, the coevaluation, which  satisfy the identities
\begin{align}
&b_{V\otimes W}=(\text{id}_V\otimes b_W\otimes \text{id}_{V^*})\circ b_V, \qquad d_{V\otimes W}=d_V\circ (\text{id}_{V^*}\otimes d_W\otimes \text{id}_V)\nonumber\\
&(\text{id}_V\otimes d_V)\circ ( b_V \otimes \text{id}_V)=(d_{V^*}\otimes \text{id}_V) \circ(\text{id}_V\otimes b_{V^*})=\text{id}_V.  
\end{align}
 The evaluation corresponds to the pairing between the vector space $V$ and its dual. In terms of the  basis  $\{e^J_m\}_{m=-J,...,J}$  and the  dual basis $\{e^{Jm}\}_{m=-J,...,J}$ of $V_J^*$, the evaluation and coevaluation are given by
$$
b_J(1) = e^J_m \otimes e^{J m}, \;\;\;\;\; d_J(e^{J m} \otimes e^J_n) = e^{Jm} (e^J_n)=\delta^m_n.
$$
They associate to each intertwiner
$\phi$ between representations $(\pi_V,V)$ and $(\pi_W,W)$ a dual intertwiner $\phi^*$ from $(\pi_{W^*},W^*)$ to $(\pi_{V^*}, V^*)$ given by
\begin{align}
\phi^*=(d_{W}\otimes \text{id}_{V^*})\circ(\text{id}_{W^*}\otimes \phi\otimes \text{id}_{V^*})\circ(\text{id}_{W^*}\otimes b_V).
\end{align}
Note that due to the identity $S^2(a)=KaK^\inv$ for all $a\in\uqr$, the functor $*\circ *$ is not the identity functor.

\paragraph{Ribbon structure and quantum trace}
The category $\cres$ is a ribbon category and thus equipped with a twist. The twist is given as 
 a collection of  self-intertwiners $\theta_V=\pi_V(v^{-1}): V\rightarrow V$
 for each object $(\pi_V,V)$, which commute with all morphisms. It is defined  by the action of  the ribbon element  $v$ of $\uqr$.  For simple objects $(\pi_J,V_J)$ it takes the form
\beq
\label{ribbon}
\theta_J := \pi_J(v^{-1}) = v_J \, \text{id}_J, \quad v_J = e^{2i \pi J} q^{-2J(J+1)}
\eeq
The ribbon element and the related element $\mu=u^\inv v$  define the quantum trace of the category $\cres$, which is given as a
collection of  linear maps $\mathrm{tr}_q: \text{End}(V)\rightarrow \C$ from the endomorphism space $\text{Hom}(V,V)$ of each object $(\pi_V,V)$ to $\C$. It defines a non-degenerate pairing $\text{Hom}(V,W)\times \text{Hom}(W,V)\rightarrow \C$, $(f,g)\mapsto \mathrm{tr}_q(f\circ g)$.  and  is given by
\beq
\mathrm{tr}_q (\phi) := \mathrm{tr}_V \, (\pi_V(\mu^{-1}) \phi), \;\;\;\; \forall \phi \in \End (V).
\eeq
The
quantum dimension is defined as the quantum trace of the identity endomorphism $\text{id}_V$. This implies in particular  that the quantum dimension of the simple objects $(\pi_J,V_J)$ is given by  $ \mathrm{tr}_q (\text{id}_J)= \mathrm{tr}_J\, (\pi_J(\mu^{-1}))=[2J +1]$.

\paragraph{Identification of objects and their duals}
To exhibit the structural similarities between the Lorentzian and the Euclidean
model, it will be useful to introduce an identification between each object $(\pi_V,V)$ and its dual $(\pi_{V^*}, V^*)$. This amounts to a choice of basis for each  representation space $V_J$. By assigning the dual basis to the representation space $V^J$, one obtains  
 bijective intertwiners  $\epsilon_V: V\rightarrow V^*$, 
$\epsilon^V: V^*\rightarrow V$. For irreducible objects $(\pi_J,V_J)$ their
 matrix elements with respect to the bases $\{e^J_m\}_{m=-J,...,J}$ and $\{e^{Jm}\}_{m=-J,...,J}$ take the same form as for a real deformation parameter
\begin{align}\label{epsexpluqr}
&\epsilon_J(e^J_m)=\epsilon_{J \, nm}e^{Jn}\qquad \epsilon^J(e^{Jm})=\epsilon^{J\,nm}e^J_n
\qquad\epsilon_{J \, mn} = \epsilon^{J\, mn}  = c_J e^{i \pi(J-m)}q^m\delta_{m,-n},
\end{align}
where 
$c_J$ is given by the ribbon element \eqref{ribbon}: $c_J = e^{-i \pi J} v_J^{-1/2}$.  The matrix elements satisfy 
\beq
\label{epsiloneuc}
\epsilon_{J mn} \, \epsilon^{J n p} = v_J^\inv
\, \delta^p_n, \;\;\;\; \mbox{and} \;\;\;\;
\epsilon_{J mn} \, \epsilon^{J p n} = v_J^{-1} e^{-2\pi i J} 
\, \pi_J(\mu)^p_{\;\; n}.
\eeq
Composing the intertwiner $\epsilon_V: V\rightarrow V^*$ with the coevaluation yields  a bilinear form $\beta_V=d_V\circ (\epsilon_V\otimes\text{id}_V): V\times V\rightarrow \C$ on each representation space $V$, which is invariant 
$$
\beta_V( v, \pi_{V}(a) w )= \beta_V( \pi_{V}(S(a)) v,  w ) \quad\forall v,w\in V, a\in \uqr.
$$
For the simple objects $(\pi_J,V_J)$, its expression in  terms of the basis 
$(e^J_m)_{m=-J,...,J}$ reads  
\begin{align}\label{betadefuq}
\beta_J( e^J_m,e^J_n ) =\epsilon_{J\, mn}.
\end{align}

\subsubsection{Quantum $\SO(4)$.}
\label{qso4}

In this section, we define the Hopf algebra that we will consider as the quantum counterpart of the Euclidean isometry group $\SO(4)$. As in the previous section, we fix the parameter $q$ to an odd primitive of unity $q = e^{2i \pi/r}$.

The self-dual and anti-self-dual factorisation of the classical Lie algebra $\spin(4) \cong \su(2)_+ \oplus \su(2)_-$ suggests a two parameter deformation $U_{q,q'}(\spin(4)) \simeq U_q(\su(2))\otimes U_{q'}(\su(2))$, where $q$ is the deformation parameter for the first  copy of $\su(2)$ and $q'$ the one for the second copy. In the following, we will be interested in one parameter deformation which arises by specialising\footnote{Note that the construction differs from that of the quantum Barrett-Crane model \cite{yet4} where $q' = q^{-1}$.} to $q' = q$. In other words, we consider as the the quantum version of the Lie group $\Spin(4)$ the ribbon star Hopf algebra 
$\uqr \otimes \uqr$.
The multiplication, unit, coproduct, counit, antipode, $R$-matrix and ribbon element on $\uqr\otimes\uqr$ are inherited from the corresponding structures on $\uqr$
\begin{align}
&m = (m \otimes m) \circ \tau,  &  &\eta(\lambda) = \lambda 1 \otimes 1, \\
&\Delta =\tau \circ (\Delta \otimes \Delta),  & &\epsilon = \epsilon \otimes \epsilon, & &S = S\otimes S, \\
&R = \tau \circ R \otimes R,  &  &\mu = \mu \otimes \mu,
\end{align}
where the flip map $\tau$ acts on the second and third copy of the four-fold tensor product $\uqr^{\otimes 4}$, i.e., $\tau \equiv \tau_{23}$. Note that  the two  copies of $\uqr$ commute.

The finite-dimensional representations of $\uqr\otimes \uqr$ are obtained by tensoring pairs of finite-dimensional representations of $\uqr$. In the following, we will restrict attention  to those representations  which are obtained by tensoring two finite-dimensional irreducible tilting
modules of $\uqr$. We denote the associated modular category by $\mathcal C_E$.
The irreducible objects of $\mathcal{C}_E$ are labelled by  a couple of spins $\alpha=(J_+,J_-)$, $J_\pm\in\{0,...,(r-2)/2\}$ each corresponding to a finite-dimensional tilting module.
A basis of the $\uqr\otimes\uqr$ module labelled by $\alpha=(J_+,J_-)$ is given by $\{e^J_{m_+ m_-} \!\!= e^{J_+}_{m_+} \otimes e^{\alpha_-}_{m_-}\}_{m_\pm=-J_\pm,...,J_\pm}$. The action of  $\uqr\otimes\uqr$ on this basis takes the form 
\beq
 \pi_\alpha(a) \rhd e^{J_+}_{m_+} \otimes e^{J_-}_{m_-} = \pi_{J_+}(a_+) \, e^{J_+}_{m_+} \otimes \pi_{J_-}(a_-)e^{J_-}_{m_-}\qquad \forall a\in\uqr\otimes\uqr.
\eeq
By tensoring two copies of the bilinear form \eqref{betadefuq}, one obtains a bilinear form on $V_J=V_{J_+}\otimes V_{J_-}$ which is invariant under this action. 
Moreover, due to the fusion rules \eqref{fusionunity} of the irreducible tilting modules, there is a well defined action of $\uqr\otimes \uqr$ on each $U_q^{(r)}(\su(2))$-module $V_K$ in \eqref{fusionunityroot}, which is  given by
$$
\pi_\alpha(a) \rhd e^{K}_{m} = \!\!\!\sum_{m_+,m_-} \!\!\!\!\left( \begin{array}{cc} m_+ & m_- \\
                          J_+ & J_-  \end{array} \right| \left. \begin{array}{c} K \\ m \end{array} \right)  \pi_I(a) \rhd e^{J_+}_{m_+} \otimes e^{J_-}_{m_-}\quad \forall a=a_+\otimes a_-  \in \uqr\otimes\uqr. 
$$
This action is non-trivial  that provided $J_+,J_-$ and $K$ satisfy the constraint  \eqref{qCG}.

\subsection{The quantum EPRL intertwiner}
\label{euceprl}

We are now ready to construct the associated quantum EPRL-model. 
We start by defining the EPRL-representations. The construction is analogous  to the Lorentzian case. The only difference is that in order to ensure that these representations are associated with objects of $\cres$, we need to restrict the value of the Immirzi parameter $\gamma$ to the interval $0<\gamma<1$ and to impose that it is rational with $\gamma  (r-2)/2\in\mathbb N$. The case $\gamma>1$  will be treated elsewhere.

\begin{definition} {\bf (EPRL representations)}

Let $\gamma$ be a fixed real  parameter satisfying $0<\gamma < 1$ and $\gamma\cdot (r-2)/2\in\mathbb N$. Let  $\mathcal{L}= \{0,1/2,1,...,(r-2)/2\}\cap\{K\in \NN_0/2: \gamma K\in\NN_0\}$ be the label set  of simple objects in $\cres$ for which $\gamma K\in \NN_0$.
Given an element $K\in\mathcal{L}$, the associated Euclidean EPRL representation $\alpha(K)$ is the object $(\pi_{\alpha(K)}, V_{\alpha(K)})$ with $V_{\alpha(K)}=V_{J_+(K)}\otimes V_{J_-(K)}$, $\pi_{\alpha(K)}=\pi_{J_+(K)}\otimes  \pi_{J_-(K)}$   with
\begin{align}\label{eprluqr}
 \alpha(K)= (J_+(K),J_-(K)) :=(\tfrac 1 2 (1 + \gamma) K, \tfrac 1 2(1-\gamma) K).
\end{align}
\end{definition}

Note that for a given $K \in \mathcal{L}$ , the constraints $0<\gamma < 1$ and  $\gamma\cdot (r-2)/2\in\mathbb N$  on the Immirzi parameter ensure that
the Euclidean EPRL representation labelled by $\alpha(K)$  decomposes into irreducibles  \beq
\label{decompE}
V_{\alpha(K)} = \bigoplus_{J=\gamma K}^K V_J.
\eeq
We denote by $\{e^\alpha_{m+,m_-}\}_{m_\pm=-J_\pm,...,J_\pm}$ the tensor basis of $V_\alpha=V_{J_+}\otimes V_{J_-}$:
$
e^\alpha_{m_+,m_-}=e^{J_+}_{m_+}\otimes e^{J_-}_{m_-}
$,
and by $f_{\alpha}^k : V_{\alpha}=V_{J_+}\otimes V_{J_-} \rightarrow V_K$ the projection map associated to \eqref{decompE} that projects  onto the highest weight factor $V_K$ in the decomposition \eqref{decompE}.  As in the Lorentzian case, we denote by $f^\alpha_K$ the associated inclusion map 
$f^{\alpha}_K : V_K \rightarrow V_{\alpha}$ defined by the identity $f_{\alpha}^{K'}  \circ f^{\alpha}_K = \text{id}_K \delta_{K}^{K'}$. The projection  map $f_\alpha^K$ commutes  with the action of $\uqr$ and is given by the Clebsch-Gordan morphisms. In terms of the bases  introduced above, it reads
\beq
\label{projectionuqr}
f_\alpha^K(e^\alpha_{m_+ m_-}) = \sum_{n} \left( \begin{array}{c} n \\ K
 \end{array} \right| \left. \begin{array}{cc} J_+ & J_- \\ m_+ & m_- \end{array} \right)  e^K_{n}.
\eeq

To define the EPRL intertwiner, we need to specify the quantum counterpart of 
the expression \beq
\label{int2}
T_{V[\alpha]} = \int_G d X \left( \bigotimes_{i=1}^n \pi_{\alpha_i} \right)(X).
\eeq
In contrast to the Lorentzian situation, this will not be achieved by using a Haar measure on the dual Hopf algebra of $\uqr$. The reason is that, in contrast to the situation for generic $q$, there is no direct characterisation of the dual Hopf algebra of $\uqr$ in terms of matrix elements of irreducible $\uqr$ representations.  A simple dimension-counting argument shows that the the matrix elements in the representations \eqref{repsuqr} do not form a basis of the dual Hopf algebra $\uqr^*$. A Poincar\'e-Birkhoff-Witt basis for $\uqr$ is given by $\{K^sE^tF^u\,|\, s,t,u\in\{0,1,...,r-1\}\}$, which implies that the dimension of the vector space $\uqr$ is $r^3$. The matrix elements with respect to the representations \eqref{repsuqr} are linear functions on $\uqr$ and hence elements of $\uqr^*$. However,  the number of matrix elements obtained from the representations \eqref{repsuqr} is
$$
\sum_{J\in\{0,1/2,...,(r-2)/2\}} \!\!\!\!\!\!\!\!\!\!\!\!(2J+1)^2=\sum_{k=0}^{r-1} k^2=\frac {r(r-1)(2r-1)} 6< r^3\quad\forall r\geq 3.
$$
The matrix elements  with respect to the representations \eqref{repsuqr} therefore span only a  proper linear subspace of  $\uqr^*$ and do not form a  basis of the vector space  $\uqr^*$. 
Although it is immediate to obtain an expression that  has  the appropriate invariance properties 
by setting  $h(\pi_J) = \delta_{J,0}$, this is not sufficient to characterise a Haar measure on $\uqr^*$. 

However, due to the fact that  $\cres$ is a modular tensor category, there is an alternative generalisation of the classical
expression \eqref{int2} which makes use of the semisimplicity, for a pedagogical explanation see 
 \cite{Oeckl}.
The fact that $\cres$ is a modular tensor category implies that for each object $(\pi_V,V)$ there is  a unique endomorphism $T_V$
which can be expressed as
\begin{align}\label{tmap}
T_V=\sum_{b\in\mathcal B} g_b\circ f_b,
\end{align}
where $\mathcal B$ labels a basis of intertwiners in $\text{Hom}(V,\C)$,  $f_b: V\rightarrow\C$ runs over the basis elements of $\mathcal B$ and  $g_b=\epsilon^V\circ f_b^*:\C\rightarrow V$.
For simple objects $(\pi_I,V_I)$, the space of morphisms $\text{Hom}(V_I,\C)$ is trivial unless $I=0$. The associated endomorphism $T_{V_I}$ therefore takes the form $T_{V_I}=\delta_{I,0}$.
For a general object $(\pi_V,V)$ the associated endomorphism $T_V$ is obtained by expressing $V$ as a direct sum of simple objects $V=\bigoplus_{J} V_J$. In particular, this implies that for two simple objects $(\pi_I,V_I)$, $(\pi_J,V_J)$ the morphism $T_{V_I\otimes V_J} $ takes the form
$$
T_{V_I\otimes V_J}=\sum_{K} \text{dim}_q(k) C^K_{IJ} T_{V_K}  C^{IJ}_K= C^0_{IJ} C^{IJ}_0=\delta_{I,J} \,\,b_{V_I}\circ d_{V_I^*}.
$$

The corresponding expression for general tensor products of irreducibles is obtained by consecutively coupling them to zero via the Clebsch-Gordan coefficients. Consequently, the endomorphism $T_V$ can be viewed as the counterpart of the classical expression \eqref{int2} and expression  \eqref{classgenq} for the quantum Lorentz group.
It generalises 
 expression \eqref{int2} to the category $\cres$ and allows us to define the quantum EPRL intertwiner.

\begin{definition} {\bf (Quantum EPRL intertwiner)}\label{equeprl}
Let $K = (K_1, ..., K_n)\in\mathcal L^n$ be a $n$-tuple of irreducible representations in $\mathcal L$   and $V[K] = \bigotimes_{a=1}^n V_{K_a}$  the corresponding representation space. Denote by $\alpha=(\alpha(K_1),...,\alpha(K_n))$ the  associated EPRL-representations defined as in \eqref{eprluqr} and by $V[\alpha]$ the  tensor product of their representation spaces $$V[\alpha]=\bigotimes_{a=1}^n V_{\alpha(K_a)}=\bigotimes_{a=1}^n V_{J_+(K_a)}\otimes V_{J_-(K_a)}.$$
Then the the quantum EPRL 
intertwiner  associated to  an element  $\Lambda_K\in\mathrm{Hom}_{\uqr}(V[k], \C)$ is the morphism  $\iota_{\alpha,K}: V[\alpha] \rightarrow \C$  defined by 
\begin{align}
\label{quintertwineruq}
\iota_{\alpha,K}=\Lambda_K\circ f^K_\alpha\circ T_{V[\alpha]}
\end{align}
where  $f^{K}_{\alpha} : V[\alpha] \rightarrow V[K]$ is given in terms of the morphism \eqref{projectionuqr}  by 
$
f_{\alpha}^K = \bigotimes_{a=1}^n f_{\alpha(K_a)}^{K_a}.
$
\end{definition}

Note that, unlike the Euclidean quantum BC intertwiner \cite{yet4}, the Euclidean EPRL intertwiner is not invariant under braiding. This follows  directly from its definition in analogy to the Lorentzian case.

In the classical Euclidean EPRL model, it is possible to perform the Haar integrals in the definition of the EPRL intertwiner. This  leads to a factorisation of the amplitude for the $4$-simplexes in terms of $15j$-symbols. The next proposition gives the equivalent of this construction for the $q$-deformed Euclidean EPRL model.

\begin{proposition} 
\label{decomp}
Let $\alpha(K)$ be the representation label of  an EPRL representation. 
Then the associated  EPRL intertwiner  $\iota_\alpha$ can be expressed as
\beq
\iota_{\alpha(K)} = \sum_\delta c (K,\delta) \, \Phi_{\alpha, \delta}, \;\;\;\; \Phi_{\alpha, \delta} \in \Hom_{\uqr\otimes\uqr}(V[\alpha], \C),
\eeq
where
$
c(K,\delta)= \Lambda_K\circ f^K_\delta \circ \Phi^{\alpha, \delta} \in\C
$,
and  $\Phi^{\alpha, \delta} \in \Hom \left( \C , V[\alpha] \right)$ is given in the proof below.
\end{proposition}

{\em Proof.} 
We use the identity $T_{V_J}=\delta_{J0}$ for any irreducible $\uqr$ representation  $J\in\mathcal L$
together with the complete reducibility of  the tensor product of $\uqr\otimes\uqr$-modules \ros
\beq\label{repfac}
(\pi_\alpha\otimes\pi_\beta)(\Delta(a)) = \sum_\delta C^{\alpha\beta}_{\; \delta} \circ \pi_\delta(a) \circ  C^\delta_{\; \alpha\beta}\qquad \forall a \in \uqr\otimes\uqr,
\eeq
where $ C^{\delta}_{\; \alpha\beta}: V_\alpha \otimes V_\beta \rightarrow V_\delta$, and $C^{\alpha\beta}_{\;\, \delta} : V_\delta \rightarrow V_\alpha \otimes V_\beta$ are the Clebsch-Gordan maps for the representations
of  $\uqr\otimes\uqr$. They factorise into Clebsch-Gordan maps for the irreducible tilting modules of $\uqr$:  $C^\delta_{\; \alpha\beta} = (C^{K_+}_{\; I_+ J_+} \otimes C^{K_-}_{\; I_- J_-} )\circ \tau_{23}$
for $\alpha=(I_+,I_-), \beta=(J_+,J_-)$, $\delta=(K_+,K_-)$.
From identity \eqref{repfac}, we then obtain 
\beq
T_{V_\alpha\otimes V_\beta}=C^{\alpha\beta}_{\; 0} \, C^{0}_{\; \alpha\beta} =  \frac{\delta_{\alpha\beta}}{[2\alpha +1]} \epsilon^\alpha \, \epsilon_\alpha, 
\eeq
where  $\epsilon_\alpha : V_{I_+} \otimes V_{I_-} \rightarrow \C$ is defined as $\epsilon_\alpha =( \epsilon_{I_+} \otimes \epsilon_{I_-}) \circ \tau_{23}$ and $[2\alpha+1]=[2I_++1][2 I_-+1]$ . 
With this result, it is  immediate to compute $T_V$ for  four-valent tensor products $V[\alpha]=\bigotimes_{a=1}^4 V[\alpha_a]$,  which appear in the definition of the  EPRL intertwiner. The computation proceeds as in the classical case.  For $\alpha=(\alpha_1,...,\alpha_4)$, $T_{V[\alpha]}$ is given as a sum over the representation $\delta$ appearing in the intermediate channel of three-valent decomposition of a $4$-valent intertwiner
\begin{align}
\label{integraleuc}
&T_{V[\alpha]}= 
 \sum_\delta \Phi^{\alpha_1,...,\alpha_4,\delta} \, \Phi_{\alpha_1...\alpha_4, \delta},\\
&\Phi_{\alpha_1...\alpha_4, \delta} = \frac{1}{\sqrt{[2\delta +1]}}\; d_\delta \circ (C_{\alpha_1\alpha_2 \delta} \otimes  C^{\delta}_{\; \alpha_3\alpha_4}) \; \in \Hom\left( \otimes_{a=1}^4 V_{\alpha_a}, \C \right), \nn
\end{align}
where $d_{\alpha} : V^*_{\alpha} \otimes V_{\alpha} \rightarrow \C$ is the evaluation map and $\Phi^{\alpha_1,...,\alpha_4,\delta}$ is the morphism dual to $\Phi_{\alpha_1,...,\alpha_4,\delta}$.
The intertwiner  $\Phi_{\alpha , \delta}$ factorises into intertwiners  of $\uqr$-modules according to
$$
\Phi_{\alpha_1...\alpha_4, \delta} = \left(\Phi_{I_{1+} I_{2+} I_{3+} I_{4+}, I_+} \otimes (\Phi_{I_{1-} I_{2-} I_{3-} I_{4-}, I_-} \right)\circ \tau\qquad \delta=(I_+,I_-), \alpha_a=(I_{a+},I_{a-}),
$$
where $\tau$ is the appropriate flip map and
\beq
\Phi_{I_1\cdots I_4,I} = \frac{1}{\sqrt{[2I +1]}} d_{I} \circ\left( C_{I_1I_2 I} \otimes C^{I}_{\; I_3I_4}\right).
\eeq
The evaluation of $\Phi$ in our basis then yields 
$$
\Phi_{JKLM, I}(e^J_{a} \otimes e^K_b \otimes e^L_c \otimes e^M_{d}) = \left( \begin{array}{cccc} J & K & L & M \\
a & b & c & d \end{array} \right)_{I},
$$
with
\beq
 \left( \begin{array}{cccc} J & K & L & M \\
a & b & c & d \end{array} \right)_{I} = \left( \begin{array}{c} e \\ I
 \end{array} \right| \left. \begin{array}{cc} J & K \\ a & b \end{array} \right) \epsilon_{I ef} \left( \begin{array}{c} f \\ I
 \end{array} \right| \left. \begin{array}{cc} L & M \\ c & d \end{array} \right),
\eeq
in analogy with previous results obtained for the  Lorentzian model.  Using this expression, it is  then immediate to obtain the evaluation of  $\Phi_{\alpha_1...\alpha_4, \delta}$
and to express  the quantum EPRL intertwiner as 
\begin{align}
&\iota_\alpha = \Lambda_K\circ f^K_\alpha\circ T_{V[\alpha]} = \sum_{\delta=(I_+,I_-)} c(K,\delta) \Phi_{\alpha, \delta},\\
&c(K,(I_+,I_-)) = \Lambda_{K_1 ... K_4} \circ \left( \bigotimes_{a=1}^4 C^{K_a}_{K_{a+} K_{a_-}} \right) \circ \left( \Phi_{I_+}^{K_{1+} K_{2+}K_{3+}K_{4+}} \otimes \Phi_{I_-}^{K_{1-}K_{2-}K_{3-}K_{4-}} \right)\nonumber
\end{align}
\hfill $\square$

\medskip

To exhibit explicitly the structural similarities with the Lorentzian case, we construct the EPRL model for $\uqr\otimes\uqr$ by making use of EPRL tensors, which are elements of $\Hom_{\uqr \otimes \uqr}( \C , \bigotimes_{i=1}^n V_{\alpha_i})$. 
The main difference with respect to the Lorentzian model
is that in the Euclidean case,
 the vector space $\Hom_{\uqr \otimes \uqr}( \C , \bigotimes_{i=1}^n V_{\alpha_i})$ is not-trivial. Therefore, EPRL tensors are bona fide intertwiners and are defined simply as the  intertwiners dual  to the EPRL intertwiners in Def.~\ref{equeprl}.

\subsection{Four-simplex amplitude and the quantum spin foam model}
\label{eucamp}

We will now show how the amplitude for the four-simplexes can be defined for the Euclidean model. As the relevant representations form a modular category, there are no problems related to divergences and no regularisation is required. 
We start by defining the colouring and state space of a tetrahedron in analogy to the Lorentzian case. We again suppose that
$M$ is an oriented, closed  triangulated $4$-manifold with sets of simplexes $\Delta^{(n)}$. Consider a $4$-simplex $\sigma$ of
$M$. We parametrise the  set $\Delta^{(3)}_{\sigma}$ of tetrahedra of $\partial \sigma$ as in Section \ref{ampsec}.
 
\begin{definition} {\bf (Colouring)}
A colouring $\alpha: \Delta_{\sigma}^{(2)} \rightarrow \mathrm{Ob}(\cres)$ associates an EPRL representation 
$\alpha_{ab}$ of $\uqr\otimes\uqr$ to each oriented triangle $(ab)$ in $\Delta_{\sigma}^{(2)}$. 
\end{definition} 

The state spaces for the tetrahedra are constructed from the colourings  similarly to  the Lorentzian case. The only difference is that in the Euclidean case, there is no need to implement  the state space as a space of Hopf-algebra valued linear maps. Instead, we identify it with the space of  intertwiners between the trivial representation and the tensor product of the  EPRL representations at their boundary triangles.

\begin{definition} {\bf (State space)}
Let $\alpha: \Delta^{(2)}_{\sigma} \rightarrow  \mathrm{Ob}(\cres)$ denote a colouring. The state space associated with a 
tetrahedron $(a)$ that appears with positive sign in $\partial \sigma$ is read out of the colouring of its boundary $\partial(a) = (ab) - (ac) + (ad) - (ae)$ and is defined by
$$
H_{a} =  \Hom_{\uqr \otimes \uqr} \left( \C, V_{\alpha_{ab}} \otimes V_{\alpha_{ad}} \otimes V_{\alpha_{ac}} \otimes V_{\alpha_{ae}} \right) 
$$
Likewise, the state space for negatively oriented  tetrahedra is given by 
$$
H_{a}^* =   \Hom_{\uqr \otimes \uqr} \left( \C, V_{\alpha_{ae}} \otimes V_{\alpha_{ac}} \otimes V_{\alpha_{ad}} \otimes V_{\alpha_{ab}} \right)
$$
A state is an assignment of an EPRL tensor $\Psi_{a}$ in either $H_{a}$ or $H_a^*$ to each tetrahedron $(a)$ in 
$\Delta_{\sigma}^{(3)}$. 
\end{definition}

We are now ready to define the amplitude, or partition function, for the $4$-simplex $\sigma$.  This is done by the same diagrammatic method as in the Lorentzian case. 

\begin{definition} {\bf ($4$-simplex amplitude)}
\label{amplitudedef}
Let $\alpha: \Delta_{\sigma}^{(2)} \rightarrow \mathrm{Ob}(\mathcal{C})$ denote a colouring and $\Psi$ a state. 
The amplitude  for a $4$-simplex $\sigma$   is the linear map
\begin{align}
&\mathcal{Z}_{\sigma} : H_{a}^* \otimes H_{b} \otimes H_{c}^* \otimes H_{d} \otimes H_{e}^* \rightarrow \C,\nonumber 
\end{align}
obtained by evaluating the spin network diagram in Figure \ref{amplitude} corresponding to the complete graph with five vertices $\Gamma$, that is, $\mathcal{Z}_{\sigma}(\Psi_a \otimes ... \otimes \Psi_e) = ev(\Gamma)$.
\end{definition}

The only difference to the Lorentzian case is the evaluation of the diagram. In the Euclidean case, there is no need to regularise the evaluation and $ev(\Gamma)$ is defined by the quantity that corresponds to $\phi(\Gamma)$  in the Lorentzian case. The diagrammatic calculus is defined in the same way as in the Lorentzian section. However,  the vertices of the diagram are now labelled by Euclidean quantum EPRL tensors, the  lines are connected  via  the invariant bilinear form
$$
\beta_{\alpha}(v,w) = \epsilon_{\alpha}(w)(v), \;\;\;\; \mbox{with} \;\;\;\; \epsilon_{\alpha} = (\epsilon_{I_+} \otimes \epsilon_{I_-} )\circ \tau_{23}, \alpha=(I_+,I_-)\qquad \forall v,w \in V_{\alpha},
$$
and the crossing diagram corresponds to the braiding $c_{\alpha, \beta}$ obtained by tensoring two copies of the  braiding \eqref{braidingres} of the category $\mathcal{C}_{res}$. Using these
structures associated with the diagram components, one can then compute the quantity $ev(\Gamma)$ with $\Gamma$ depicted in Figure \ref{amplitude}. The computation is analogous to the Lorentzian case, but all series arising there are replaced by finite sums. The amplitude therefore converges 
without further regularisation.

Using the factorisation of the quantum EPRL intertwiner in Def.~\ref{decomp}, one can express the amplitude for the four-simplexes in terms of quantum $15j$ symbols\footnote{Note that the notion of  $15j$ symbol used here is closely related to but not equivalent to the standard definition. The standard  $\uqr$ $15j$ symbol is obtained by composing five Clebsch-Gordan morphisms for $\uqr$ and closing the resulting expression with a quantum trace. As it is our aim to explicitly exhibit the similarities with the Lorentzian model, we work with a different expression, which is the analogue of the quantities called $15j$ symbols  there and in the classical EPRL models.} for $\uqr$.

\begin{definition} {\bf (Quantum $15j$ symbol)}
Let $I : \Delta_{\sigma}^{(2)} \rightarrow \mathrm{Ob}(\mathcal{C}_{res})$ be a $\uqr$ colouring and $\Lambda : \Delta_{\sigma}^{(3)} \rightarrow \Hom_{\uqr}(\mathcal{C}_{res})$ a state constructed from the representation theory of $\uqr$. The quantum $15j$ symbol is defined as
$
(15j)_q(I;\Lambda) = \mathcal{Z}_{\sigma}(\Lambda_a \otimes ... \otimes \Lambda_e). 
$
The quantum $15j$ symbol for $\uqr\otimes\uqr$ is the product  of the  $15j$-symbols for the two copies of $\uqr$
$$
(15j)_q(\alpha;\Lambda) = (15j)_q(I_+,\Lambda_+) \, (15j)_q(I_-,\Lambda_-)\quad\text{where}\;\alpha=(I_+,I_-),\; \Lambda = (\Lambda_+ \otimes \Lambda_- )\circ \tau.
$$
\end{definition}

Using this definition, it is immediate to prove the following proposition.

\begin{proposition}
The amplitude for the $4$-simplex $\sigma$ can be expressed in terms of a quantum $15j$ symbol for $\uqr$ as follows \ros
\beq
\mathcal{Z}_{\sigma}(\Psi_a \otimes ... \otimes \Psi_e) = \prod_{a=0}^4 \sum_{\delta_a} c_a(K_{a},\delta_a) \, (15j)_q(\alpha_a,\Phi_{\alpha_a,\delta_a}), \nn
\eeq
where we  use the shorthand notations $K_a=(K_{ab})_{b \neq a}$, $\alpha_a=(\alpha_{ab})_{b \neq a}$, and the intertwiners $\Phi_a$ and the constants $c_a$ are defined in proposition \ref{decomp}.
\end{proposition}

This  proposition shows that the amplitude for the $4$-simplexes factorises in a way that is directly  analogous to the classical case. Note that  such factorisation would also occur in the Lorentzian model if one did not remove one of the Haar measures at the five vertices. In physical terms, the removal of one Haar measure amounts to a gauge fixing, which ensures the convergence of the model but breaks its symmetry. Without this gauge fixing, one would obtain a divergent  expression for the amplitude that factorises {\em formally} into $15j$ symbols.

Using the definitions introduced in the previous section, we can now define the quantum EPRL spin foam model. 

\begin{definition}
Let $M$ be a closed oriented triangulated $4$-manifold with sets of $n$-simplexes $\Delta^{(n)}$. Choose a fixed linear order of the vertices of $M$. Let $\sigma \in \Delta^{(4)}$ be a $4$ simplex of $M$ and let $\Delta^{(2)}_{\sigma}$ and $\Delta^{(3)}_{\sigma}$ denote its set of triangles and tetrahedra respectively. Let $\alpha: \Delta^{(2)}_{\sigma} \rightarrow \mathrm{Ob}(\mathcal{C})$ denote the corresponding EPRL-colouring, $\Psi : \Delta^{(3)}_{\sigma} \rightarrow \Hom(\mathcal{C})$ the associated EPRL-state, and let the 4-simplex amplitudes $\mathcal{Z}_{\sigma}$ be given by Definition \ref{amplitudedef}. The partition function for the quantum EPRL model associated to $M$ is given by 
\beq
\mathcal{Z}(M) = \sum_{K,J} \prod_{t} [2K_{t} + 1] \prod_{\sigma} \mathcal{Z}_{\sigma} \left(\alpha(K_{t(\sigma)}),\Psi_{\sigma}(J_{tet(\sigma)}) \right).
\eeq
Here, the sum ranges over all $K$ in $\mathcal{L}$ and over the elements of a basis of $U_q(\su(2))$-intertwiners $(\Lambda_{K,J})_J$ entering the definition of the EPRL tensors for each tetrahedron of $M$. The state $\Psi_{\sigma}$ for the $4$-simplex $\sigma$ is given by $\Psi_{\sigma} = \bigotimes_{tet(\sigma) \in \partial \sigma} \Psi_{tet(\sigma)}$, where $\Psi_{tet(\sigma)}$ is the state associated to the tetrahedron $tet(\sigma)$ and the order in the tensor product is the one given in Definition \ref{amplitudedef}. Finally, the product is over all the triangles $t$ and $4$-simplexes $\sigma$ of $M$.
\end{definition}
As an immediate consequence of the definition of the model we obtain that the partition function $\mathcal Z (M)$ converges. This follows directly from the fact that the label set $\mathcal L$ contains only a finite number of representations. Consequently, the expression for $\mathcal Z (M)$ involves only a finite sums and hence does not exhibit  any of the convergence issues associated with the classical models. 

\section{Discussion and conclusions}

\subsubsection*{Summary}

In this paper, we construct the $q$-deformed spin foam models  corresponding to the Lorentzian and Euclidean EPRL models. In the first part of the paper, we concentrate on the Lorentzian model. We review the construction of the quantum Lorentz group as the quantum double of $U_q(\su(2))$ together with its harmonic analysis given in \cite{PE,PE2}. Using these results, we generalise the Lorentzian EPRL intertwiner to the quantum Lorentz group by means of a Haar measure. We prove that the quantum EPRL intertwiner converges and show that, in contrast to the intertwiners arising in the BC model \cite{BC,BC2},  it is not invariant under braiding. We then provide a consistent definition of an amplitude for the four-simplexes using a graphical calculus for the braided category of representations of the quantum Lorentz group. The resulting spin foam model for a compact triangulated four-manifold $M$ converges.

In the Euclidean case, we  employ different methods to generalise the Euclidean EPRL intertwiner. Using the language of categories, we provide a generalisation of the classical intertwiner to the modular category of representations of $\uqr$ at root of unity, obtained via the tilting module construction. The resulting intertwiner converges trivially and is also not invariant under braiding. We then define an amplitude for the four-simplexes analogously to the Lorentzian model by using the appropriate graphical calculus. As a result, we obtain that the amplitude can be expressed in terms of quantum $15j$ symbols and fusion coefficients. As it only involves only sums over a finite set of representation labels, the resulting spin foam model converges.

\subsubsection*{Physical interpretation}

Due to their  convergence properties, $q$-deformed spin foam models can be viewed as infra-red regularisations  of the corresponding classical models. The divergences in the classical models  associated with large values of the representation labels disappear as there is a cut-off which eliminates these representations.  This cut-off on the representation labels can be interpreted as the presence of a horizon and appears to be related to the cosmological constant.

The correspondence between $q$-deformations and the cosmological constant is well-known
and understood in detail in the three-dimensional case.  In three dimensions, the Turaev-Viro invariant is related to Euclidean gravity with a positive cosmological constant. This relation can be inferred  from  two perspectives. The first one is the investigation of the semi-classical limit of the model via the asymptotics of the amplitude for the three-simplexes (given by a quantum $6j$ symbol). This limit involves  the physical regime where the spins $j$ are bounded from below by the Planck length $l_p$ and from above by the cosmological length $l_c=1/\sqrt{\Lambda}$ that corresponds to the de-Sitter horizon. The asymptotics of the quantum $6j$ symbol has been computed in \cite{M} and 
yields the cosine of the Regge action with a cosmological term. 

An alternative way to relate the Turaev-Viro model to three-dimensional gravity with positive cosmological constant is obtained from the correspondence between three-dimensional gravity, BF theory and Chern-Simons gauge theory. Three-dimensional gravity with Euclidean signature and positive cosmological constant can either be expressed as a $\SU(2)$ BF theory with cosmological constant $\Lambda>0$ or in terms of two   $\SU(2)$-Chern-Simons actions  \cite{Witten1} with level $k={\pi}/{\sqrt{\Lambda}}$. As a result, the partition function of a three-dimensional BF theory with cosmological term defined on a three-manifold $M$ is equal to the modulus square of the partition function of Chern-Simons theory on $M$ : $\mathcal{Z}_{BF, \Lambda}(M) = | I_{CS, k}(M) |^2$. On the other hand, it  was shown by Witten \cite{Witten2} that the  amplitude $I_{CS, k}(M)$ for a $\SU(2)$ Chern-Simons theory on $M$ is equal to the Reshetikhin-Turaev \cite{RTuraev} invariant for $M$ and $\uqr$. 
The fact that the Turaev-Viro sum of a manifold $M$ is equal to the modulus square of the Reshetikhin-Turaev invariant of $M$, relates the Turaev-Viro state sum to BF theory with a positive cosmological constant: \beq\mathcal{Z}_{TV}(q,M) = \mathcal{Z}_{BF, \Lambda}(M).\eeq By reintroducing natural units, one then finds that the deformation parameter $q$ of the Turaev-Viro model is related to the cosmological constant through the relation $q=\mbox{exp}( \mathrm i l_p / l_c)$.

In four dimensions, the situation is less clear, but there are indications for a similar relation between $q$-deformation and the cosmological constant in the context of the Crane Yetter invariant.  The relation between cosmological constant and $q$-deformation in four dimensions was first suggested in \cite{Lee, Lee2}. 
The Crane-Yetter invariant appears to be related to a four-dimensional $BF$ theory with cosmological constant. The first step in establishing this relation, shown for instance in \cite{Roberts}, is the proportionality relation between the Crane-Yetter state sum $\mathcal{Z}_{CY}(q,M)$ on a manifold $M$ with boundary $\partial M$ and the Chern-Simons amplitude $I_{CS,k}(\partial M)$ with level $k$ defined on the boundary three-manifold  $\partial M$: $Z_{CY}(q,M) \propto I_{CS,k}(\partial M)$. 
As in the three-dimensional case, the deformation parameter $q$ in the Crane Yetter state sum is related to the Chern-Simons level $k$ via $k={\pi}/{\sqrt{\Lambda}}$. 

If one considers  the path integral of $\SU(2)$ BF theory with cosmological constant $\Lambda > 0$ on a four-manifold $M$, one can integrate out the degrees of freedom associated with the $B$ field, which yields  the second Chern form. This form is in turn the differential of the Chern-Simons three-form and the path integral reduces to the path integral for Chern-Simons theory on the boundary $\partial M$ of $M$. Together with the relation between Crane-Yetter and Chern-Simons amplitudes,  this result gives a proportionality  relation between the Crane Yetter invariant and the state sum for $BF$ theory with positive cosmological constant   \beq \label{propo} \mathcal{Z}_{CY}(q,M) \propto \mathcal{Z}_{BF,\
 \Lambda}(M), \eeq where $q$ is given by $q=\mbox{exp} (\mathrm i l_p^2 / l_c^2)$ with  $l_c =1/\sqrt{\Lambda}$.

The models considered in this paper are $q$-deformed versions of the EPRL spin foam models. As the latter  are models for quantum gravity with vanishing cosmological constant, 
 it is plausible that  the  $q$-deformed models as models should
 describe some aspects of four-dimensional quantum gravity in Lorentzian or Euclidean de-Sitter space. Their deformation parameters should then be related to the cosmological constant via
 $$
q = \exp (- l_P^2 / l_c^2) \;\;\;\; \mbox{and} \;\;\;\; q = \exp(\mathrm  i l_P^2 / l_c^2),
$$
for, respectively,  the Lorentzian and Euclidean model. Interestingly, this relation leads to a bound on the area spectrum. The area spectrum  for a triangle $t$ of $M$ that is coloured by a representation  $K$ is given by
$$
A(t) = 8 \pi G \hbar \gamma \sqrt{K(K+1)}.
$$
With the relation between  the deformation parameter $q$ and the cosmological constant given above, one obtains a bound on this spectrum in terms of the cosmological length $l_c$ in the regime where $l_P<<l_c$:
\beq
A(t) \leq 32 \pi^2 l_c^2, \;\;\;\; \mbox{and} \;\;\;\; A(t) \leq \gamma 8 \pi^2 l_c^2,
\eeq
for the Lorentzian and Euclidean models, respectively. A detailed study of the asymptotic properties of the models should shed light on this question.

\subsubsection*{Open questions}

It would be interesting to explore in more depth the physical interpretation of the deformation parameter $q$ in the two models. In particular, this includes the question if  the $q$-deformation  can be related rigourously  to the cosmological constant $\Lambda$  along the lines described above.  For this one would need a notion of semi-classical limit and hence a detailed analysis of the asymptotics of the four-simplex amplitude. For the classical EPRL models this has been achieved in \cite{not1,notts3,notts4}  and \cite{fc1,fc2}. Generalising this analysis to the two $q$-deformed models constructed in this paper would require a suitable notion of coherent states, a generalisation of coherent intertwiners \cite{LS} and a careful analysis of their properties. We expect that such as an asymptotic analysis would give rise to the Regge action augmented with a volume term. 

\section*{Acknowledgements}
The research of the authors is funded by the German Research Foundation (DFG)
through the Emmy Noether fellowship ME 3425/1-1.  While completing this paper, both authors were also members
of the Collaborative Research Center 676 ``Particles, Strings and the Early
Universe''.  We are grateful to Philippe Roche and Karim Noui for helpful discussions.

\subsection*{Note added in proof.}
During the proofreading of this paper, we became aware that related results were derived independently by Muxin Han \cite{muxim}.

\end{document}